\documentclass[11pt,a4paper]{article}

\usepackage[comma]{natbib}
\usepackage{anysize}
\usepackage[usenames, dvipsnames]{color}
\usepackage{graphicx}
\usepackage[margin=3cm]{geometry}
\usepackage{amsmath,amssymb,amsthm,amsopn,mathrsfs,enumerate,enumitem}
\usepackage{ulem}
\usepackage{xcolor,color}
\usepackage[utf8]{inputenc}

\newcommand{\argmin}[1]{\underset{#1}{\mathrm{argmin}} \ }

\newcommand{\D}{\mathsf{D}}
\newcommand{\bH}{{\bf H}}
\newcommand{\bG}{{\bf G}}
\newcommand{\bI}{{\bf I}}
\newcommand{\ba}{{\boldsymbol a}}
\newcommand{\bb}{{\boldsymbol b}}
\newcommand{\be}{{\boldsymbol e}}
\newcommand{\bt}{{\boldsymbol t}}
\newcommand{\bu}{{\boldsymbol u}}
\newcommand{\bw}{{\boldsymbol w}}
\newcommand{\bx}{{\boldsymbol x}}
\newcommand{\by}{{\boldsymbol y}}

\newcommand{\bX}{{\boldsymbol X}}
\newcommand{\bY}{{\boldsymbol Y}}

\newcommand{\bpsi}{{\boldsymbol \psi}}

\newtheorem{lemma}{Lemma}
\newtheorem{theorem}{Theorem}

\theoremstyle{remark}

\renewcommand{\vec}{\operatorname{vec}}

\DeclareMathOperator{\E}{\boldsymbol{\mathbb{E}}}
\DeclareMathOperator{\Var}{Var}

\DeclareMathOperator{\Bias}{Bias} 
\DeclareMathOperator{\tr}{tr}
\DeclareMathOperator{\MSE}{MSE}

\DeclareMathOperator\MISE{\mathrm{MISE}}

\DeclareMathOperator\ISB{\mathrm{ISB}}
\DeclareMathOperator\IV{\mathrm{IV}}

\DeclareMathOperator\Diag{\mathrm{Diag}}
\DeclareMathOperator{\bexp}{\mathbf{exp}}

\renewcommand{\today}{\begingroup
\number \day\space  \ifcase \month \or January\or February\or
March\or April\or May\or June\or July\or August\or September\or
October\or November\or December\fi \space  \number \year \endgroup}

\title{Exploratory data analysis 
for moderate extreme values using non-parametric kernel methods}

\author{B. B\'eranger\thanks{Corresponding author. Email: b.beranger@unsw.edu.au}
\thanks{Theoretical and Applied Statistics Laboratory (LSTA)
University Pierre and Marie Curie - Paris 6, F-75005, Paris, France}
\thanks{School of Mathematics and Statistics, 
University of New South Wales, Sydney, Australia} , 
T. Duong\footnotemark[2] \thanks{Current address: Computer Science Laboratory (LIPN)
University Paris-Nord - Paris 13, F-93430, Villetaneuse, France} ,  
S. E.  Perkins-Kirkpatrick\thanks{Climate Change Research Centre, University of New South Wales, Sydney, Australia} , 
and S. A. Sisson\footnotemark[3]
}

\begin{document}

\maketitle

\begin{abstract}
\noindent  
In many settings it is critical to accurately model the extreme tail behaviour of a random process.
Non-parametric density estimation methods are commonly implemented as exploratory data analysis techniques for this purpose 
and can avoid the model specification biases implied by using parametric estimators.
In particular, kernel-based estimators place minimal assumptions on the data, and provide improved visualisation over scatterplots and histograms. However kernel density estimators can perform poorly when estimating extreme tail behaviour, which is important when interest is in process behaviour above some large threshold, 
and they can over-emphasise bumps in the density for heavy tailed data.
In this article we develop a transformation kernel density estimator which is able to handle heavy tailed and bounded data, as well as being robust in terms of the choice of
the extreme value threshold. We derive closed form expressions for its asymptotic bias and variance,
which demonstrate its good performance in the tail region.
Finite sample performance is illustrated in numerical studies, and in an expanded analysis of the ability of well known global climate models to reproduce observed temperature extremes in Sydney, Australia.


\end{abstract}

\section{Introduction}

The extreme values (very large or very small values) of a dataset
are frequently of interest as they are closely related to uncommon events 
with important consequences.
For climate data, these extreme events include heat waves (prolonged extreme high temperatures),
cold snaps (extreme low temperatures), floods (extreme high levels of waterways or tides
or waves), storms (extreme high wind speeds or amounts of precipitation) and droughts (prolonged extreme low
amounts of rainfall)
\cite[e.g.][]{kotz2000,coles2001}. 

Suppose that
$\bX=(X_1,\ldots,X_d)^\top$ is a $d$-dimensional random vector with cumulative distribution function (c.d.f.) $F_\bX$ and
probability density function (p.d.f.) $f_\bX$. 
A common representation of the extremes arising from this distribution focuses on examining those values which exceed some high threshold $\bu=(u_1,\ldots,u_d)^\top$, which determines the support on which values of $\bX$ are considered to be extreme. We can denote these extreme values by
$\bX^{[\bu]}\equiv \bX |\bX>\bu$, under which each marginal inequality must hold i.e. $X_j>u_j$ for $j=1,\ldots, d$.
Estimating the tail behaviour of $\bX$ is one of the goals of extreme value theory.

In the simplest case of univariate extremes ($d=1$),
common approaches to distributional tail estimation rely on parametric models, typically based on generalised 
extreme value (GEV) or generalised Pareto distributions (GPD), for which numerous estimation procedures are available.
These methods include maximum likelihood \citep{prescott1980, hosking1985b, smith1985,macleod1989}, 
probability weighted moments \citep{hosking1985},
maximum product spacing \citep{cheng1993},
least squares estimation \citep{maritz1967}, estimation based on order statistics and records \citep{pickands1975,hill1975}, the method of moments \citep{christopeit1994} and Bayesian estimation \citep{lye1993}.
As with all parametric estimators, these approaches suffer from the possibility of misspecification, particularly when the asymptotic GEV and GPD models cannot be assumed to hold.
This potential can be avoided by non-parametric estimation which does not make assumptions on particular parametric forms.
See \citet[Chapter 3]{markovich2007} for a summary of non-parametric estimation of univariate heavy tailed densities.

For multivariate extremes ($d>1$), non-parametric estimation of indicators of extremal dependence is an intensively studied field, and includes estimation of the Pickands or extremal dependence function \citep{pickands1975, hall2000, marcon2014}, the tail dependence function \citep{huang1992, drees1998, einmahl2008,einmahl2012}, and the spectral measure \citep{einmahl2001, einmahl2009, decarvalho2013}. 
The motivation for non-parametric estimators is stronger for multivariate extremes than in the univariate case as there is no general parametric form to describe the range of extremal behaviour of max-stable processes. This means that a choice of any particular parametric family has a possibility of resulting in model misspecification.

In this article we focus on moderate extreme values which consist of values that exceed 
some upper quantile of the data for which asymptotic models may not hold.
Hence our motivation differs from that in the traditional extreme value theory literature in the sense that the our main primary interest here is inference at sub-asymptotic levels i.e. around and slightly beyond the range of the observed data.

This manuscript proposes a new class of multivariate kernel-based nonparametric density estimators to model the extreme tail behaviour of $\bX$, without having to pre-specify a parametric family.
Kernel density estimators are among the most widely used 
non-parametric estimators \cite[e.g.][]{silverman1986,wand1995} 
and they possess
excellent visualisation properties which can naturally form part of an exploratory data analysis.
However, standard kernel estimators can produce spurious bumps in the estimated tails of $f_\bX$ 
if it has heavy tails,  as is particularly the case when analysing extremes and moderate extremes. 
On the other hand,  if we focus on the tail sample $\bX^{[\bu]}$, which is truncated at the threshold $\bu$, a standard estimator of the tail density $f_{\bX^{[\bu]}}$ is strongly influenced by the boundary effects due to the truncated support, as well as the choice of this threshold.

Here we focus on modifications to standard kernel density estimation which 
attenuate these spurious bumps, and accommodates a truncated support  
when estimating  $f_{\bX^{[\bu]}}$. 
Our approach is based on the standard identity $f_{\bX^{[\bu]}}(\bx) = f_\bX (\bx) / \bar{F}_\bX (\bu)$ where $\bar{F}_\bX(\bu)=1-F_\bX (\bu)$ is the corresponding survival function. We estimate the complete density $f_\bX$ using transformation kernel density estimation techniques 
\citep[see e.g.][]{silverman1986,charpentier2015}: this approach has the double advantage
of being able to handle bounded supported data as well
as reducing the spurious bumps in the tail. 
Estimating the normalisation constant $\bar{F}_\bX (\bu)$ is straightforward once
an estimator of $f_\bX$ is established.  

While multivariate transformation estimators are
a well-known method for complete densities, 
our contribution consists of modifying them for the
estimation of tail densities. We also supply new results for the pointwise 
bias and variance which describe the behaviour of the estimator at the boundary and in the tails.
In the context of moderate extremes, this permits the construction and theoretical justification of more efficient kernel based nonparametric density estimators for the tails of observed processes. It additionally allows these estimators to be used within existing goodness-of-fit measures \cite[e.g.][]{perkins2013} in place of more poorly performing histogram estimates of tail behaviour.

One may argue that in the univariate setup the Pickands-Balkema-de Haan theorem \citep{pickands1975,balkema1974} can be used. This states that, under some mild conditions on the underlying c.d.f., all observations above some large threshold $u$ are well approximated by the GPD. We will show that our proposed kernel estimator produces a comparable fit to the one produced by the GPD on the exceedances and can  even occasionally outperform it.

The layout of this article is as follows.
Our primary contribution is presented in 
Section \ref{sec:models}, which develops the transformation kernel estimator for tail density estimation, establishes its pointwise bias and variance (with proofs deferred to the Appendix) and examines optimal bandwidth estimation. We also assess histogram based tail density estimation, and develop the role of tail density estimators in goodness-of-fit (model selection) procedures.
In Section \ref{sec:numerical} we verify our results on finite samples 
for simulated data in both univariate and multivariate settings, and in
Section \ref{sec:GCM} we expand the work of \cite{perkins2013} by performing an analysis of 22 global climate models (GCMs) and assess how well they are able to reproduce observed temperature extremes.
Section \ref{sec:discussion} concludes with a discussion.

\section{Tail densities for moderate extreme values}
\label{sec:models}

\subsection{Transformation tail density estimation}

\label{sec:tail}

Let $X_{1}, \dots, X_n$ 
be a random sample drawn from a common univariate distribution $F_X$ with density $f_X$. 
If $f_X$ has heavy tails, standard kernel estimators are susceptible to producing spurious bumps in the tails of the density estimates, as they apply a fixed amount of smoothing over the entire sample space.  
A common approach is to apply a transformation  on the data sample to reduce the
inter-point distances in these moderate extreme values so that a global constant smoothing
is more appropriate. 
We focus on transformation kernel estimators, where a  
known monotonic transformation $t(\cdot)$ maps the data support to the real line where standard kernel estimators
are well-established, before back-transforming to the original data support.
See e.g. \citet{silverman1986,charpentier2015}.

Let $Y= t(X)$ be a transformed random variable, with distribution $F_Y$ and density $f_Y$.
The relationship between the transformed random variable $Y$ 
and the original $X$ at a non-random point $x$ is given by
\begin{equation*}
f_{X}(x) = |t'(x)| f_Y(t(x))
\end{equation*} 
where $t'$ is the first derivative of $t$. 
Consider the transformed sample $Y_1, \dots, Y_n$ where 
$Y_i = t ( X_i ), i= 1, \dots, n$, and $y=t(x)$. 
Since many (moderate) extreme value data samples are also bounded, 
e.g. $X_1, \dots X_n$ are supported on $(u_0, \infty)$, a suitable transformation would be
$t(x) = \log(x-u_0)$. In the case for unbounded data, the logarithm transformation
can still be used if we set $u_0 < \min\{X_1, \dots, X_n\}$.    
As $Y_1,\dots, Y_n$ are supported on the real line,  
$f_Y$ can then be estimated 
by the standard kernel density estimator
\begin{equation*}
\hat{f}_{Y} (y;h)
= n^{-1} \sum_{i=1}^n K_h ( y - Y_i )
\end{equation*} 
where $K_h (y) = h^{-1} K(y/h)$ is a scaled kernel, $h>0$ is 
the bandwidth or smoothing parameter and $K$ is a symmetric kernel density function. 
The estimator for $f_X$ can then 
be defined by replacing the true density by its  kernel estimator
\begin{equation*}
\label{eq:KPDFEuni}
\hat{f}_{X} (x;h) = |t'(t^{-1}(y))| \hat{f}_Y(y;h).
\end{equation*} 
Using the  standard identity $f_{X^{[u]}}(x) = f_X (x) / \bar{F}_X (u)$,  our proposed
estimator of the tail density is
$$
\hat{f}_{X^{[u]}}(x;h) = \hat{f}_X (x;h) / \Hat{\bar{F}}_X (u;h)
$$
where $\Hat{\bar{F}}_X (u;h) = \int_{-\infty}^u  \hat{f}_X (x;h) \, \mathrm{d}x$
can be numerically approximated.

A generalisation of this transformation kernel estimator to multivariate data is established through a  $d$-dimensional random vector  $\bX=(X_1, \dots, X_d)^\top$ with distribution function $F_{\bX}$ and  density  function $f_{\bX}$.
The random variable of values greater than a vector threshold $\bu=(u_1, \dots, u_d)^\top$ is denoted
as $\bX^{[u]} \equiv \bX \lvert \bX > \bu$ 
under which each marginal inequality must hold, i.e. $ X_j > u_j$ for $j=1,\dots, d$. The support of $\bX^{[\bu]}$ is the Cartesian product 
$(\bu,\infty) = ( u_1,\infty ) \times \cdots \times ( u_d,\infty )$. For $\bx \in (\bu,\infty)$,
the corresponding tail density is $f_{\bX^{[\bu]}}(\bx) = f_\bX (\bx) / \bar{F}_\bX (\bu)$
and tail distribution is $F_{\bX^{[\bu]}}(\bx) = F_\bX (\bx) / \bar{F}_\bX (\bu)$, 
where $\bar{F}_{\bX}(\bu) = \int_{(\bu, \infty)} f_\bX(\bw) d\bw$ is the survival function
of $\bX$ evaluated at $\bu$. 

Let $\bX_{1}, \dots, \bX_n$ 
form a random sample drawn from the common $d$-variate distribution $F_\bX$. 
Consider the transformed random variable $\bY = \bt (\bX)$ where $\bt:( \bu_0,\infty ) \rightarrow 
\mathbb{R}^d$ is defined by $\bt ( \bx ) = ( t_1 (x_1), \dots, t_d (x_d) )^\top$ where the $t_j$
are monotonic functions on $(u_{0j}, \infty)$ e.g. $t_j(x_j) = \log (x_j - u_{0j}), j=1,\dots,d$.
The density of $\bX$ is then related to the density
of $\bY$ by 
\begin{align*}
f_{\bX^{[\bu]}} ( \bx ) = f_{\bY} ( \bt(\bx) ) |\mathbf{J}_{\bt} ( \bx )| 
\end{align*}
where $|\mathbf{J}_{\bt}|$ is the Jacobian of $\bt$. 
Denoting the transformed data sample as
$\bY_{1}, \dots, \bY_{n}$, with $\bY_i = \bt(\bX_i),
i = 1,\dots, n$,
the kernel estimator of $f_\bY$ at a non-random point $\by = (y_1, \dots, y_d)^\top = \bt(\bx)$ is then
given by   
\begin{align*} 
\hat{f}_{\bY} (\by; \bH) = n^{-1} \sum_{i=1}^n K_{\bH} ( \by - \bY_i )
\end{align*}
where $K$ is a symmetric $d$-variate density function, the bandwidth matrix $\bH$ is
a $d \times d$ positive definite symmetric matrix of smoothing parameters, 
and the scaled kernel $K_\bH (\by) = |\bH^{-1/2}| K(\bH^{-1/2} \by)$. 
The tail density can then
be defined by replacing the true density function 
by its kernel estimator 
$$
\hat{f}_{\bX} (\bx;\bH)
= |\mathbf{J}_\bt(\bt^{-1}(\by))| \hat{f}_\bY(\by;\bH) 
$$
where $\bt^{-1} (\by) = (t_1^{-1} (y_1), \dots, t_d^{-1}(y_d))^\top$ is the element-wise
inverse of $\bt(\by)$. Therefore
\begin{equation}
\label{eq:KPDFEmulti}
\hat{f}_{\bX^{[\bu]}} (\bx;\bH)
= \hat{f}_{\bX} (\bx;\bH)/ \Hat{\bar{F}}_\bX(\bu;\bH)
\end{equation} 
where  $\Hat{\bar{F}}_\bX(\bu;\bH) = \int_{(-\infty, \bu)}  \hat{f}_\bX (\bx; \bH) \, \mathrm{d}\bx$ can be numerically approximated, for example by a Riemann sum.

In this approach, the threshold $\bu$ is only required 
to be specified in Equation~\eqref{eq:KPDFEmulti}. The statistical properties of $\hat{f}_{\bX^{[\bu]}}$ are
almost completely determined by those of $\hat{f}_\bX$ which do not rely on the 
choice of the threshold $\bu$. This is in contrast to an estimator 
of $f_{\bX^{[\bu]}}$ based on only the truncated sample $\lbrace \bX_i : \bX_i > \bu \rbrace$, as this is highly dependent
on the choice (and the estimation) of the threshold. Conveniently, for our proposed estimator, it is possible
to efficiently explore the tail behaviour for several thresholds, as the most onerous calculations are carried out to compute $\hat{f}_\bX$, and need not be repeated for
each threshold choice. Furthermore, 
with this decoupling of the density estimation from the threshold estimation, this 
leaves the potential for the incorporation of more sophisticated estimators of $\bu$, 
although this is beyond  the scope of this paper.

\subsection{Tail density estimator performance}
\label{sec:tdep}
 
Under standard regularity conditions and using standard analysis techniques, Lemma~\ref{lem:fhat} in the Appendix demonstrates that the pointwise 
bias and variance of the kernel density with unbounded data support $\hat{f}_\bY$ is given by
\begin{align*}
\Bias \lbrace \hat{f}_\bY (\by; \bH) \rbrace
&= \tfrac{1}{2}  m_2(K) \tr (\bH \D^2 f_\bY(\by)) \lbrace 1 + o(1) \rbrace \\
\Var \lbrace \hat{f}_\bY (\by; \bH)  \rbrace
&= n^{-1} |\bH|^{-1/2} f_\bY(\by) R(K) \lbrace 1 + o(1) \rbrace,
\end{align*}
where 
$m_2(K) = \int_{\mathbb{R}^d} y_1^2 K(\by) d\by, R(K) = \int_{\mathbb{R}^d} K(\by)^2 d\by$ and $\D^2 f_\bY$ is
the Hessian matrix of second order partial derivatives of $f_\bY$ with respect to $\by$.  
The equivalent result for the transformation kernel estimator $\hat{f}_\bX$ is more difficult to establish, especially  for a general transformation $\boldsymbol{t}$, 
so we focus on the logarithm transformation,
 $\bt(\bx) = (\log(x_{1d}), \dots, \log(x_d))^\top$.

\begin{theorem}
\label{thm:bias-fX}
Suppose that $\bX$ is supported on $(\boldsymbol{0}, \infty)$.
Under the regularity conditions (A1)--(A3) in the Appendix, the bias and variance of the logarithm
transformation kernel estimator $\hat{f}_\bX$ at an estimation point $\bx \in (\boldsymbol{0}, \infty)$ are
\begin{align*}
\Bias \lbrace \hat{f}_{\bX} (\bx;\bH) \rbrace
&=\tfrac{1}{2}  m_2(K)  \big[ \pi(\bx)^{-1}  f_\bX(\bx) \tr (\bH \Diag(\bx))
+ 2 \pi(\bx)^{-1} \tr (\bH \bx \D f_\bX(\bx)^\top \Diag(\bx)) \\
&\quad 
+ \tr (\bH \Diag (\bx) \Diag (\D f_\bX (\bx))) + \tr (\bH \Diag (\bx) \D^2 f_\bX (\bx) \Diag (\bx)) \big ]\lbrace 1 + o(1) \rbrace \\
\Var \lbrace\hat{f}_{\bX} (\bx;\bH) \rbrace
&= n^{-1} |\bH|^{-1/2} R(K) \pi(\bx)^{-1} f_\bX(\bx)\lbrace 1 + o(1) \rbrace,
\end{align*}
where $\pi(\bx) = \prod_{j=1}^d x_j$, $\Diag(\bx)$ is the $d \times d$ diagonal matrix with main diagonal  given by $\bx$, 
and $\D f_\bX$ and $\D^2 f_\bX$ are the gradient vector and Hessian matrix of $f_\bX$
with respect to $\bx$.  
\end{theorem}

\begin{proof}
See Appendix.
\end{proof}

Without loss of generality, the above results for $\bX$  supported on $(\boldsymbol{0}, \infty)$
may be extended to the general case for $\bX$ supported on $(\bu_0, \infty)$ following a suitable translation.

For $d=1$, the results under Theorem~\ref{thm:bias-fX} reduce to  
\begin{align*}
\Bias \lbrace \hat{f}_X (x;h) \rbrace
&= \tfrac{1}{2}  m_2(K) h^2 [  f_X(x) + 3 x f'_X(x) + x^2 f''_X(x)] \lbrace 1 + o(1) \rbrace \\
\Var \lbrace \hat{f}_X (x;h) \rbrace 
&= \frac{R(K)}{nhx} f_X(x)\lbrace 1 + o(1) \rbrace,
\end{align*}
which agree with those in \citet[Equations~(14) and (17)]{charpentier2015}. 
These authors note that if $f_X(0), f'_X(0), f''_X(0)$ are all finite, then the bias tends to $\tfrac{1}{2}  m_2(K) h^2 f_X(0)$ as $x \rightarrow 0$.
So if $f_X(0) \neq 0$, then bias and variance problems may persist when approaching the boundary.   
On the other hand, away from the boundary the bias and variance tend to 0 as $x \rightarrow \infty$.

The multivariate expressions
are not as straightforward to interpret in general, however computing the $d=2$ case explicitly 
is instructive. Writing $\bH = [h_1^2, h_{12}; h_{12}, h_2^2 ]$ as a $2\times 2$ matrix, then
\begin{align*}
\Bias \lbrace \hat{f}_\bX (\bx;\bH) \rbrace
&= \tfrac{1}{2}  m_2(K)  \Big[ 
\Big(\frac{h_1^2}{x_1} + \frac{h_2^2}{x_2} \Big) f_\bX(\bx) 
+ 2\Big(\frac{h_1^2 x_1}{x_2} \frac{\partial f_\bX(x)}{\partial x_1}
+ \frac{h_2^2 x_2}{x_1} \frac{\partial f_\bX(x)}{\partial x_2} \Big)\\
&\quad + \Big( h_1^2 x_1 \frac{\partial f_\bX(\bx)}{\partial x_1} 
+ h_2^2 x_2 \frac{\partial f_\bX(\bx)}{\partial x_2}\Big) \\
&\quad + \Big( h_1^2 (x_1^2 +x_1 x_2)\frac{\partial^2 f_\bX(\bx)}{\partial x_1^2} 
+ h_2^2 (x_1x_2 + x_2^2) \frac{\partial^2 f_\bX(\bx)}{\partial x_2^2}\Big)
\Big] \lbrace 1 + o(1) \rbrace \\
\Var \lbrace\hat{f}_{\bX} (\bx;\bH) \rbrace 
&= \frac{ R(K) f_\bX(\bx)}{n(h_1^2 h_2^2 -h_{12}^2)^{1/2} x_1 x_2} 
\lbrace 1 + o(1) \rbrace.
\end{align*}
The variance is a straightforward extension of the univariate expression.
However this is not the case for the bias: the coefficient for $f_\bX$
now involves $(h_1^2/x_1 + h_2^2/x_2)$ and  $\D f_\bX$ involves
$[h_1^2 x_1/x_2, h_2^2 x_2/ x_1]$ in addition to $[h_1^2 x_1, h_2^2 x_2]$,
due to the action of the Jacobian $|{\bf J}_\bt(\bx)| = \pi(\bx)^{-1}$. 
If $f_\bX(\boldsymbol{0}), \D f_\bX(\boldsymbol{0}), \D^2 f_\bX(\boldsymbol{0})$ are all finite
then the bias tends to  
$$\tfrac{1}{2}  m_2(K) \Big[ \Big(\frac{h_1^2}{x_1} + \frac{h_2^2}{x_2} \Big) f_\bX(\boldsymbol{0})
+ \Big( \frac{h_1^2 x_1}{x_2} \frac{\partial f_\bX(\boldsymbol{0})}{\partial x_1}
+ \frac{h_2^2 x_2}{x_1} \frac{\partial f_\bX(\boldsymbol{0})}{\partial x_2}\Big) \Big]
$$
as $x_1, x_2 \rightarrow 0$. Hence if $f_\bX(\boldsymbol{0}) \neq 0$,
then the bias grows without bound; and likewise for the variance. 
Away from the boundary, the MSE tends to 0 as $x_1,x_2 \rightarrow \infty$.
Furthermore, for general $d$, for a fixed $\bx$ in the tail region, then we have $\MSE\lbrace \hat{f}_\bX(\bx;\bH)\rbrace
=O(n^{-1} |\bH|^{-1/2} + \tr^2 (\bH))$  as $n \rightarrow \infty$.

Returning  our proposed tail density estimator, we have   
$$
\hat{f}_{\bX^{[\bu]}}(\bx; \bH) = \hat{f}_{\bX}(\bx; \bH)/\Hat{\bar{F}}_{\bX}(\bu; \bH) = \hat{f}_{\bX}(\bx; \bH)/ \bar{F}_{\bX}(\bu) \lbrace 1 + o_p (1) \rbrace 
$$  
as $\Hat{\bar{F}}_{\bX}(\bu)$ is pointwise MSE convergent to $\bar{F}_{\bX}(\bu)$ -- 
see \cite{jin1999}. 
Under the regularity conditions in Theorem~\ref{thm:bias-fX}, this implies that 
$\MSE \lbrace \hat{f}_{\bX^{[\bu]}}(\bx; \bH) \rbrace  = 
\MSE \lbrace \hat{f}_{\bX}(\bx; \bH) \rbrace / \bar{F}_{\bX}(\bu)^2 \lbrace 1 + o  (1) \rbrace$, so the properties of the tail density estimator $\hat{f}_{\bX^{[\bu]}}$ 
largely carry over
from the transformation kernel density estimator $\hat{f}_\bX$,
with the important difference that $\hat{f}_{\bX^{[\bu]}}$ suffers much less from 
boundary problems than $\hat{f}_{\bX}$.
This is because $\hat{f}_\bX$ has potentially undesirable behaviour near its boundary $\bu$, whereas we only require  $\hat{f}_\bX$ to be calculated on $(\bu_0, \infty)$, with $\bu \gg \bu_0$.
Note that $\bu_0$ is fixed, and there is no true value to estimate; see the simulation studies of Section~\ref{sec:numerical} for an example.

An alternative for density estimation in heavy tails is to vary the
amount of smoothing, rather than to apply a stabilising transformation 
e.g. \citet{loftsgaarden1965,abramson1982},
though these estimators do not account for data boundedness. 
To handle the boundedness of the data sample, another approach is based on modifying the kernel 
function itself to avoid assigning probability mass outside the
data support, e.g. \citet{gasser1979,chen1999}.
These techniques are focused on boundary behaviour and do not address the issue of spurious bumps in the tails which are far away from this boundary.
Our proposed logarithm transformation kernel estimator is able to 
handle both of these issues simultaneously.

\subsection{Optimal bandwidth computation}
\label{sec:bandwidth}

The complicated form of pointwise $\MSE \lbrace \hat{f}_\bX (\bx;\bH)\rbrace$ does not facilitate the computation
of a closed form mean integrated squared error, so it is not feasible to define an oracle bandwidth
for the transformation density estimator $\hat{f}_\bX$.  
Since the estimation is carried out in the unbounded space of $\bY_1, \dots, \bY_n$, then
our strategy is to carry out the bandwidth selection on these transformed data,
as there is large body of data-based bandwidth selectors which lead to 
consistent density estimates. The back-transformation to the original data scale 
does not require any adjustment to this bandwidth to compute the transformation
density estimator, and subsequently to the tail density estimator.   

From Lemma~\ref{lem:fhat} in the Appendix,  
the mean integrated squared error (MISE) of  
the density estimator $\hat{f}_\bY$ is 
\begin{align*}
\MISE \lbrace\hat{f}_\bY (\cdot; \bH) \rbrace
&= \big [\tfrac{1}{4}  m_2^2(K) ( \vec^\top \bH \otimes \vec^\top \bH) \bpsi_{\bY,4}
 + n^{-1} |\bH|^{-1/2} R(K) \big] \lbrace 1 + o(1) \rbrace.
\end{align*}
where $\bpsi_{\bY,4} = \int_{\mathbb{R}^d} \D^{\otimes 4} f_\bY (\by) f_\bY (\by) \mathrm{d}\by$, 
as defined in \citet{chacon2010}, and vec is the vectorisation operator which stacks the columns of matrix into a single column. 
Using this MISE expression, we can then
define an oracle optimal bandwidth choice as the minimiser of the MISE
\begin{equation}
\label{eq:H1H2}
\bH^* = \argmin{\bH \in \mathcal{F}} \MISE \lbrace\hat{f}_\bY ( \cdot; \bH)\rbrace 
= O(n^{-2/(d+4)})
\end{equation}
where $\mathcal{F}$ is the space of $d \times d$ symmetric positive definite matrices. 
Furthermore, utilising this optimal bandwidth in $\hat{f}_\bY$, the minimal MISE is 
$\inf_{\bH \in \mathcal{F}} \MISE \lbrace\hat{f}_{\bY} (\cdot;\bH)\rbrace 
= O (n^{-4/(d+4)})$.  
With this bandwidth matrix order, 
for a non-random point $\bx$ in the tail region, the minimal MSE for the tail density estimator is 
$\inf_{\bH \in \mathcal{F}}  \MSE\lbrace \hat{f}_{\bX^{[\bu]}}(\bx;\bH)\rbrace
=O(n^{-4/(d+4)})$ also, as $n \rightarrow \infty$.

The optimal bandwidth selector defined in Equation~\eqref{eq:H1H2} is mathematically intractable  
as it depends on unknown quantities. 
Accordingly a vast body of research in the density estimation literature has focused on providing data-based
bandwidth selectors which estimate or approximate the optimal bandwidth. 
There are three main classes: (i) normal scale (or rule of thumb), (ii) plug-in and (iii) cross validation. 

The class of normal scale selectors is an extension to the multivariate case of the quick and simple bandwidth selectors 
where the unknown density $f$ is replaced by a normal density, leading to
$$
\hat{\bH}_{\mathrm{NS}} = \left[ \frac{4}{(d+2)n} \right]^{2/(d+4)} {\bf S} n^{-2/(d+4)}
$$ 
where ${\bf S}$ is the sample covariance matrix of $\bY_1, \dots, \bY_n$
\cite[see e.g.][p.~111]{wand1995}.

The class of plug-in selectors consists of a generalisation 
of the work of \cite{sheather1991} for univariate data by \citet{wand1994} and \citet{duong2003} for multivariate data.
Plug-in selectors use as a starting point the AMISE formula (Asymptotic MISE) where the 
only unknown quantity is  the $\bpsi_{\bY,4}$ functional.
The fourth order differential $\D^{\otimes 4}$ is expressed as a vector of length $d^4$, resulting 
from a four-fold Kronecker product of the first order differential $\D$. Replacing this by an estimator 
$\hat{\bpsi}_{\bY,4}$
yields the plug-in criterion
$$
\mathrm{PI} (\bH) = \tfrac{1}{4}  m_2^2(K) ( \vec^\top \bH \otimes \vec^\top \bH) \hat{\bpsi}_{\bY,4}(\bG)
 + n^{-1} R(K) |\bH|^{-1/2}  
$$
where 
$m_2(K)$ is defined in Section \ref{sec:tdep},
$\vec$ is the operator that stacks the element of a matrix column-wise into a vector, 
$\hat{\bpsi}_{\bY,4}(\bG) = n^{-2} \sum_{i,j=1}^n \D^{\otimes 4} L_\bG (\bY_i - \bY_j)$, 
$L_{\bG}$ is an initial pilot kernel with pilot bandwidth matrix $\bG$ and 
$R(K) = \int_{\mathbb{R}^d} K(\bx)^2 \mathrm{d}\bx$.
The plug-in selector $\hat{\bH}_{\mathrm{PI}}$ is the minimiser over $\mathcal{F}$ of $\mathrm{PI} (\bH)$.

For the class of cross validation selectors we focus on unbiased (or least squares) cross validation
and smoothed cross validation.
Unbiased cross validation (UCV) was introduced by \citet{bowman1984} and \citet{rudemo1982}
for the univariate case.    
The unbiased cross validation selector, $\hat{\bH}_{\mathrm{UCV}}$ for the multivariate case \citep{sain1994}, is defined as the minimiser 
over $\mathcal{F}$ of
\begin{align*}
\mathrm{UCV} ( \bH ) = \int_{\mathbb{R}^d} \hat{f}_{\bY} ( \by; \bH )^2 \mathrm{d} \bx
-2 n^{-1} \sum_{i=1}^n \hat{f}_{\bY,-i} ( \bY_i; \bH ),
\end{align*}
where 
$
\hat{f}_{\bY,-i} ( \bY_i; \bH ) = [n( n-1 )]^{-1} \sum_{j=1}^n K_{\bH} ( \bY_i - \bY_j ).
$
The smoothed cross validation (SCV) selector 
$\hat{\bH}_{\mathrm{SCV}}$, is defined as the minimiser over $\mathcal{F}$ of
\begin{align*}
\mathrm{SCV} ( \bH ) 
&= n^{-2} \sum_{i=1}^n \sum_{j=1}^n
( K_{\bH} * K_{\bH} * L_{\bG} * L_{\bG} - 2K_{\bH} * L_{\bG} * L_{\bG}
+ L_{\bG} * L_{\bG}) ( \bY_i - \bY_j) \\
&\qquad \qquad \qquad  +n^{-1} R( K ) | \bH |^{-1/2},
\end{align*}
where $*$ is the convolution operator, as introduced by \cite{hall1992} for univariate data, and by \cite{sain1994} for multivariate data. 
If there are no replications in the data, then SCV with $\bG = 0$ is identical to UCV 
as the pilot kernel $L_0$ can then be thought of as the Dirac delta function.

The UCV selector can be directly computed as it contains no unknown quantities, however specification of the bandwidth $\bG$ of the pilot kernel is required
for the plug-in and SCV
selectors.  Computational data-based algorithms which address this are found in \citet{wand1995,duong2003} for plug-in selectors and 
\citet{hall1992,duong2005} for SCV selectors.

\subsection{Tail density estimation via histograms}

Histograms, especially for univariate data, are widely used as alternatives to kernel estimators
for visualising data samples, even when focusing on distributional tails \cite[see e.g.][]{perkins2007,perkins2013}. Their advantages include computational and mathematical simplicity,
and that they do not suffer from the boundary bias problems of standard kernel estimators.  
In the context of tail density estimation, we divide the data range of the sample $\bX_1, \ldots, \bX_n$ 
 into a regular partition of hypercubes $A_i$ of size 
$b_1 \times \cdots \times b_d$, and define the binwidth as $\bb = (b_1, \ldots, b_d)^\top \in \mathbb{R}^d$.
The histogram estimator of $\tilde{f}_{\bX}$ at a point $\bx$ in a bin $A_i$ is
$$
\tilde{f}_{\bX} ( \bx; \bb ) = \frac{\gamma_i}{n b_1\cdots b_d}
$$
where $\gamma_i$ represents the number of observations in the hypercube $A_i$. 
The histogram estimator of the tail density $\tilde{f}_{\bX^{[\bu]}}$ is 
$$
\tilde{f}_{\bX^{[\bu]}} ( \bx; \bb ) = \tilde{f}_{\bX} (\bx; \bb)/ \Tilde{\bar{F}}_{\bX} (\bu; \bb)
$$
where $\Tilde{\bar{F}}_{\bX} (\bu; \bb)$ counts the number of observations in the hypercubes covered by $(\bu, \infty)$, divided by $n b_1\cdots b_d$. 
If conditions similar to (A1) and (A3) in the Appendix are fulfilled then, by \citet[Theorem 3.5]{scott2015},
the MISE of the histogram estimator is %
$\MISE \lbrace \tilde{f}_{\bX^{[\bu]}} (\cdot; \bb )\rbrace = O ((n b_1\cdots b_d)^{-1} +\bb^\top\bb) $
with minimal MISE
$
\inf_{\bb>0} \MISE \lbrace\tilde{f}_{\bX^{[\bu]}} (\cdot; \bb)\rbrace = O(n^{-2/(d+2)}).
$
This is asymptotically slower than the $O(n^{-4/(d+4)})$ minimal MSE rate for 
the kernel estimator $\hat{f}_{\bX^{[\bu]}}(\bx)$ for $\bx$ not in  the boundary region.
Hence, from the mean squared error perspective, the kernel density estimator is preferable
to a histogram for density estimation in the tail region, especially as the dimension $d$ increases.

Analogous with the data-based optimal bandwidth selectors, 
the normal scale optimal binwidth \cite[Theorem 3.5]{scott2015} is 
\begin{equation}
\label{eq:binwidth}
\hat{b}_j = 2 \times 3^{1/(d+2)} \pi^{d/(d+4)} s_j n^{-1/(d+2)}
\end{equation}
where $s_j,j=1,\ldots,d$ are the marginal sample standard deviations of $\bX_1^{[\bu]}, \ldots, \bX_n^{[\bu]}$.
There is no equivalent variety of binwidth selectors which generalise Equation~\eqref{eq:binwidth} compared to bandwidth
selectors (Section \ref{sec:bandwidth}) due the slower asymptotic performance of histograms as compared to kernel estimators.

\subsection{Model assessment via tail density estimation}
\label{sec:model}

Tail density estimation can provide one way to assess the fidelity of the observed dataset to one or more candidate models. For example, in climate science different climate models commonly produce competing predictions of  environmental variables. The performance of these models is often validated by comparing the model predicted output, with that of the observed data \citep{flato2013}. These comparisons may be based on the full body of the predicted variables, or focus primarily on the extremes \cite[e.g.][]{perkins2007,perkins2013}.
Similarly, in the context of extreme value theory, the analyst is regularly required to determine which of multiple competing parametric families, such as max-stable distributions, provides the best fit to an observed extremal dataset \cite[e.g.][]{coles+t94}

Suppose that we have a suite of parametric models indexed by $\mathcal{M}=\{1,\ldots,M\}$, and we wish to determine which of them most appropriately describe the tails of the underlying distribution of the observed dataset, $f_{{\boldsymbol X}}$.
\citet{perkins2013} utilised the histogram estimator $\tilde{f}_{\bX^{[u]}}$ 
of the observed data sample 
as a surrogate for the unknown
target $f_{{\boldsymbol X}^{[\bu]}}$, and so the fit of the parametric models was assessed according to the discrepancy 
of the parametric (tail) density functions $g_1, \dots, g_M$ defined over $(\bu,\infty)$ and the histogram $\tilde{f}_{\bX^{[\bu]}}$.
Their tail index (generalised here to $d$ dimensions) is given by
$$
\tilde{T}_1(g_j) = \int_{(\bu, \infty)} |g_j(\bx) - \tilde{f}_{\bX^{[\bu]}} (\bx; \bb)|   \mathrm{d}\bx,
$$
with the preferred models being those which give the smaller or smallest discrepancy 
$$\argmin{j\in\mathcal{M}} \tilde{T}_1(g_j).$$  
Note that the subscript of $\tilde{T}_1$ indicates the $L_1$ error measure used in its definition. 
We prefer to use the $L_2$ error to assess a model fit:
\begin{align} 
\tilde{T}_2(g_j) &=  \int_{(\bu, \infty)} [ g_j(\bx) - \tilde{f}_{\bX^{[u]}}(\bx; \bH) ]^2  
\, \mathrm{d}\bx. \label{eq:model_hist}
\end{align}
An improvement to this procedure is
to  replace the histogram in Equation~\eqref{eq:model_hist} with the transformation kernel estimator  $\hat{f}_{\bX^{[\bu]}}$:
\begin{equation}
\label{eq:model_kern}
\hat{T}_2(g_j) = \int_{(\bu, \infty)} [g_j(\bx) - \hat{f}_{\bX^{[\bu]}} (\bx; \bH)]^2  \mathrm{d}\bx. 
\end{equation}
This will accordingly allow the usual artefacts of histogram estimators to be avoided or at least reduced. 
These include the anchor point problem (i.e. how to specify the locations of the histogram bins)
and the empty bin problem (where it is unclear whether histogram bins with empty counts
should be interpreted as a true zero probability or are due to insufficient observed data). 
This latter case is important for extreme values as they are sparsely distributed in the tail regions.

In the following section, we
highlight the purpose of working with transformed density estimators, by directly contrasting $\tilde{T}_1(g_j)$ and $\hat{T}_1(g_j)$ with the index based on the standard kernel density estimator 
\begin{equation}
\label{eq:model_st_kern}
\hat{T}^*_2(g_j) = \int_{(\bu, \infty)} [g_j(\bx) - \hat{f}^*_{\bX^{[\bu]}} (\bx;\bH)]^2 \mathrm{d}\bx,
\end{equation}
where $\hat{f}^*_{\bX^{[\bu]}} (\bx)$ represents the standard kernel density estimator constructed without applying the transformation $\bt$.
As presented in the Introduction, in the univariate case, under the Pickands-Balkema-de Haan theorem, the observations above some high threshold $u$ can be approximated by the GPD distribution. 
We thus define by $\check{f}_{X^{[u]}}$ the GPD tail density estimator constructed from the observations above the threshold $u$ and respective tail index using the $L_2$ error by
\begin{equation}
\label{eq:model_gpd}
\check{T}_2(g_j) = \int_{(u, \infty)} [g_j(x) - \check{f}_{X^{[u]}} (x)]^2 \mathrm{d}x.
\end{equation}
The integrals in Equations~\eqref{eq:model_hist}--\eqref{eq:model_gpd} can be
approximated by (weighted) Reimann sums.

\section{Numerical studies}
\label{sec:numerical}

\subsection{Simulated data - univariate}
\label{sec:num_res_univ}

We now numerically examine the performance of the kernel density estimator introduced in Section \ref{sec:tail} for moderate univariate extremes, and demonstrate that it is a good surrogate for the true tail distribution. We additionally evaluate the estimator's performance through the model assessment procedure of Section \ref{sec:model}.
Gaussian kernels are adopted throughout given the usual secondary level of importance given to kernel choice in standard kernel methods (see e.g. Table 2.1 of \citet{wand1995} where the difference between the most efficient (Epanechnikov) kernel and the least efficient (uniform) kernel is less than 7\%).

We generate a dataset of size $n=2,000$ from each of the Gumbel, Fr\'{e}chet and generalised Pareto (GPD) target distributions, and set the threshold $u$ at the $95\%$ upper sample quantile.
For each sample, we  compute: 
\begin{enumerate}
\item[1)] The appropriate maximum-likelihood based, parametric estimator: Fr\'echet (FRE), Gumbel (GUM) or generalised Pareto (GPD);

\item[2)] The maximum likelihood estimator of the generalised Pareto distribution, $\check{f}_{X^{[u]}}$,  using only the observations above the threshold $u$ (GPD+);

\item[3)] The histogram estimator $\tilde{f}_{X^{[u]}}$ with normal scale optimal binwidth (HIS);
\item[4)] The transformation kernel based estimator $\hat{f}_{X^{[u]}}$ with transformation $t(x) = \log (x- u_0)$, where $u_0 = \min (X_1, \dots, X_n) - 0.05\operatorname{range}(X_1, \dots, X_n)$ using  the normal scale (KNS), plug-in (KPI), 
unbiased cross validation (KUC)  and smoothed cross validation (KSC) optimal bandwidth selectors;
\item[5)] The standard kernel based estimators $\hat{f}^*_{X^{[u]}}$, 
using the normal scale (KNS*), plug-in (KPI*), 
unbiased cross validation (KUC*)  and smoothed cross validation (KSC*) optimal bandwidth selectors.
\end{enumerate}

The top row of Figure \ref{fig01} illustrates the various tail density estimates for the three target distributions, with the true density shown as a solid black line.
Displayed are the generalised Pareto estimator $\check{f}_{X^{[u]}}$ (GPD+; grey long-dashed), the histogram estimator $\tilde{f}_{X^{[u]}}$ (HIS; green dotted line), and the transformed $\hat{f}_{X^{[u]}}$ (KPI; red dashed) and standard $\hat{f}^*_{X^{[u]}}$ (KPI*; blue dot-dashed) kernel density estimates, both with plug in estimators only for clarity. 
Visually, the transformed kernel estimators appear to be the more accurate non-parametric estimators of tail behaviour in each case, being noticeably smoother and less noisy.
That the kernel density estimators are naturally continuous functions also leads to better visualisations than histogram based estimators, and they are more helpful when comparing to a continuous target density.

Both the transformed kernel and the GPD (using the largest observations only) estimators appear to provide comparable fit in each case.

\begin{figure}[h!]
\centering 
\setlength{\tabcolsep}{3pt}
\begin{tabular}{@{}ccc@{}}
Target Fr\'{e}chet & Target  Gumbel & Target GPD \\
\includegraphics[width=0.32\textwidth]{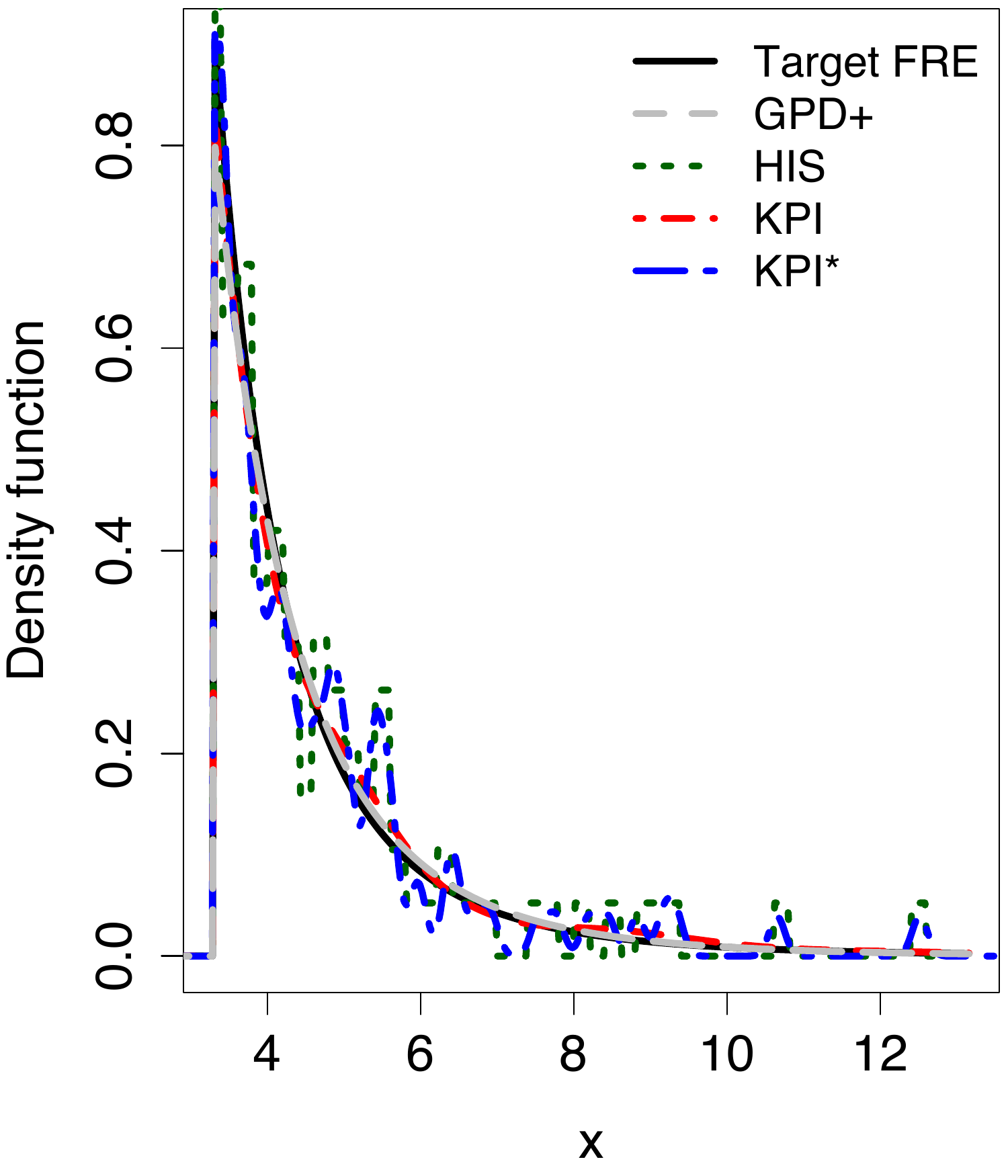} &
\includegraphics[width=0.32\textwidth]{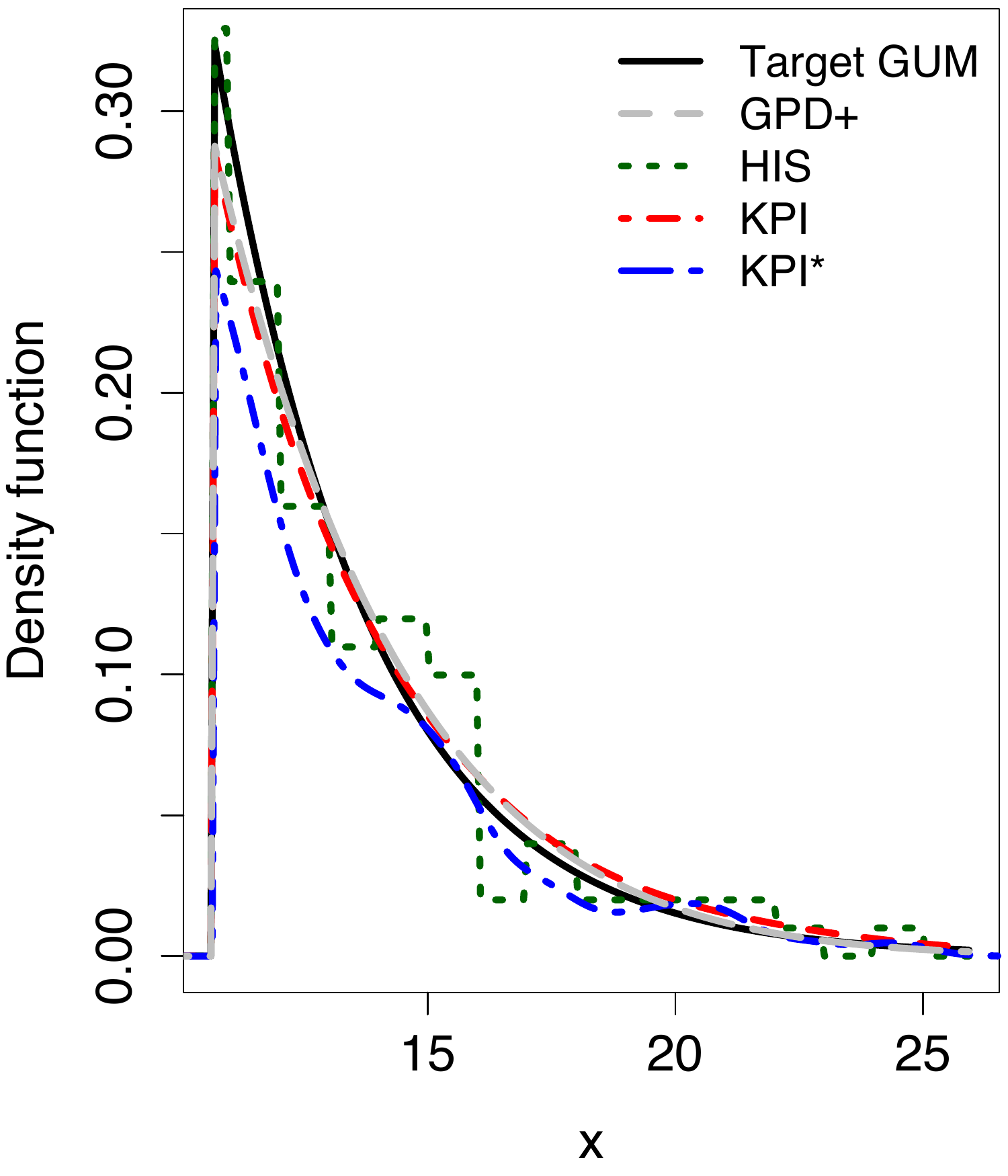} &
\includegraphics[width=0.32\textwidth]{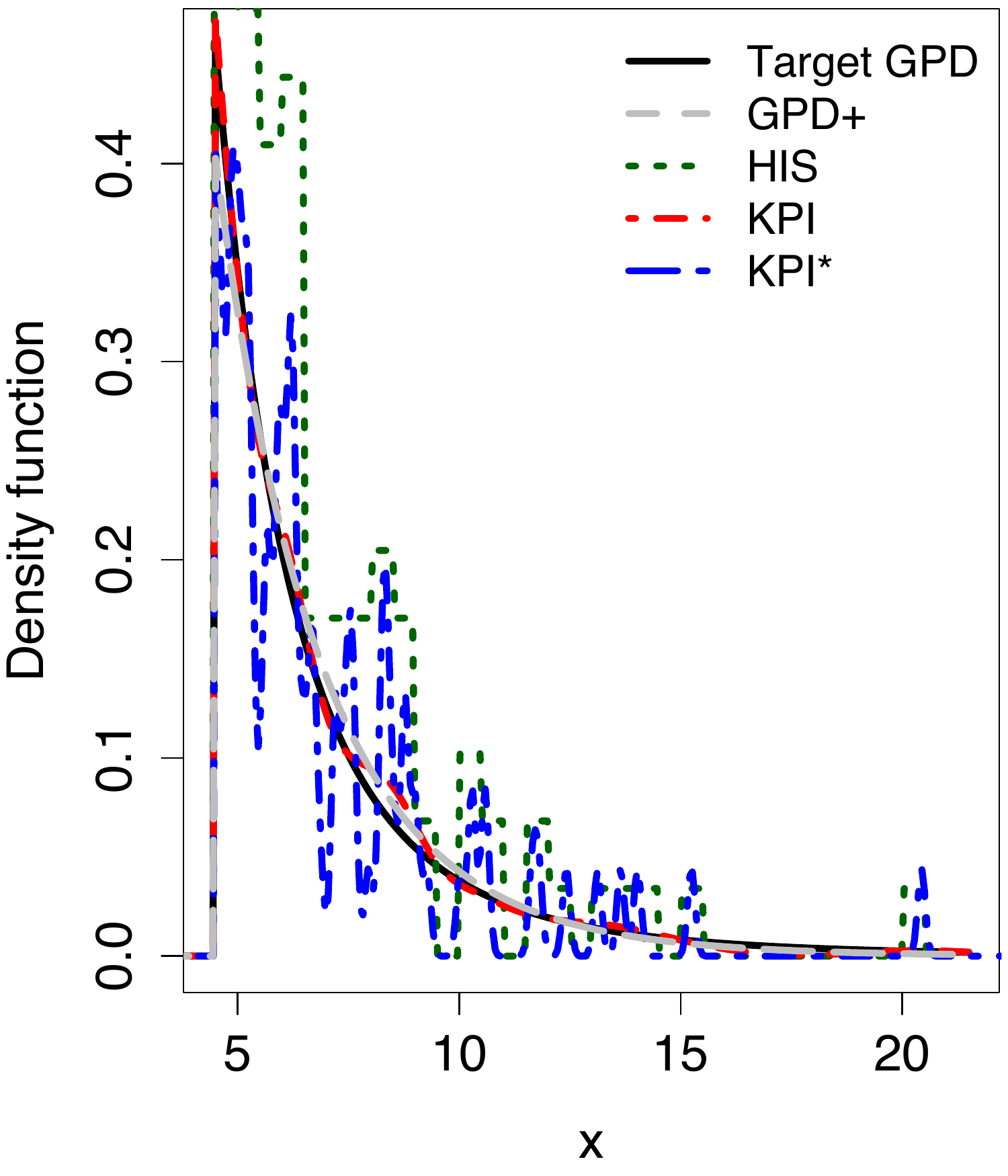} \\
\includegraphics[width=0.32\textwidth]{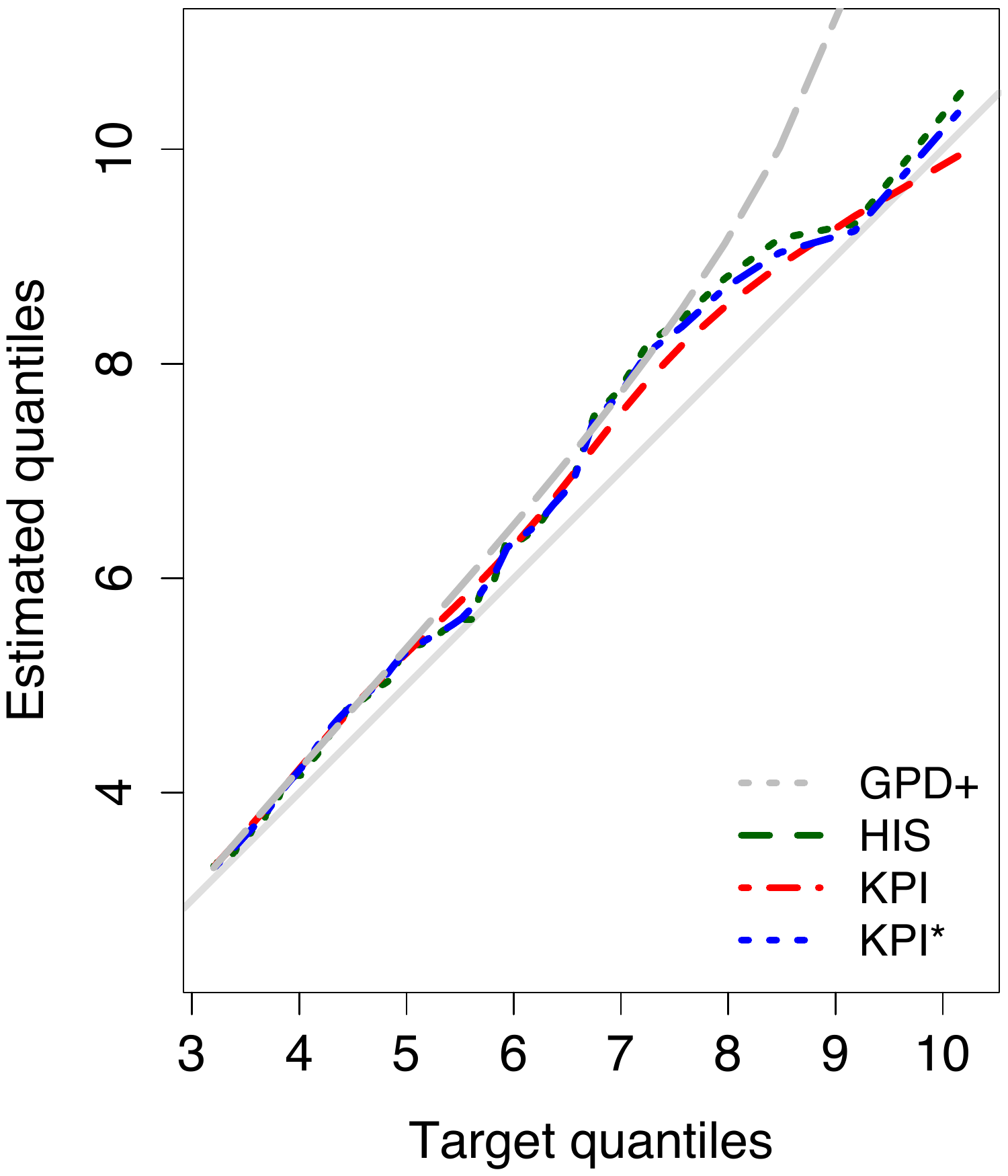} &
\includegraphics[width=0.32\textwidth]{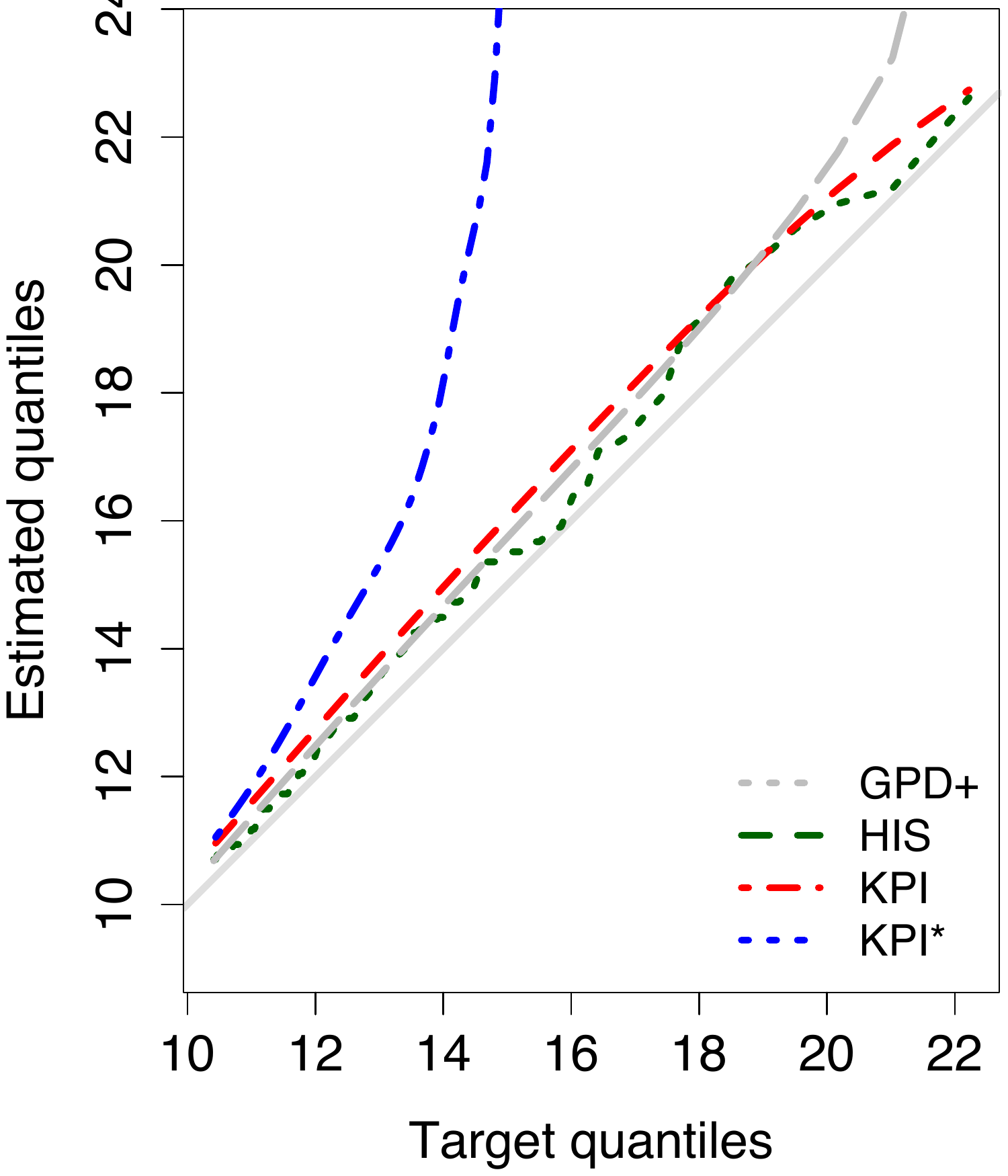} &
\includegraphics[width=0.32\textwidth]{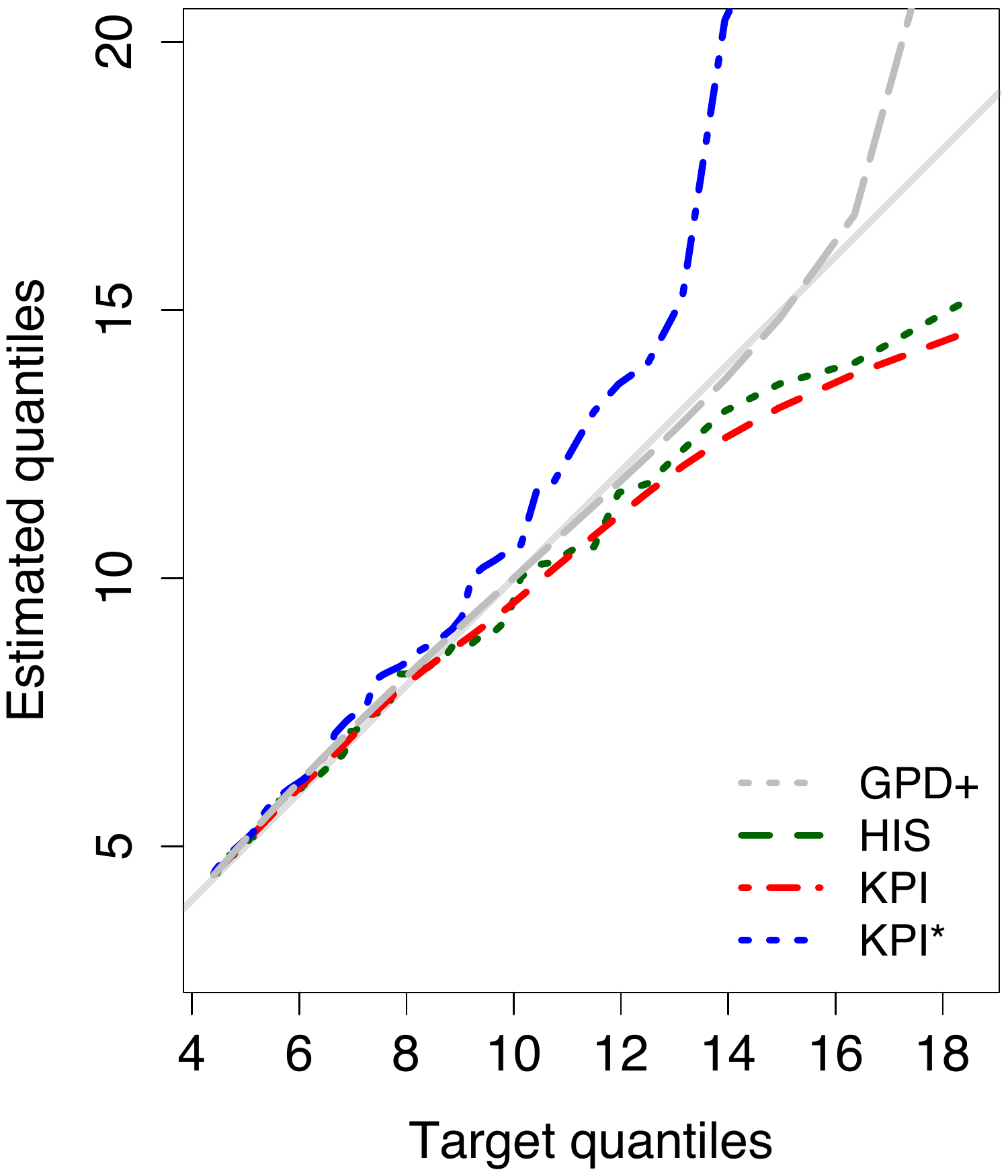}
\end{tabular}
\caption{\label{fig01} \small Generalised Pareto estimator $\check{f}_{X^{[u]}}$ (GPD+; grey long-dashed) and non-parametric estimators of the univariate tail density (top) and of the tail quantiles (bottom) when the target density is Fr\'{e}chet (left), Gumbel (centre) and generalised Pareto (right). Sample size is $n=2000$.
Fr\'{e}chet ($\mu=1$, $\sigma=0.5$, $\xi=0.25$), Gumbel ($\mu=1.5$, $\sigma=3$) and Pareto ($\mu=0$, $\sigma=1$, $\xi=0.25$) target densities are represented by a solid black line. The histogram estimator $\tilde{f}_{X^{[u]}}$ with normal scale binwidth (HIS) is represented by a dotted green line, the
transformed kernel plug-in estimator $\hat{f}_{X^{[u]}}$ (KPI) by a short dashed red line and the standard kernel estimator $\hat{f}^*_{X^{[u]}}$ (KPI*) by a dot-dash blue line. }
\end{figure}
The bottom row of Figure \ref{fig01} examines the extremal performance of the same estimators through qq-plots of the target quantiles versus the GPD+ and non-parametric estimated quantiles, for target quantiles ranging from 95\% to 99.9\%.
Of all non-parametric estimators, the histogram estimator most consistently  approximates the true quantiles. This performance compared to the kernel-based estimators is not unexpected, however, as the latter aim to optimally estimate the density function rather than the quantile function.
Comparing the two kernel-based estimators, the transformed kernel estimator tends to either outperform (centre, right panels) or perform as well as (left panel) the standard estimator, which can be attributed to the standard estimator's natural boundary bias.
The tail quantiles obtained from the transformed kernel estimators appear to perform better than those of the generalised Pareto estimator (GPD+) when the target density is Fr\'{e}chet, and they are comparable for the Gumbel target. Unsurprisingly, the GPD+ estimator performs the strongest when the data are in fact GPD distributed.

Note that the transformed kernel estimator has produced estimates with lighter tails than the true density. For large $n$ this is possibly due to the choice of a Gaussian kernel $K_h$ to construct the density estimates, so that the upper tail of this estimate (mapped through the inverse transform $t^{-1}$) is light compared to the true Fr\'echet and Pareto tails.
For smaller $n$, finite sample variation can produce a density estimate with either lighter or heavier tails in the body of the data (see Supplementary Information).

For more quantitative results we repeat this process over $400$ replicates for different dataset sizes $n=500, 1000$ and $2000$, producing tail samples of size $m= 25, 50$ and $100$, with the threshold $u$ set at the $95\%$ upper quantile.  
As these three sample sizes gave similar results, we only 
present those for $n=2000$ here for brevity. See the Supplementary Information for results with $n=500,1000$.
We take a numerical approximation 
(Reimann sum)  of the $L_2$ loss ($\tilde{T}_2$, $\hat{T}_2$, $\hat{T}^*_2$ and $\check{T}_2(g_j)$).

Figure \ref{fig02} presents box-plots of the accuracy of each tail density estimator for each true tail distribution.
As expected, for each target distribution the most accurate density estimator is the correctly specified parametric model.
The transformed kernel density estimators systematically perform better than their standard kernel counterparts,  although they can be more variable.
The standard kernel estimators and the histogram estimator compete for the worst estimate of the tail density, depending on the true target distribution. 
The differences in the accuracy between kernel estimators with different bandwidth selectors is small, in contrast to 
studies where the bandwidth selection class is a crucial factor
\cite[see e.g.][Chapter~3]{sheather1991,wand1995}, indicating that it is the difference between estimators that is dominating performance.
Using the normal scale bandwidth selector (KNS) provides a greater accuracy compared to other bandwidth selectors when the target distribution is GPD, and this is also slightly  evident for Fr\'{e}chet  distributed data.
The best transformed kernel and the GPD tail density estimators (KNS and GPD+) appear to perform equally well when the target distribution is Fr\'{e}chet or Gumbel, with a slight advantage to the KNS estimator in the case of Gumbel distributed data.
Clearly the GPD+ estimator is over-performing when the target is GPD. 

\begin{figure}[tb]
\centering 
\setlength{\tabcolsep}{3pt}
\begin{tabular}{@{}ccc@{}}
Target Fr\'{e}chet & Target Gumbel & Target GPD \\
\includegraphics[width=0.31\textwidth]{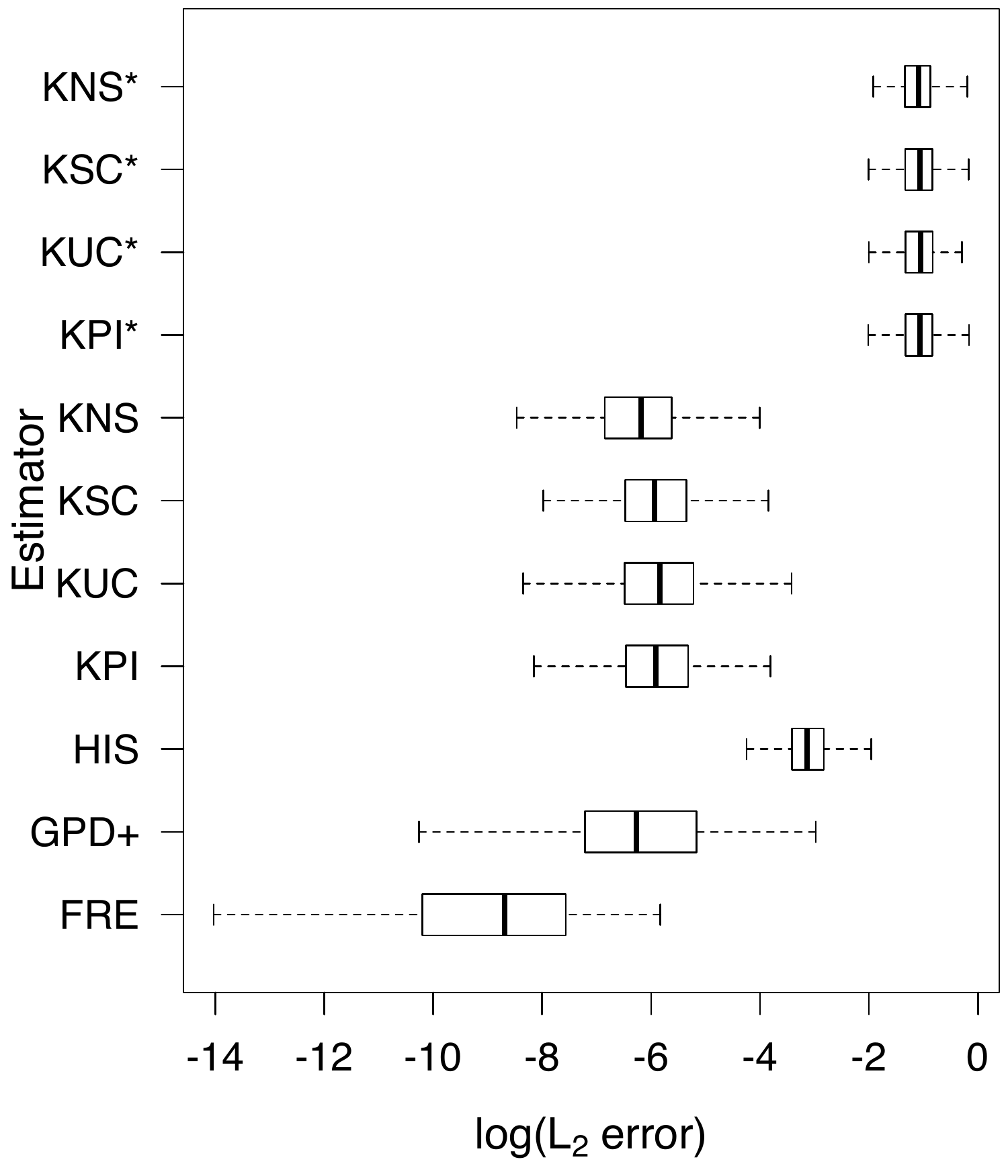} &
\includegraphics[width=0.31\textwidth]{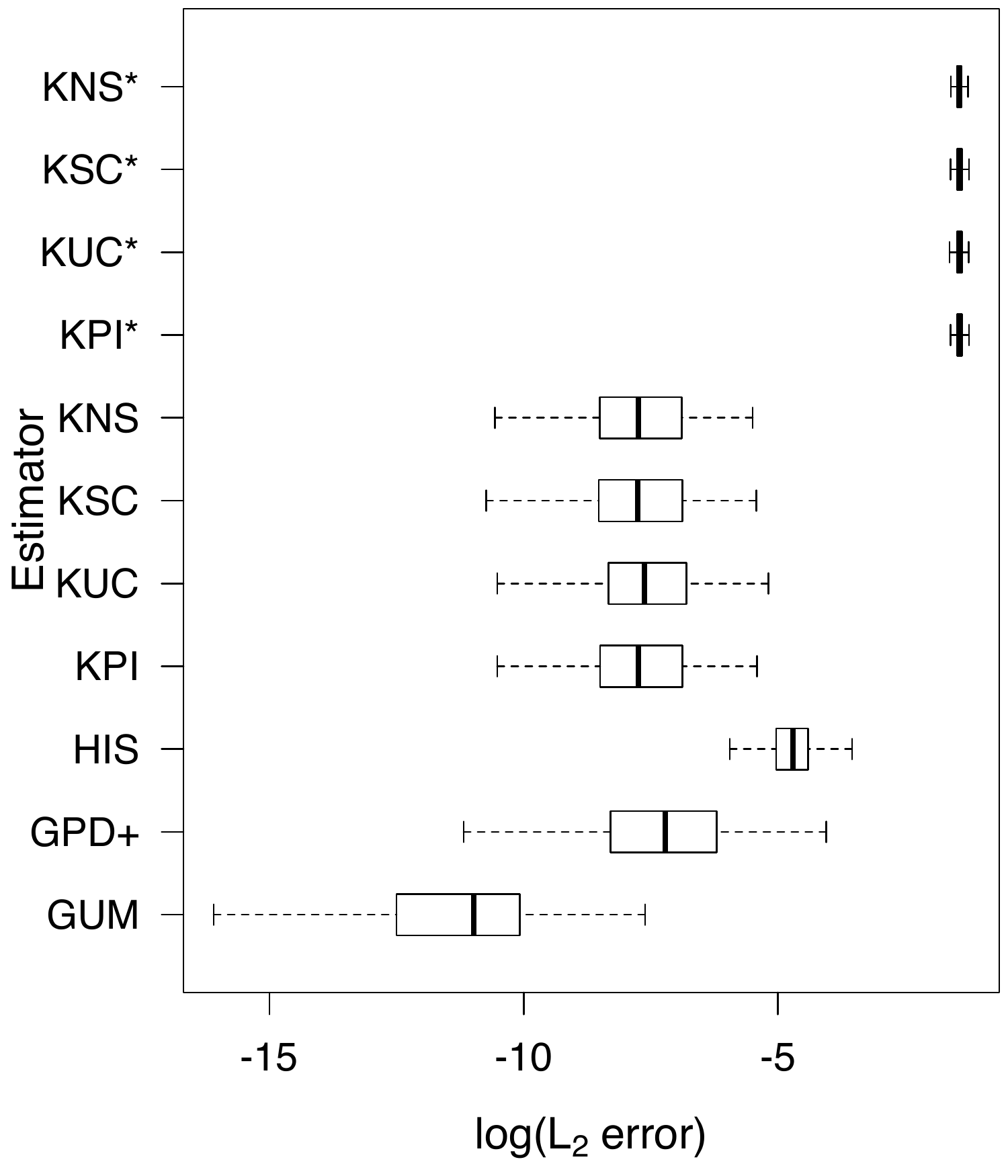} &
\includegraphics[width=0.31\textwidth]{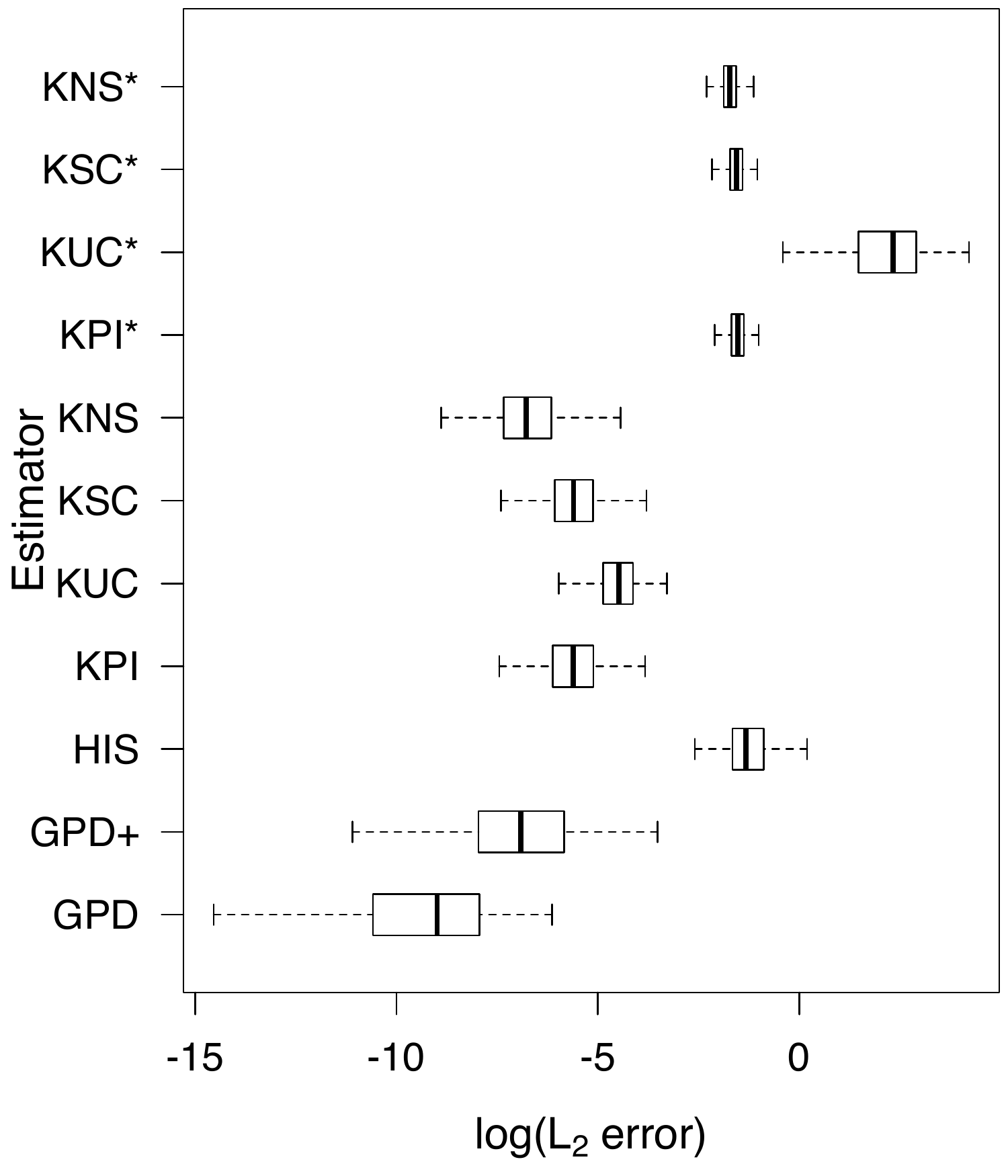} 
\end{tabular}
\caption{\label{fig02} \small Box-plots of the $\log L_2$ errors for the parametric Fr\'{e}chet (FRE), Gumbel (GUM), generalised Pareto (GPD), and histogram (HIS) tail density estimators as well as the generalised Pareto $\check{f}_{X^{[u]}}$ (GPD+). Transformed kernel density estimators  $\hat{f}_{X^{[u]}}$ use the plug-in (KPI), 
unbiased cross validation (KUC), smoothed cross validation (KSC) 
and normal scale kernel (KNS) optimal bandwidth selectors. Standard kernel density estimators $\hat{f}^*_{X^{[u]}}$ are indicated by an asterisk (*). 
True target densities are (left panel) Fr\'{e}chet, (centre) Gumbel and (right) GPD. Box plots are based on 400 replicates of $n=2,000$ observations.
}
\end{figure}

Finally, we examine the density estimator performance in terms of its ability to correctly select the true, data-generating model (Section \ref{sec:model}). 
For each of the $3\times 400$ datasets generated previously, we compute the tail indices
of the $L_2$ loss ($\tilde{T}_2$, $\hat{T}_2$, $\hat{T}^*_2$ and $\check{T}_2(g_j)$)
with respect to each parametric model. 

Table \ref{tab01} displays the proportion of times that samples from a given true distribution are identified as coming from either Fr\'echet, Gumbel or GPD distributions (i.e. by having the smallest tail index value), as a function of tail density estimator.
In each case, the highest proportion of replicates selecting the correct model is given in bold. 
As the Gumbel distribution ($\xi=0$)  is on the limiting border of the parameter space of the  Fr\'{e}chet distribution ($\xi>0$), to avoid possible model misidentification, we additionally perform a deviance test.
If the Gumbel provides a significantly better fit than the Fr\'{e}chet distribution, meaning that the shape parameter is not significantly different from zero, then we only consider Gumbel and GPD distributions as candidate models.

As might be expected, for any target distribution in Table \ref{tab01}, 
using the transformation based estimator ($\hat{T}_2$) as a surrogate for the target density generally selects the correct target in the
vast majority of cases, with proportions substantially higher than those achieved through the standard kernel and histogram tail indices, $\hat{T}^*_2$ and $\tilde{T}_2$. 
Both the transformation and GPD based estimators ($\hat{T}_2$ and $\check{T}_2$) have comparable abilities in correctly identifying the underlying distribution with proportions around $0.90$ and higher.
When the true density is Fr\'{e}chet, the best performing estimator determined by the $L_2$ error measure favours $\check{T}_2$ while it favours $\hat{T}_2$ when the target is GPD.
When this study is repeated with smaller sample sizes ($n=1,000$ and $500$), the superiority of the transformation-based tail index compared to the GPD-based tail index is more clear (see Tables \ref{tab06} \& \ref{tab07} in the Supplementary Information). 
Overall, the transformation kernel-based index $\hat{T}_2$ performs as strongly as, and in some cases better than the GPD-based index $\check{T}_2$ 
and consistently better than the non-parametric-based indices $\hat{T}^*_2$ and $\tilde{T}_2$.

\begin{table}[h!]
\centering
\begin{tabular}{@{\extracolsep{4pt}}lcccccccccc@{}}
Target & $\tilde{T}_2$ & $\hat{T}_2$ & $\hat{T}^*_2$ & $\check{T}_2$ \\
\hline
FRE & 0.74 &  0.89 & 0.80 & \bf{0.92}\\
GUM & {\bf 1.00} & {\bf 1.00} & 0.00 & {\bf 1.00} \\
GPD & 0.19 & {\bf 0.98} & 0.00 & 0.95\\
 \hline
\end{tabular}

\caption{\small Proportion of 400 simulated datasets from each known target distribution (Fr\'echet, Gumbel and GPD) that are correctly identified as coming from each of these distributions by having the smallest tail index value. Bold text indicates the highest proportion for each target model. Nonparametric density estimators are the
 histogram ($\tilde{T}_2$), the transformed kernel ($\hat{T}_2$) 
and the standard kernel ($\hat{T}^*_2$). The parametric GPD estimator on tail data is $\check{T}_2$. Tail indices are calculated according to the $L_2$ loss.
}
\label{tab01}
\end{table}

\subsection{Simulated data - multivariate}
\label{sec:num_res_mult}

The analysis of multivariate extremes is considerably more challenging than its univariate counterpart. A powerful motivation for exploratory data analysis using kernel-based estimation is that no single parametric family exists for max-stable distributions. See e.g. 
 \citet{kotz2000, coles2001, beirlant2004, dehaan2006, falk2011} and \citet{beranger2015b} for theoretical details and applications.
Although  a multivariate extension of the GPD distribution is available  \citep[see for example,][]{rootzen2006,rootzen2017} we do not consider it here, as
its principles are based on at least one marginal component exceeding some high threshold rather than considering all components to be above a threshold, which is our focus here.

We now numerically examine the performance of the bivariate transformation-based density estimator.
We generate datasets of size $n=4,000$ 
from the 
asymmetric negative logistic \citep[ANL;][]{joe1990}, the bilogistic \citep[BIL;][]{smith1990b}  and the 
H\"{u}sler-Reiss \citep[HR;][]{husler1989} distributions. 
The threshold $\bu$ is determined as each dataset's marginal $90\%$ upper quantiles.

For each dataset we compute the appropriate maximum likelihood based parametric estimator (assuming simultaneously estimated generalised extreme value distribution margins), a 2-dimensional histogram with normal scale optimal bandwidth (HIS) and the transformation and standard kernel estimators with plug-in optimal bandwidth selectors (respectively KPI and KPI*) and transformation $\bt(\bx)=(\log(x_1-u_{01}),\log(x_2-u_{02}))^\top,$ 
where $u_{0j} = \min\{X_{1j}, \dots, X_{nj}\} - 0.05 \operatorname{range}\{X_{1j}, \dots, X_{nj}\}$, $j=1,2$.
The results from the other bandwidth selectors are not displayed both for clarity, and due to the limited impact of the bandwidth selector method on the performance of the density estimator.

Figure~\ref{fig03} illustrates the 25\%, 50\%, 75\% and 99\% highest density level sets of the histogram (long dashed green line, top panels), transformed kernel density estimates (dashed red line, bottom panels) and standard kernel density estimates (dot-dashed blue line, bottom panels) in comparison with the  target distribution (solid black line).
Visually, the transformed kernel estimator performs extremely well -- it is able to identify and describe most of the features of the target densities as it's contours follow the target contours very closely. 
In contrast, the blocky, discrete nature of the histogram estimator makes it difficult to discern the nature of the underlying target,
and the standard kernel estimator is clearly unable to capture the features of the tail density as accurately as the transformation kernel estimator, displaying many spurious bumps in the tail.

\begin{figure}[tb]
\centering 
\setlength{\tabcolsep}{3pt}
\begin{tabular}{@{}ccc@{}}
Target BIL & Target ANL  & Target HR \\
\includegraphics[width=0.32\textwidth]{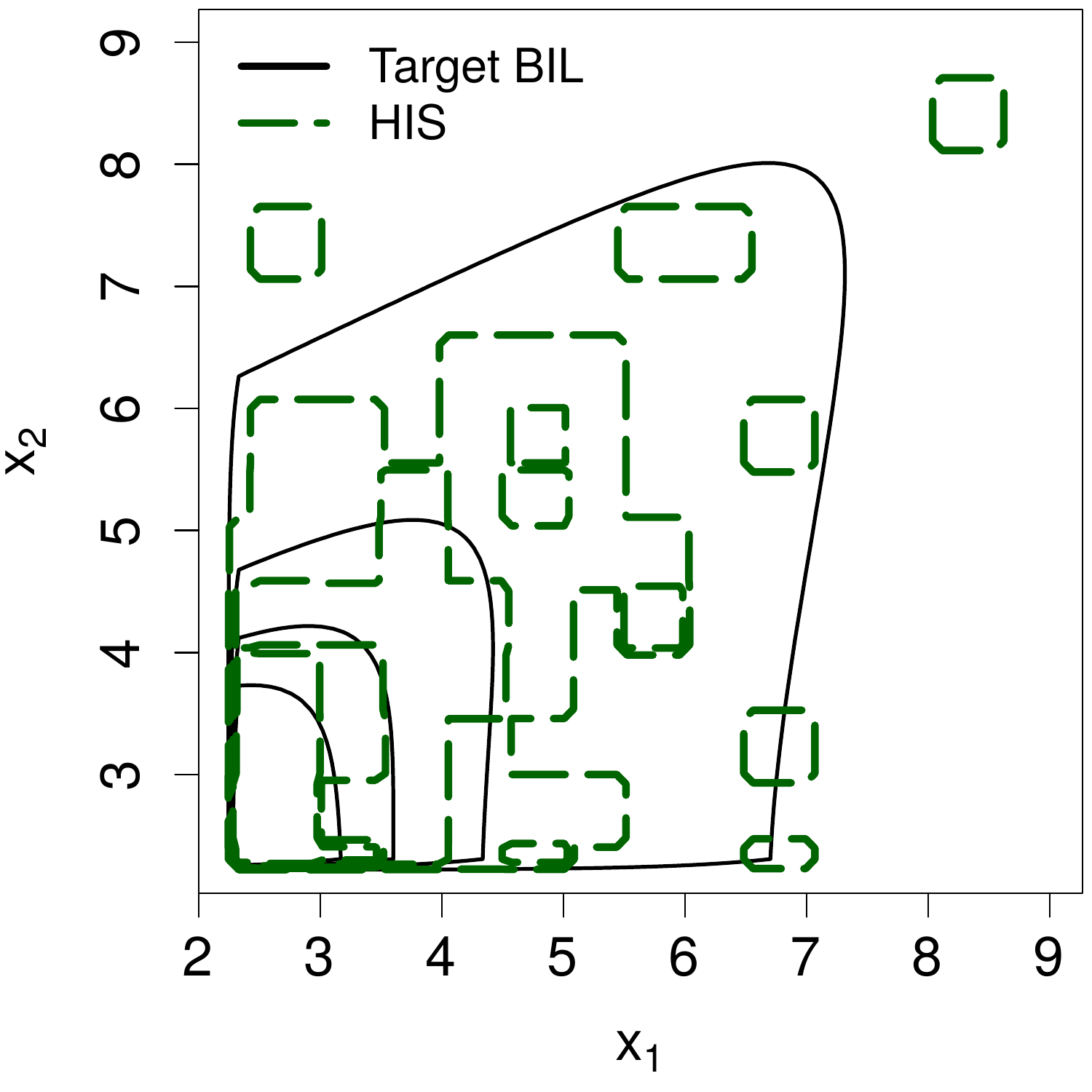} &
\includegraphics[width=0.32\textwidth]{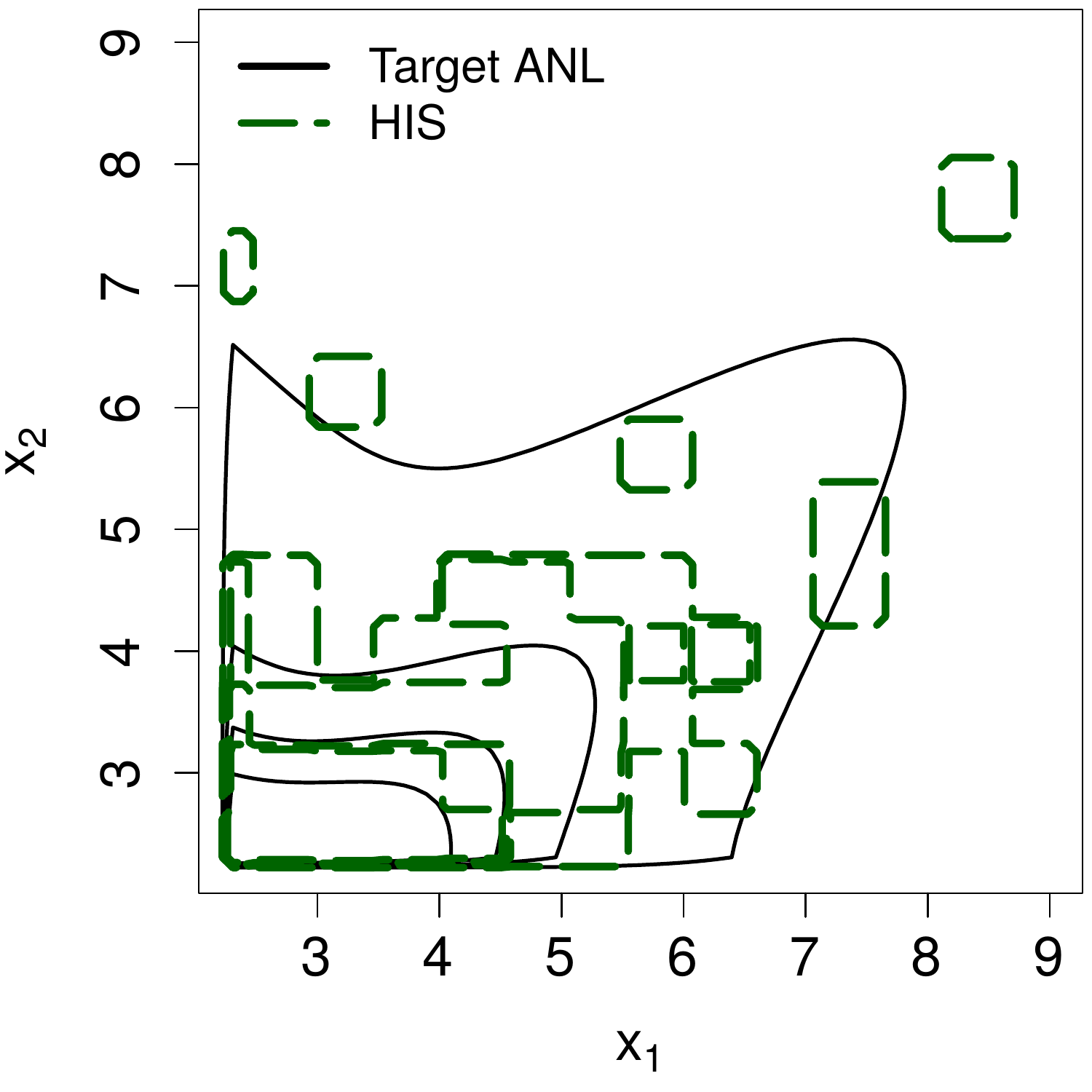} &
\includegraphics[width=0.32\textwidth]{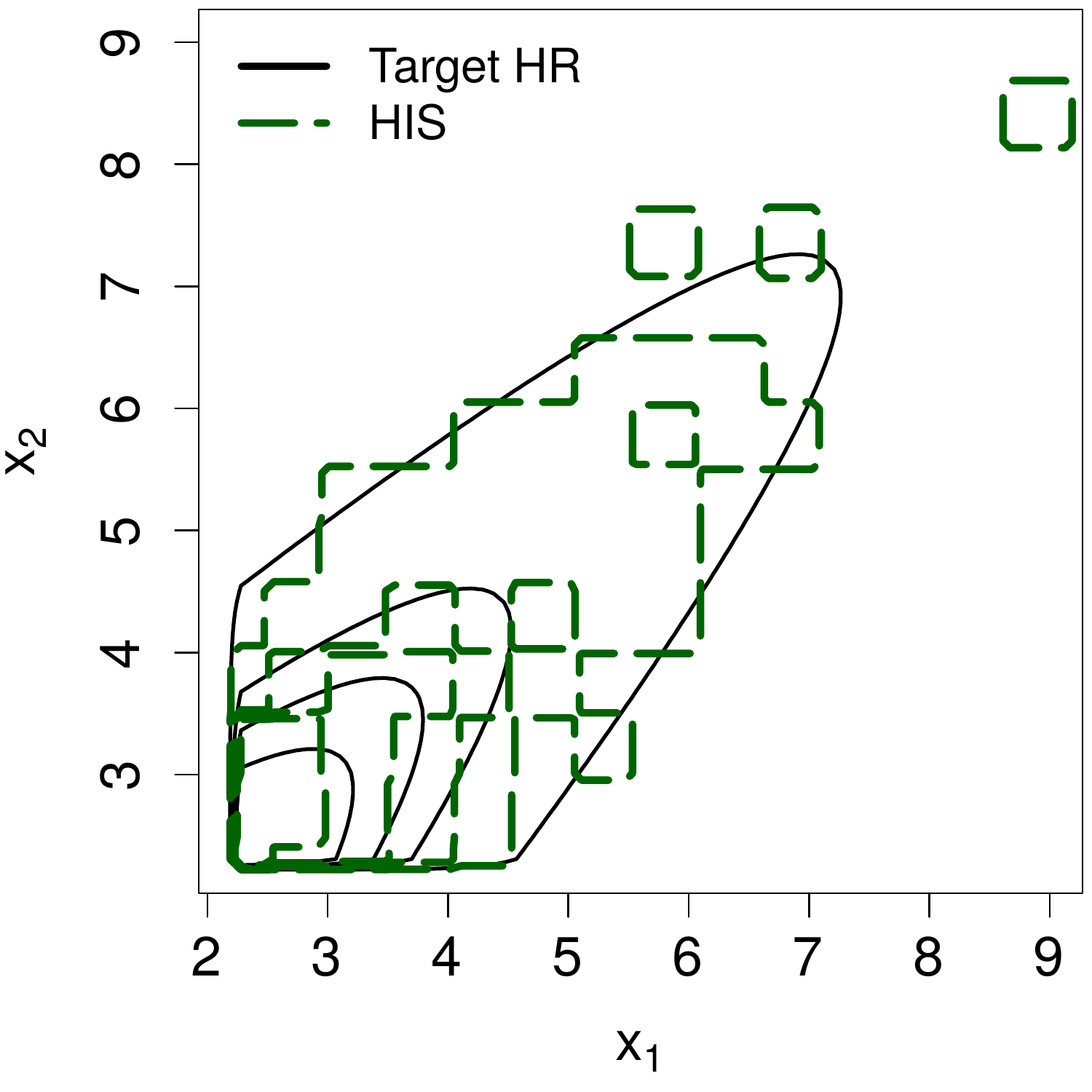} \\
\includegraphics[width=0.32\textwidth]{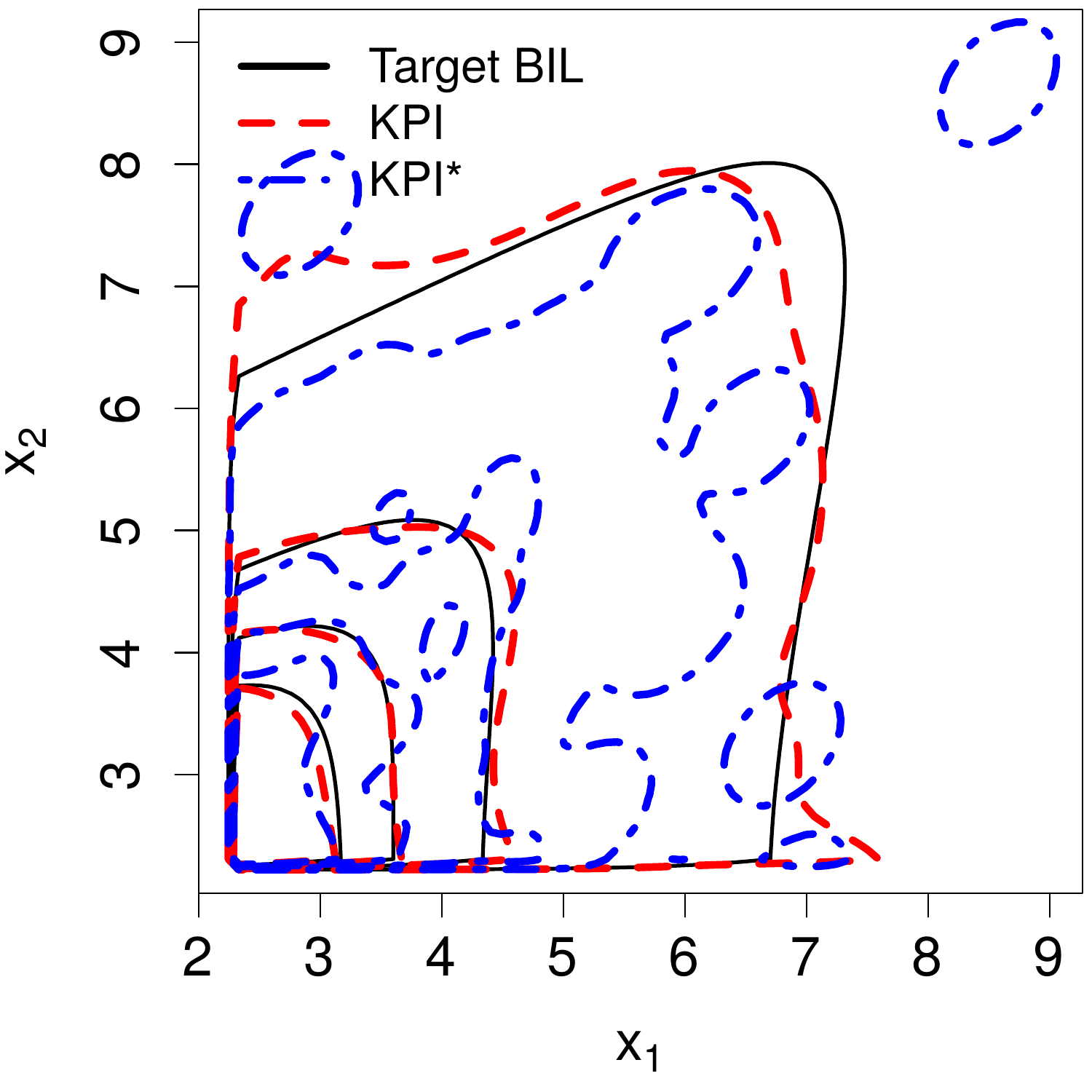} &
\includegraphics[width=0.32\textwidth]{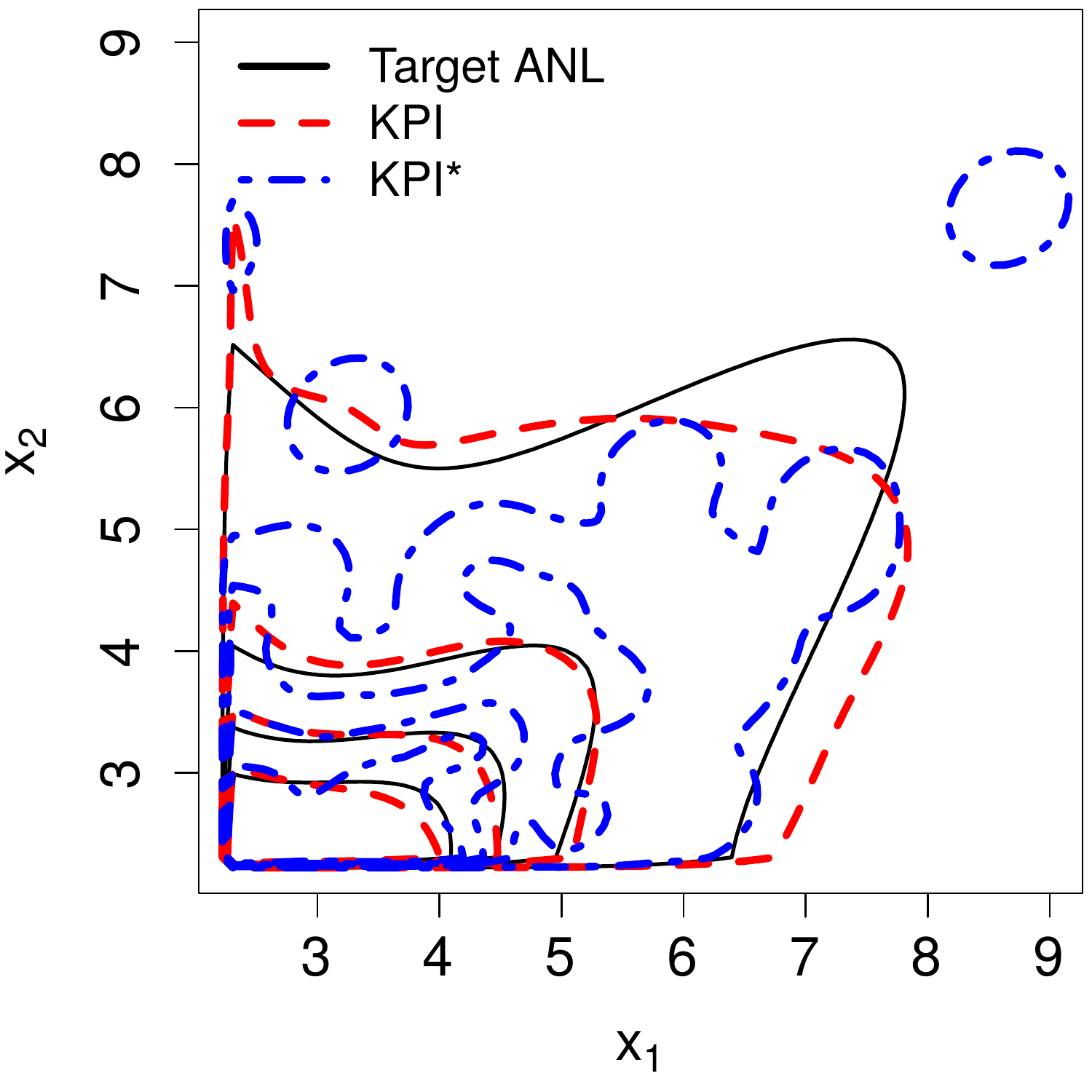} &
\includegraphics[width=0.32\textwidth]{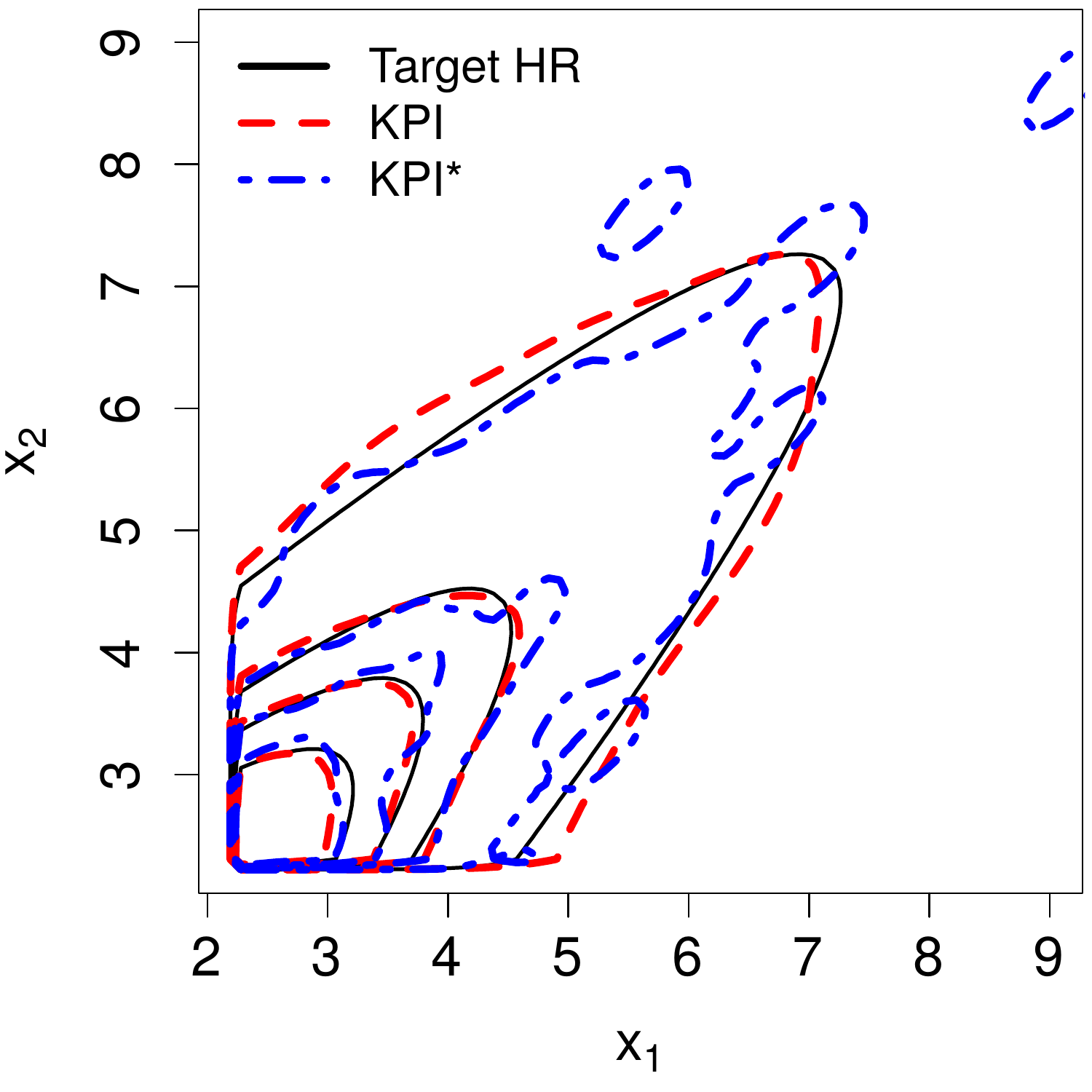}
\end{tabular} 
\caption{\label{fig03}\small Non-parametric estimators of the bivariate tail density when the target density is  bilogistic (BIL), asymmetric negative logisitic (ANL) and H\"{u}sler-Reiss (HR). Sample size is $n=200$.
Bilogistic ($\alpha=0.8$, $\beta=0.52$), asymmetric negative logistic (dependence parameter $=1.3$, asymmetry parameter $=(0.2,0.7)$) and H\"{u}sler-Reiss (dependence parameter $=2.4$) target quantiles are represented by a solid black line.
[Top panels] The histogram estimator $\tilde{f}_{\bX^{[\bu]}}$ with a normal scale bin width (HIS) is represented by a long dashed green line, [bottom panels] the transformation    kernel estimator $\hat{f}_{\bX^{[\bu]}}$ with plug-in bandwidth estimator (KPI) by the short dashed red line and the standard kernel estimator $\hat{f}^*_{\bX^{[\bu]}}$ with plug-in bandwidth estimator (KPI*) by the dot-dashed blue line.
}
\end{figure}

Figure~\ref{fig04} measures the $\log L_2$ performance of the parametric (BIL, ANL, HR), histogram (HIS), and transformation and standard kernel (KPI and KPI*) tail density estimators in approximating the known target distribution, for each of the target distributions considered in Figure \ref{fig03}.
As for the univariate case, the correctly specified parametric estimator of each distribution generates the smallest error. The transformation kernel density estimator produces the next most efficient estimator, with the bivariate histogram and the bivariate standard kernel performing the most poorly in each case.  
  
\begin{figure}[tb]
\centering 
\setlength{\tabcolsep}{3pt}
\begin{tabular}{@{}ccc@{}}
Target BIL & Target ANL & Target HR \\
\includegraphics[width=0.32\textwidth]{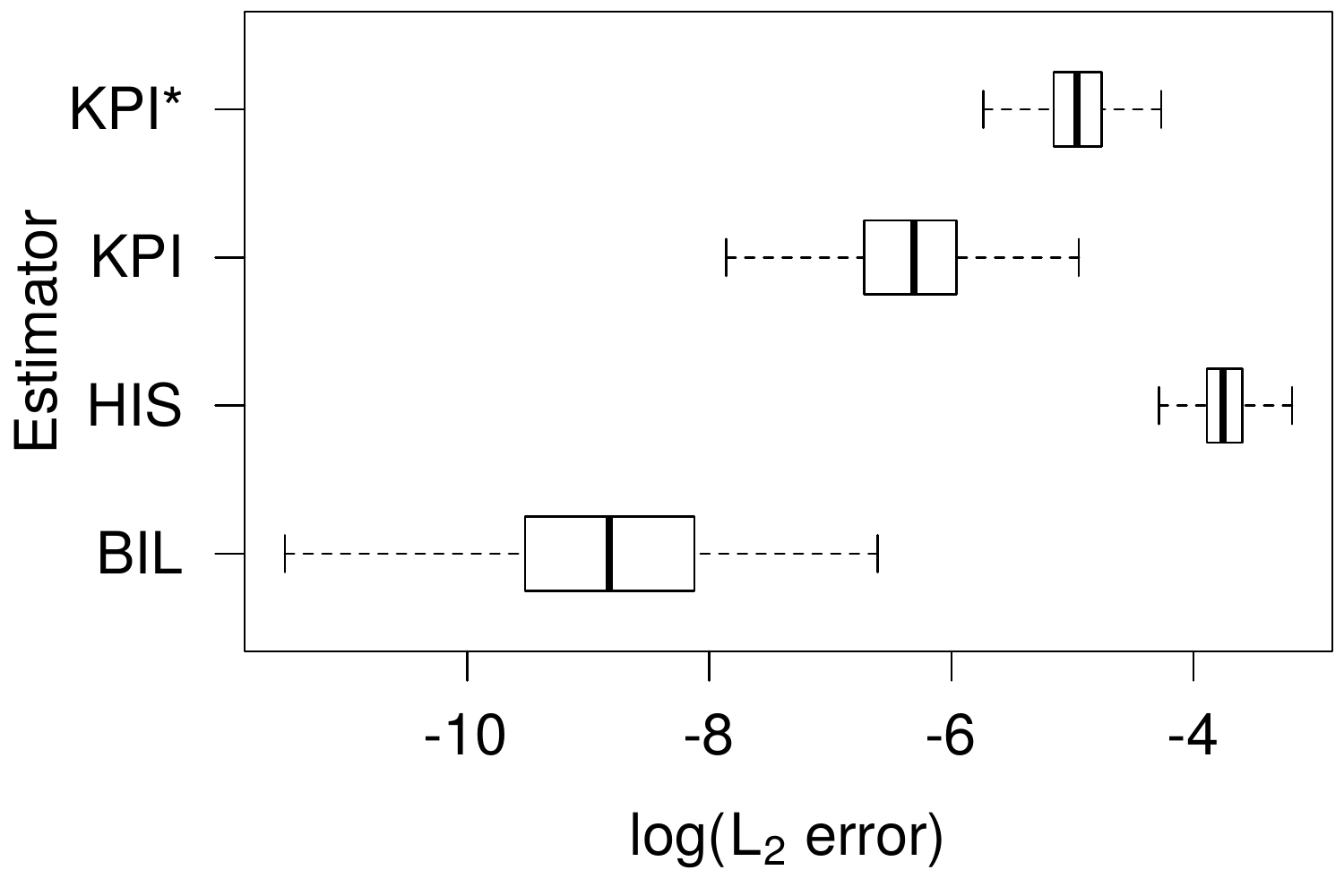} &
\includegraphics[width=0.32\textwidth]{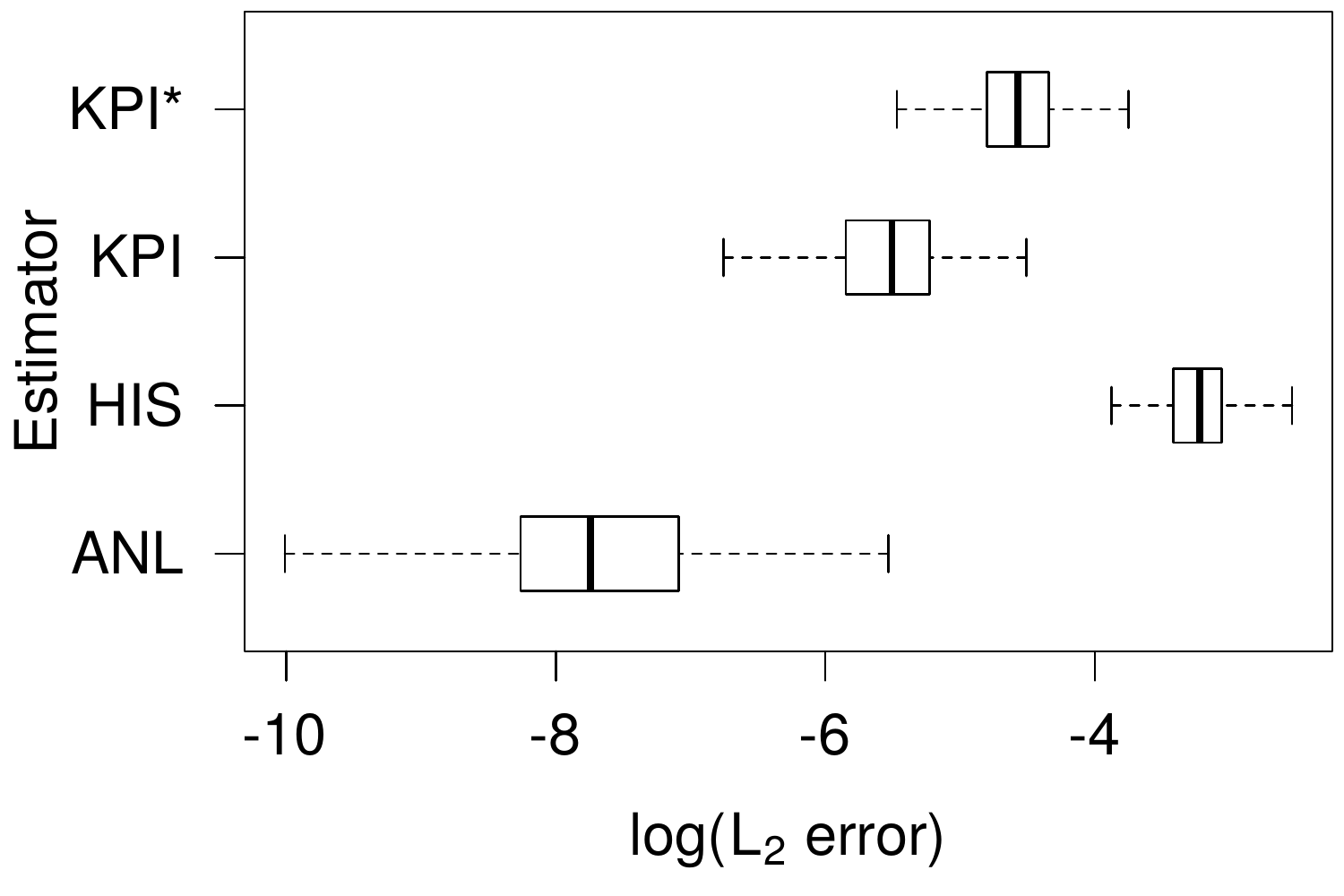} &
\includegraphics[width=0.32\textwidth]{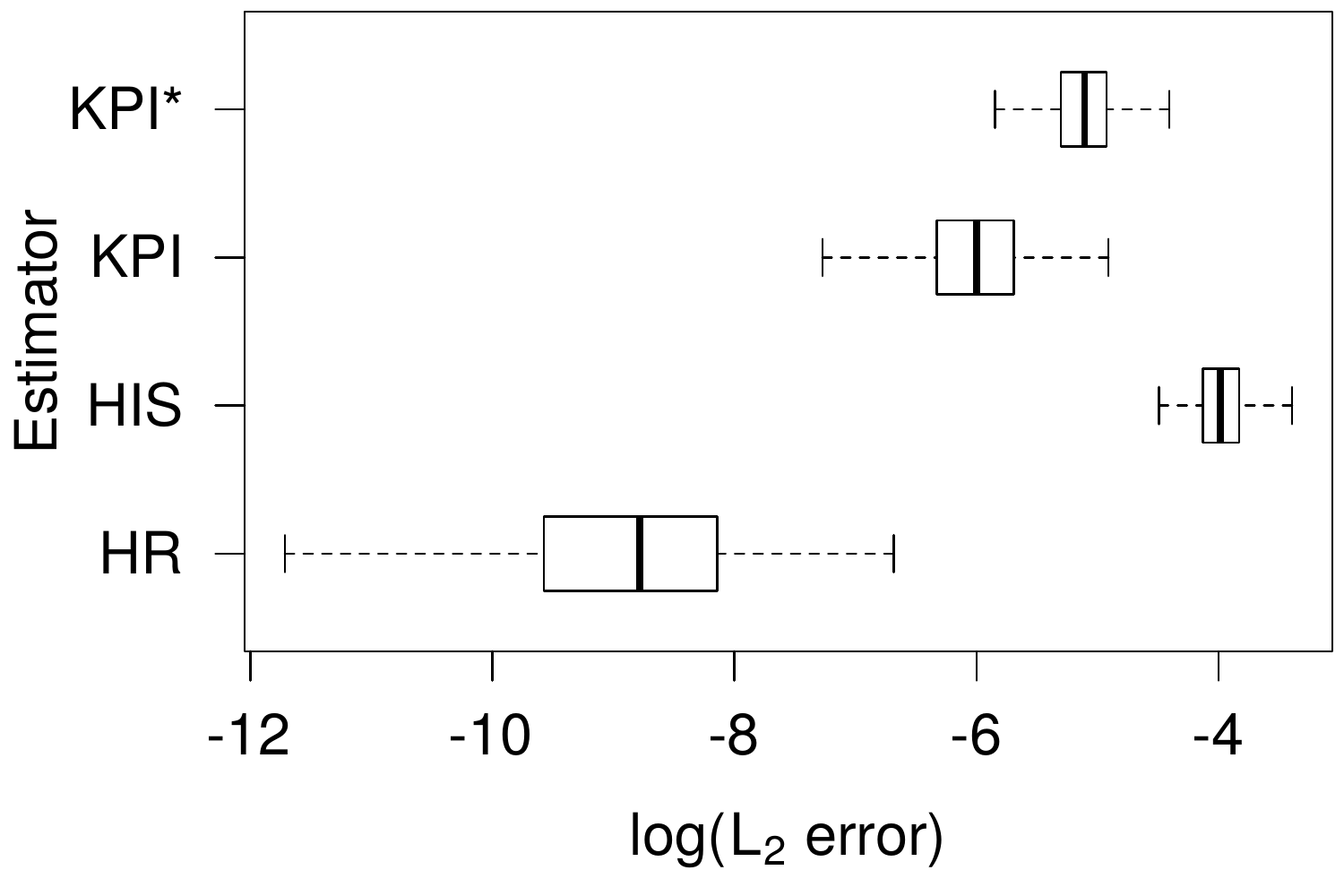} 
\end{tabular}
\caption{\label{fig04}\small Box-plots of the $\log L_2$ errors for the bivariate bilogistic (BIL), asymmetric negative logistic (ANL) and H\"{u}sler-Reiss (HR) parametric estimators, the 2-dimensional histogram (HIS), and the bivariate transformation  (KPI) and standard (KPI$^*$) kernel estimator with plug-in optimal bandwidth selector. 
True target densities are (left) the bivariate biologistic, (centre) the asymmetric negative logistic and (right) the H\"{u}sler-Reiss models. 
Box plots are based on 400 replicates of $n=4,000$ observations.
}
\end{figure}

Finally, we examine the ability of the density estimator to correctly select the true data generating model. Similarly to Table \ref{tab01}, Table \ref{tab02b} shows the proportion of times that each model was selected based on the bivariate histogram ($\tilde{T}_2$), transformation kernel ($\hat{T}_2$) and standard kernel $\hat{T}^*_2$ based tail indices using data generated from a known model, where the parametric fitted models $g_j(\bx)$ are each of BIL, ANL and HR. The bold figures indicate the estimator most often correctly selecting each target model.

In contrast to the univariate analyses, the results are mixed. 
The standard kernel based index  appears to be able to choose the correct model slightly more consistently  
than the histogram based index. However, while the transformation density estimator outperforms the standard kernel based index for H\"{u}sler-Reiss distributed data, it underperforms in other circumstances, particularly for ANL data. In general it seems that the best performing estimator for model selection is dataset dependent.

A more detailed examination of these results reveals that the transformation-based estimator is clearly the best performer in terms of its ability to estimate the true density precisely.
Table \ref{tab03b} presents the mean $L_2$ error when fitting each model to data generated under each of the BIL, ANL and HR models, taken over 400 replicate datasets. For any true model and fitted model (i.e. any row in Table \ref{tab03b}), the transformation kernel density estimate provides the most accurate density estimate (on average). This echoes the high performance findings for our density estimator in Figure \ref{fig04}.

For a given true model, and for a specified density estimator, the bold figure indicates the fitted model that is chosen most often (on average) in terms of minimising the $L_2$ error. Thus for e.g. BIL data, the BIL model is likely to be selected most often, regardless of the choice of density estimator. (Note that as these are mean values, there is some overlap of the distribution of $L_2$ errors within each density estimator, which ultimately produces the proportions observed in Table \ref{tab02b}.) This is also the case, on average for HR data tending to choose the HR model most often for each estimator.
However, for ANL data, the transformation density estimator $\hat{T}_2$ will select the BIL model most often (resulting in the low 0.15 correct classification rate in Table \ref{tab02b}), even though it is by far the better estimator of the ANL density (with a mean $L_2$ score of 0.004, compared to 0.040 and 0.011), simply because this estimator is also a slightly closer match to the fitted BIL model in this case.
In general this suggests that while the transformation based kernel density estimator clearly outperforms both the standard kernel and histogram based density estimators in terms of the quality of the tail density estimation, care should be taken when using these estimators in a model selection scenario, particularly for models in more than one dimension.

\begin{table}[tb]
\setlength{\tabcolsep}{4pt}
\centering
\begin{tabular}{@{\extracolsep{4pt}}lccccccccc@{}}
Target & $\tilde{T}_2$ & $\hat{T}_2$ & $\hat{T}^*_2$ \\
\hline
BIL & 0.69 &  0.62 & {\bf 0.74} \\
ANL & 0.83 & 0.15 & {\bf 0.85} \\
HR & 0.89 & {\bf 0.99} & 0.81 \\
 \hline
\end{tabular}

\caption{\small Proportion of 400 simulated datasets from each known target distribution (BIL bilogistic, ANL asymmetric negative logistic, and HR H\"{u}sler-Reiss) that are identified as coming from each of these distributions by having the smallest tail index value, as a function of nonparametric density estimator. Bold text indicates the highest proportion for each target  model. Nonparametric density estimators are the bivariate histogram ($\tilde{T}_2$), the transformed kernel ($\hat{T}_2$) and the standard kernel ($\hat{T}^*_2$).
}
\label{tab02b}
\end{table}

\begin{table}[tb]
\setlength{\tabcolsep}{4pt}
\centering
\begin{tabular}{@{\extracolsep{4pt}}llccccccccc@{}}
True Model & Fitted Model & $\tilde{T}_2$ & $\hat{T}_2$ & $\hat{T}^*_2$ \\
\hline
                                  & BIL & {\bf 0.024} &  {\bf 0.002} & {\bf 0.007} \\
BIL                            & ANL & 0.026 & 0.003 & 0.009 \\
                                  & HR & 0.030 & 0.004 & 0.012 \\
\hline
                                  & BIL & 0.045 &  {\bf 0.002} & 0.015 \\
ANL                            & ANL & {\bf 0.040} & 0.004 & {\bf 0.011} \\
                                  & HR & 0.060 & 0.012 & 0.030 \\
\hline
                                  & BIL & 0.021 &  0.007 & 0.008 \\
HR                            & ANL & 0.022 & 0.007 & 0.009 \\
                                  & HR & {\bf 0.019} & {\bf 0.003} & {\bf 0.006} \\
 \hline
\end{tabular}
\caption{\small Mean $L_2$ errors of the non-parametric estimators for 400 simulated datasets from each known true target distribution (BIL bilogistic, ANL asymmetric negative logistic, and HR H\"{u}sler-Reiss), compared to each parametric fitted model.  Nonparametric density estimators are the bivariate histogram ($\tilde{T}_2$), the transformed kernel ($\hat{T}_2$) and the standard kernel ($\hat{T}^*_2$). Bold text highlights the minimum $L_2$ error for each estimator, indicating the fitted model most often selected.
}
\label{tab03b}
\end{table}

\section{Exploratory data analysis of climate models}
\label{sec:GCM}

\citet{perkins2007,perkins2013} previously 
used univariate histogram density estimators for both visualisation and model selection
to evaluate the ability of global climate models (GCMs) to simulate extreme temperatures (minima and maxima) over Australia.
The models which they considered are the climate models assessed by the Intergovernmental Panel on Climate Change (IPCC) Fourth Assessment Report (AR4) to investigate changes in temperature extremes.
A well-known challenge for these models is to be able to accurately project extreme temperatures \citep{perkins2007,perkins2013,sillmann2013a, sillmann2013b,cowan2014,fischer2013}.

Following earlier work, 
\citet{perkins2013} developed a univariate tail index (see Section \ref{sec:model}, equation (\ref{eq:model_hist})) which evaluates the amount of overlap between a model-predicted distributional tail, $g_i$, and the distribution of the extreme observed data.
This index reflects the discrepancy between two distributional tails, whereby a model perfectly fitting 
the observed data has zero score, and increasing scores imply an increasing lack-of-fit 
of the model to the observed data. 
Unlike for the simulated parametric models in Section~\ref{sec:num_res_univ} and~\ref{sec:num_res_mult}, there is no closed form for the density function $g_i, i=1, \dots, M$, 
to characterise the data values generated by the climate models.   
\citet{perkins2013} replaced the unknown target density $g_i$ with a histogram $\tilde{g}_i$, based on model generated data, when comparing
to the histogram of the observed data $\tilde{f}_{X^{[u]}}$ in Equation~\eqref{eq:model_hist}. I.e. they used the index 
$\tilde{T}_1(\tilde{g}_i)$ to determine the most appropriate model.
Because of this difference with the model selection analysis in Sections~\ref{sec:num_res_univ} and~\ref{sec:num_res_mult}, there is reason to believe that this procedure is more reliable in model selection terms, as the comparison is between two data-based tail density estimators, and it is accordingly likely that the better the density estimator, the more credible the comparison between the two datasets will be.

We extend this previous histogram estimator-based analysis by considering a wider 
and more modern ensemble of global climate models
than those in \citet{perkins2013}, as well as exploring alternatives to 
$\tilde{T}_1(\tilde{g}_i)$ as the model selection criterion. 
Here we use $M=22$ climate models participating in the World Climate Research Programme's 5th phase Coupled Model Intercomparison Project \cite[CMIP5; see][]{flato2013}, 
which currently underpin global and regional climate projections of extremes \cite[e.g.][]{sillmann2013b}.  
The choice of models was based on the availability of daily maximum and minimum temperature data for the historical experiment 
\cite[$\sim$1860-2005 ; see][]{taylor2012}.
Other targeted temperature extreme evaluation studies on the CMIP5 ensemble have found generally well-simulated changes in observed trends of specific indices \cite[e.g.][]{sillmann2013a, flato2013}. However unlike this study, no consideration has been given to the full underlying distribution of extremes.

The observed data sample are the daily observed maximum temperatures for Sydney, Australia, from 01/01/1911 to 31/12/2005 yielding  a sample of $n=34,699$ observations. 
Observations were obtained from the Australian Water Availability Project dataset \citep{jones2009}, a gridded product covering all of Australia. 
All AR4 climate models were run to generate data in this same time frame, and the GCM grid box in which Sydney is located was extracted. 
The threshold determining the extreme maximum temperatures is the 95\% upper quantile $u=30.98^\circ \rm{C}$.
Additionally, note that the climate models are physical, not statistical, and run their own climate. 
Hence when ran for long enough their properties of non stationarity are very clear. 
Furthermore, they are forced via anthropogenic climate emissions, which induce a highly non stationary climate.

Table~\ref{tab_univ_real_data} displays the modified \cite{perkins2013} histogram-based tail indices, $\tilde{T}_2(\tilde{g}_i)$, the transformation kernel density estimator based index, $\hat{T}_2(\hat{g}_i)$ and the GPD based tail index $\check{T}_2(\check{g}_i)$ for ten out of the 22 models. 
Note that $\hat{T}_2(\hat{g}_i)$ implements the transformation kernel estimator for both the observed data ($\hat{f}_{X^{[u]}}$) and the GCM generated data (within $\hat{T}_2$). 
 The bold figures indicate the four best performing models (out of 22) for each tail index. In this one-dimensional analysis, both histogram- and transformation kernel density-based estimators strongly identify the same two models (i.e. with lowest tail index): MPI-ESM-LR and MPI-ESM-MR, as best describing the observed univariate extremes. 
All three tail indices share models MPI-ESM-MR and CNRM-CMS in their top four best models to simulate moderate extremes.
The three tail indices also have eight models in common out of their top ten.

\begin{table}[tb]
\centering
\begin{tabular}{@{\extracolsep{4pt}}ccccc@{}}
& \multicolumn{4}{c}{Model selection index} \\
Model & $\tilde{T}_2(\tilde{g})$ & $\hat{T}_2(\hat{g})$ & $\check{T}_2(\check{g})$  \\
\hline
	CanESM2 & 0.0042 & {\bf 0.0006} & 0.0015 \\
	CMCC-CESM & 0.0055 & 0.0039 & 0.0009 \\
	CMCC-CM & 0.0053 & 0.0031 & 0.0011 \\
	CNRM-CMS & \bf{0.0033} &  \bf{0.0005} & \bf{0.0003} \\
         HadGEM2-CC & 0.0060 & 0.0036 & \bf{0.0004} \\
         HadGEM2-ES & 0.0039 & 0.0018 & \bf{0.0002} \\
         MIROC5 & \bf{0.0037} & 0.0009 & 0.0020 \\
         MPI-ESM-LR & {\bf 0.0029} & {\bf 0.0003} & 0.001 \\
         MPI-ESM-MR & {\bf 0.0018} & {\bf 0.0002} & {\bf 0.0005} \\
         MPI-ESM-P & 0.0063 & 0.0030 &  0.0063 \\
\end{tabular}
\caption{\small
Univariate histogram- $\tilde{T}_2(\tilde{g}_i)$, kernel- $\hat{T}_2(\hat{g}_i)$ and GPD-based $\check{T}_2(\check{g}_i)$ tail index scores, based on histogram $\tilde{g}_i$,
kernel $\hat{g}_i$ and GPD $\check{g}_i$ density estimators, for the moderately extreme maximum temperatures produced by the twenty-two AR4 climate models. 
The models displayed are the ten best performing models in one dimension. Bold figures indicate the four best performing models under each model selection index.}
\label{tab_univ_real_data}
\end{table}

Figure~\ref{fig_univ_real_data} illustrates both the tail density estimators and qq-plots for the common top performing models across all tail indices.
The histogram-, kernel- and GPD-based  estimates are represented by the solid, dashed and dotted lines respectively, whereas the observed and GCM data are denoted by black and grey lines. 
For both models (CNRM-CMS and MPI-ESM-LR), each of the three density estimates of the
simulated data closely follow their respective density estimator of the observed data. 
The quality of the density estimates is also evident in the qq-plots. Here the smoother transformation kernel density estimator (dashed lines) is able to find a better match between observed and GCM model data than the histogram (solid lines) for both models, with MPI-ESM-MR providing a better overall fit (in particular for large quantiles).
The quality of the fit provided by the transformation kernel and GPD density estimators appears to be very similar.

\begin{figure}[tb]
\centering
\setlength{\tabcolsep}{2pt}
\begin{tabular}{@{}cccc@{}}
\multicolumn{2}{c}{CNRS-CMS} & \multicolumn{2}{c}{MPI-ESM-MR} 
\\
\includegraphics[width=0.24\textwidth]{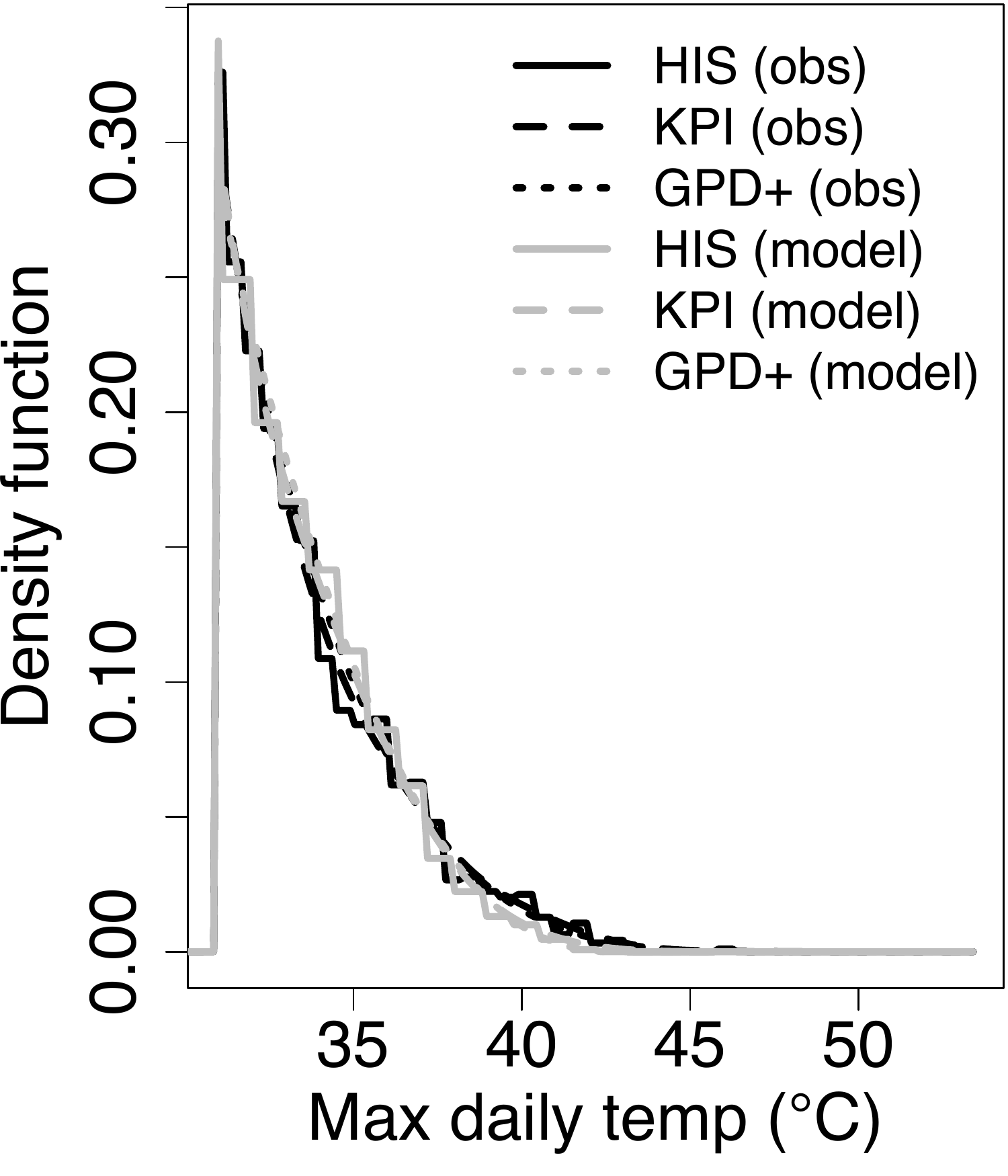} & 
\includegraphics[width=0.24\textwidth]{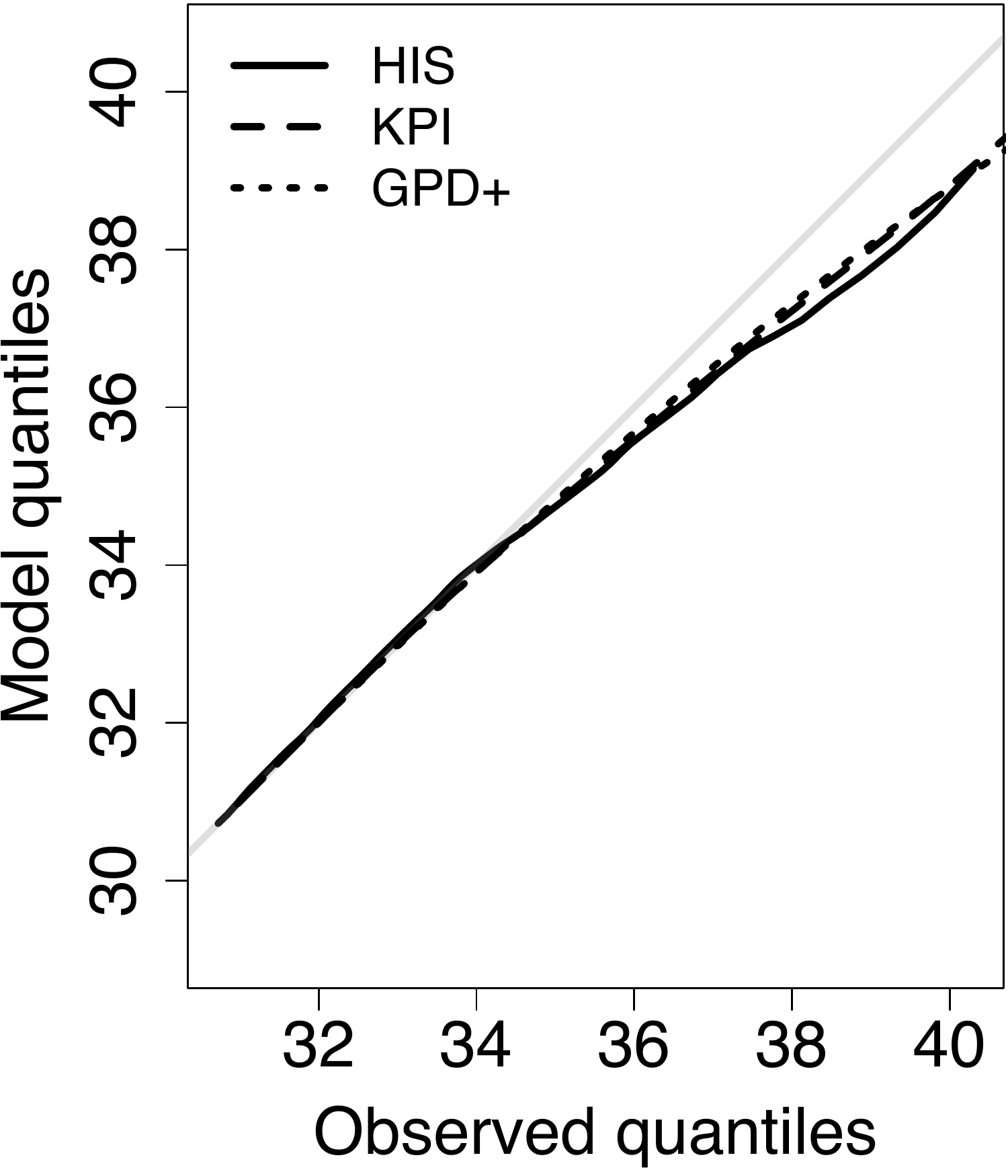} &
\includegraphics[width=0.24\textwidth]{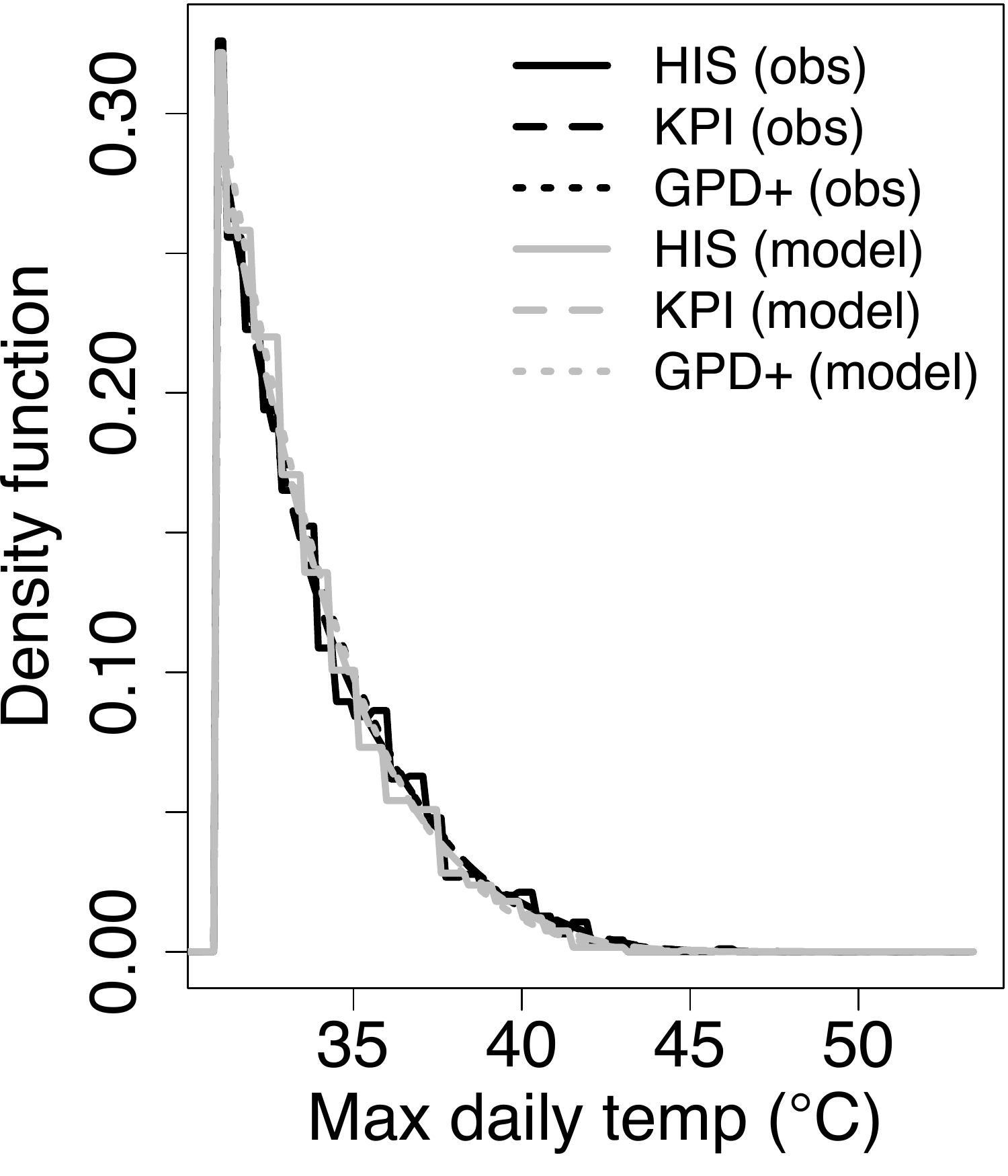} &
\includegraphics[width=0.24\textwidth]{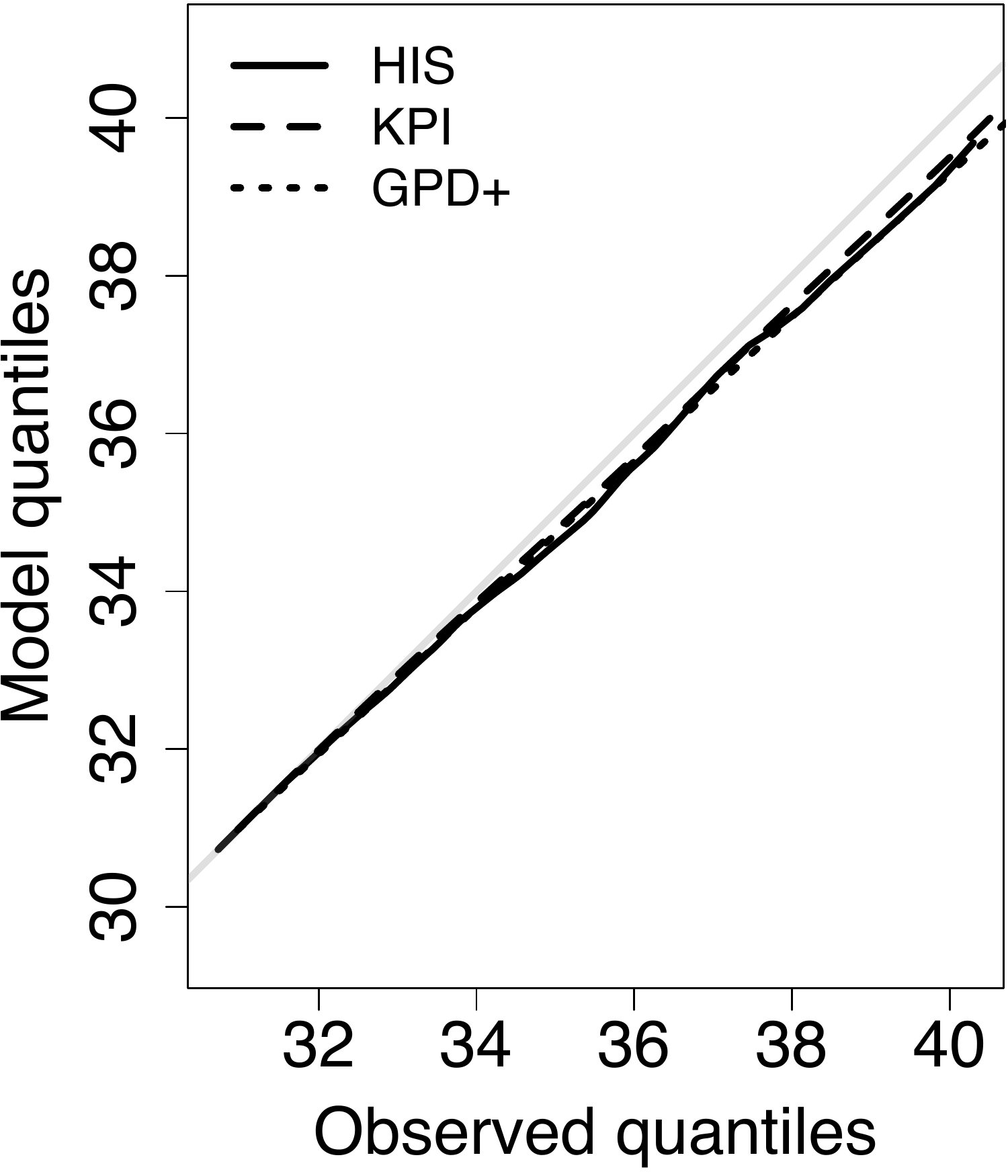} 
\end{tabular}
\caption{\small Histogram ($\tilde{f}_{X^{[u]}}$), transformation kernel ($\hat{f}_{X^{[u]}}$) and GPD ($\check{f}_{X^{[u]}}$) based estimators of the tail densities for two of the best AR4 models: (left to right) 
 CNRM-CMS and MPI-ESM-MR. 
 Histogram estimators (HIS) are denoted by solid lines, kernel plug-in estimators  (KPI) by dashed lines and GPD estimators (GPD+) by dotted lines. Observed data (obs) is illustrated in black and GCM data (model) in grey. 
}
\label{fig_univ_real_data}
\end{figure}

For a bivariate analysis, we consider the upper tail of pairs of maximum and minimum temperatures over the same time period, in order to investigate which of the climate models can predict joint extremes. (The largest minimum temperatures are important in understanding the duration and severity of heatwaves.) The threshold for the maximum temperatures are the 90\% marginal upper quantiles 
$\bu=( 28.77^\circ\rm{C}, 18.07^\circ\rm{C})^\top$.  
Table \ref{tab_biv_real_data} presents the same information as Table \ref{tab_univ_real_data} but for the bivariate data (without the GPD-based tail index).

Here, both model selection indices selecting the same best four models CNRMS-CMS, IPSL-CM5B-LR, MPI-ESM-LR and MPI-ESM-P.
Two of these were already identified in Table~\ref{tab_univ_real_data} for their ability to simulate moderately large univariate extremes in comparison with the observed data.
In particular, the CNRM-CMS model is clearly identified by both indices, achieving the lowest tail index scores, and has, along with the MPI-ESM-LR model,  the best ability to simulate moderately large bivariate extremes.

\begin{table}[tb]
\centering
\begin{tabular}{@{\extracolsep{4pt}}crr@{}}
& \multicolumn{2}{c}{Model selection index} \\
Model & $\tilde{T}_2(\tilde{g})$ & $\hat{T}_2(\hat{g})$ \\
\hline
	CMCC.CM & 0.0149 & 0.0092 \\
	CNRM-CMS & {\bf 0.0076} &  {\bf 0.0039} \\
         HadCM3 & 0.0123 & 0.0066 \\
         HadGEM2.ES & 0.0133 & 0.0.0081 \\
         IPSL-CM5A-LR & 0.0123 & 0.0060 \\
         IPSL-CM5B-LR & {\bf 0.0079} & \bf{0.0041} \\
         MIROC5 & 0.0124 & 0.0070 \\
         MPI-ESM-LR & \bf{0.0096} & \bf{0.0048} \\
         MPI-ESM-MR & 0.0106 & 0.0059 \\
         MPI-ESM-P & \bf{0.0083} & {\bf 0.0040} \\
\end{tabular}
\caption{\small
Bivariate histogram- $\tilde{T}_2(\tilde{g}_i)$ and transformation kernel-based $\hat{T}_2(\hat{g}_i)$ tail index scores, based on
histogram $\tilde{g}_i$
and kernel $\hat{g}_i$ density estimators,  for the  extreme (minimum, maximum) temperatures produced by the twenty-two AR4 climate models. 
The models displayed are the ten best performing models. Bold figures indicate the four best performing models under each model selection index.
}
\label{tab_biv_real_data}
\end{table}

\begin{figure}[tb]
\centering
\setlength{\tabcolsep}{2pt}
\begin{tabular}{@{}cccc@{}}
\multicolumn{2}{c}{CNRM-CMS} & \multicolumn{2}{c}{MPI-ESM-P}  \\
\includegraphics[width=0.24\textwidth]{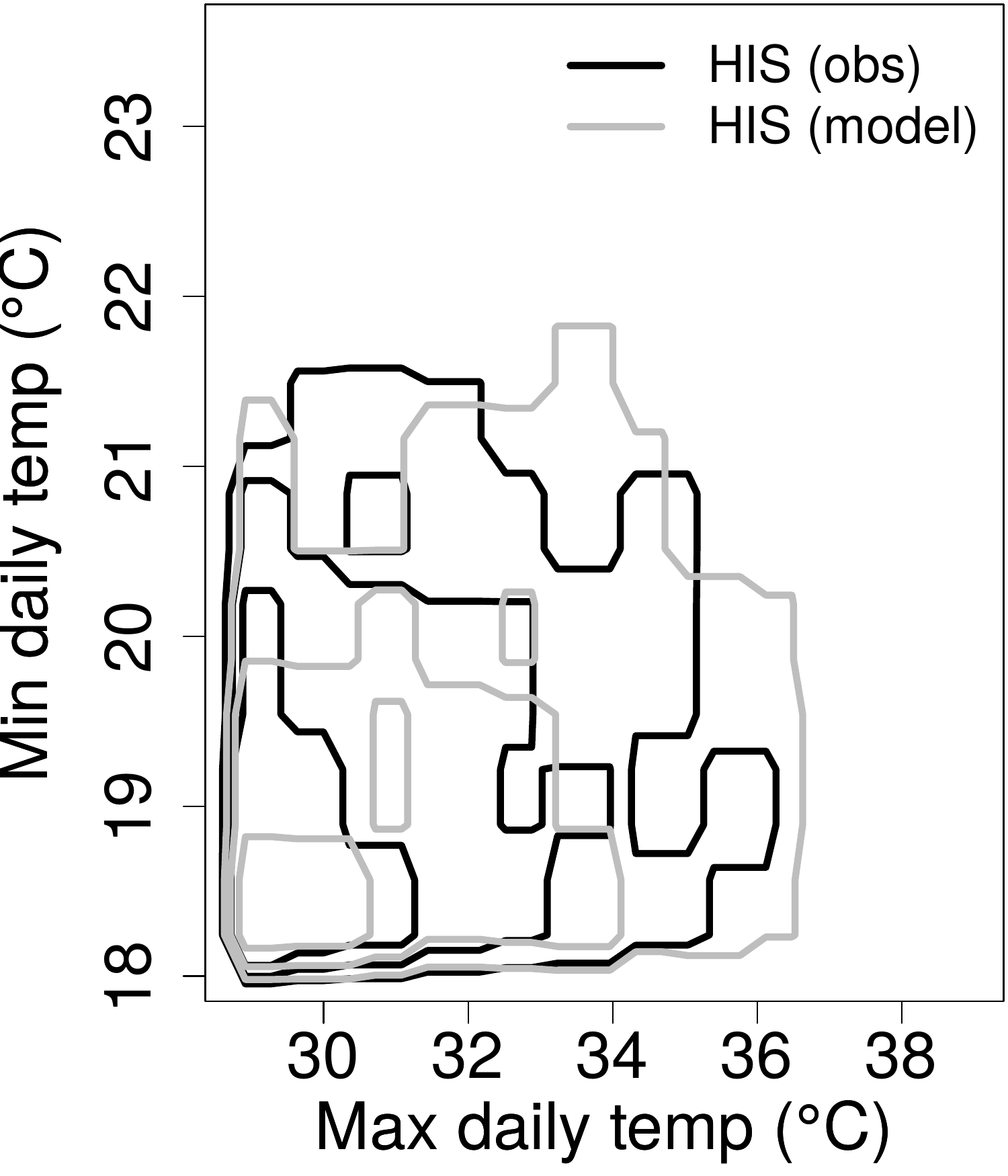} &
\includegraphics[width=0.24\textwidth]{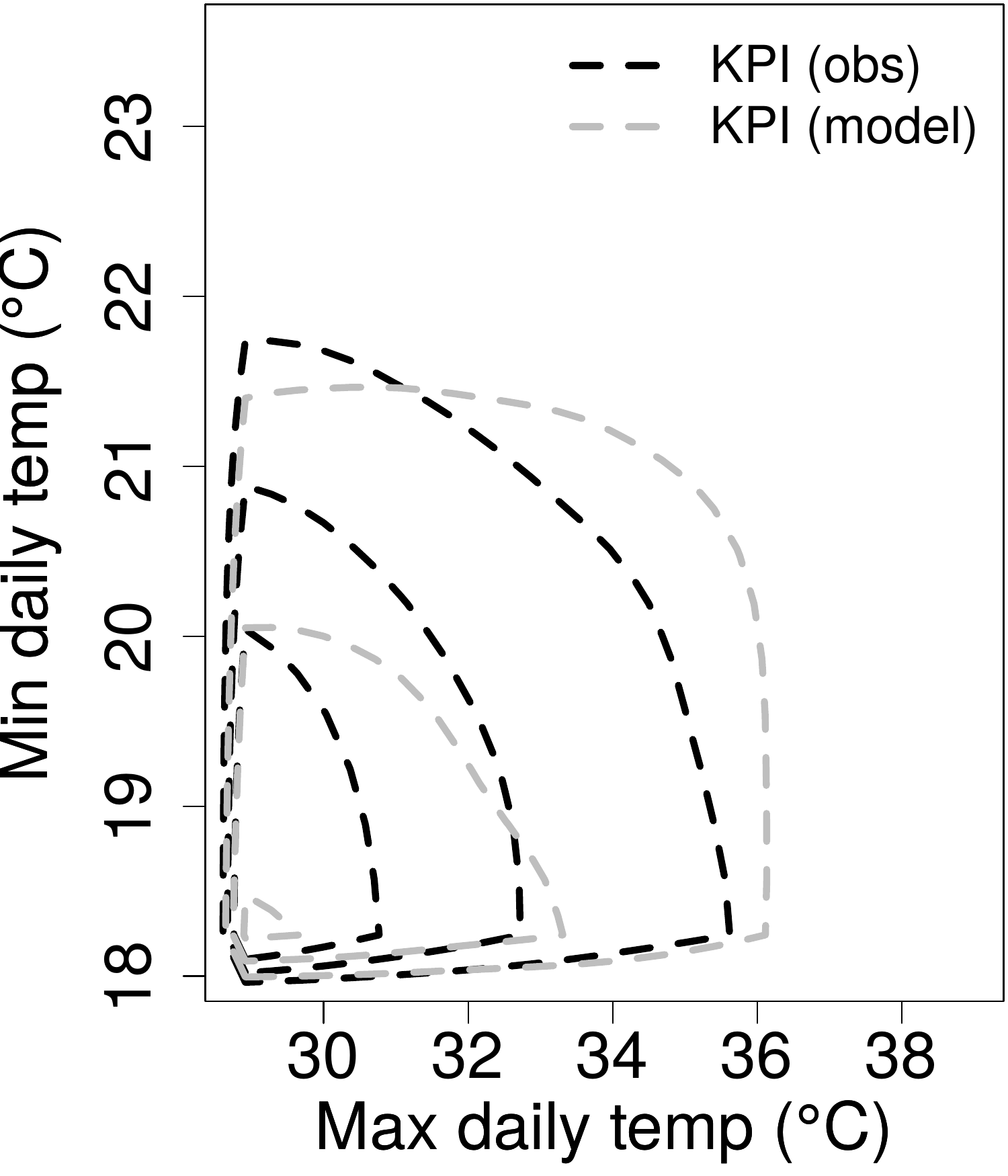} &
\includegraphics[width=0.24\textwidth]{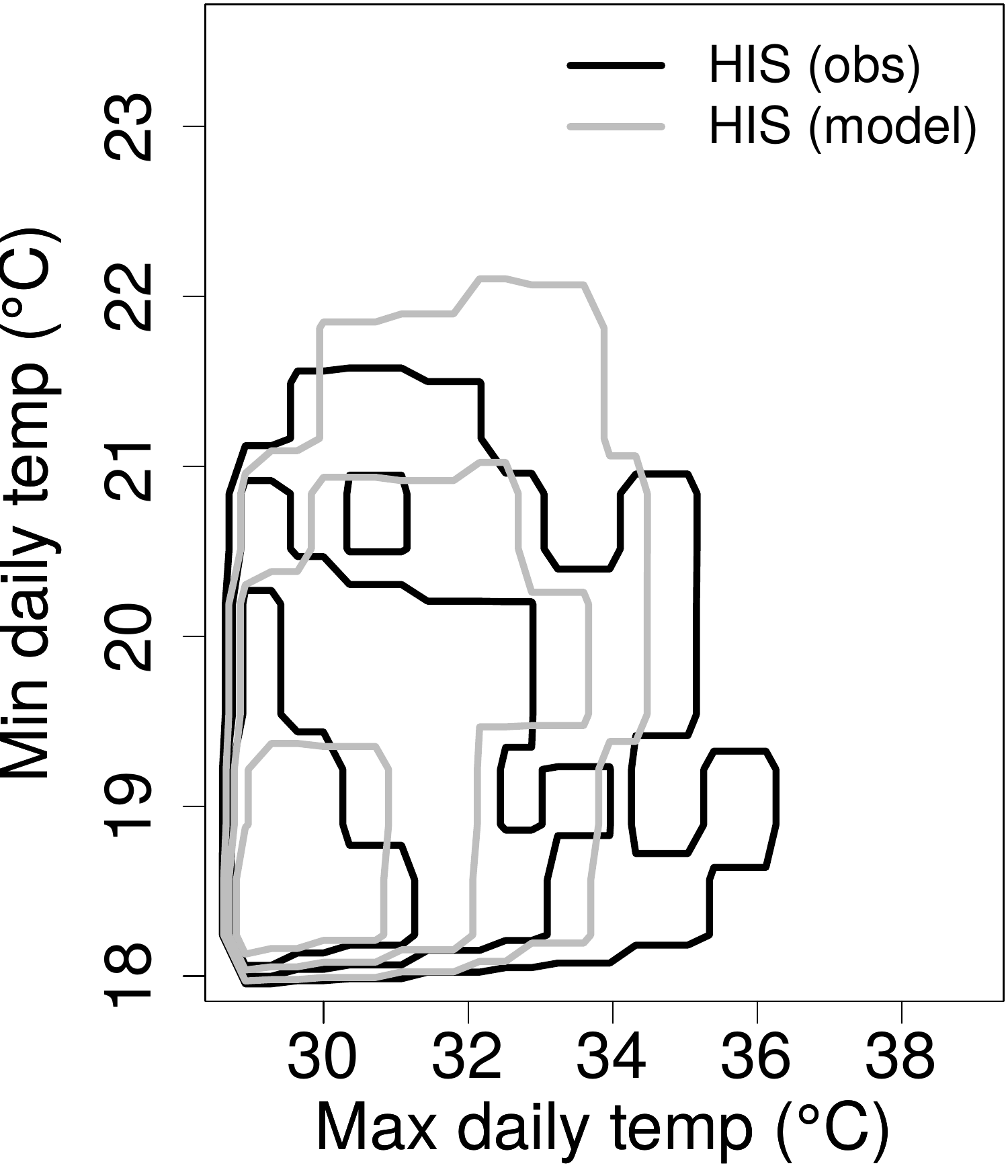} & 
\includegraphics[width=0.24\textwidth]{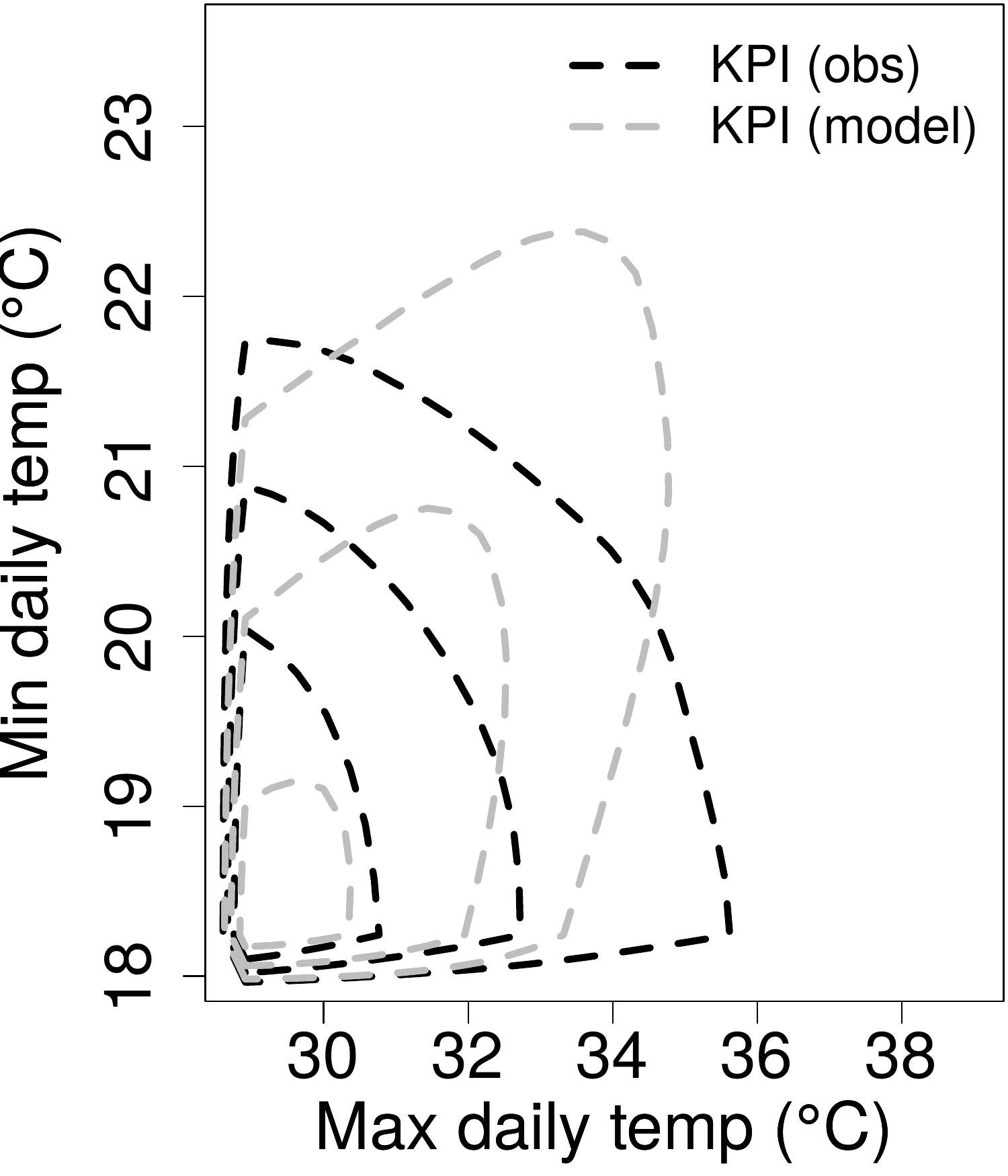}
\end{tabular}
\caption{\small 
Histogram (left panel) and kernel estimators (right panel) of the tail densities for the CNRM-CMS and MPI-ESM-P climate models.
Histograms estimators (HIS) are denoted by solid lines and kernel plug-in estimators  (KPI) by dashed lines. Observed data (obs) is illustrated with black lines and climate model data (model) by grey lines. Kernel estimator contours indicate the 25\%, 50\% and 75\% highest density level sets.
}
\label{fig_biv_real_data}
\end{figure}

Figure \ref{fig_biv_real_data} illustrates the bivariate tail density estimators for the CNRM-CMS and MPI-SM-P models, two of the top-performing bivariate models.
Similarly to Figure \ref{fig_univ_real_data}, histograms are shown by solid lines, and transformation kernel density estimators by dashed lines. Observed and GCM data are represented by black and grey lines respectively. It is immediately apparent that the kernel-based density estimates are visually much cleaner, and easier to evaluate than their histogram counterparts. 
In particular, it is immediate that the CNRM-CMS model is a visually better match to the observed data than MPI-ESM-P.

Both of these top performing GCM models appear to simulate the extremes of the maximum temperatures quite well ($x$-axis, Figure  \ref{fig_biv_real_data}) (agree with this comment?), however each of them underestimate the magnitude of  the minimum temperatures ($y$-axis).
This indicates two possibilities. 
That the kernel estimators need to be further refined at the boundary or, perhaps more likely,   that
 physical parameters in the GCMs need to be revised for a more realistic simulation of minimum temperature extremes. 
 While minimum temperatures are physically simpler for a climate model to simulate than maximum temperature \citep{perkins2007}, the mis-representation of the observed temperature distribution is a well-known issue for GCMs, which is at least in part explained by their coarse resolution \citep{seneviratne2012}. For example, dynamically downscaled regional climate models that are run at finer resolutions for a limited spatial domain can offer some improvement in the simulation of extreme temperatures \citep[e.g.][]{seneviratne2012, vautart2013, perkins2014}.

The goal of this exploratory data analysis is to propose
feasible geophysical models which adequately describe the observed temperature data maxima and minima.
As these geophysical models are expressed as a set of differential equations, their overall statistical properties are not well-known. Our estimates and visualisations of the tail densities of these
geophysical models are a first step in elucidating their statistical properties,
upon which more sophisticated data analysis drawn from extreme value theory can be  subsequently applied.

\section{Discussion}
\label{sec:discussion}

Nonparametric density estimation is a useful exploratory data analysis tool for  extremes. In this article we have
introduced a non-parametric kernel estimator for the analysis of 
the tail density of univariate and multivariate data by applying a logarithm transformation 
to account for the heavy tails and boundedness of moderately extreme value data samples. 
Our proposed tail density estimator does not suffer from the usual boundary problems
associated with kernel estimators. It is also robust in terms of the choice of
the extreme value threshold.

Our theoretical results (centred on Theorem \ref{thm:bias-fX}) indicate the good performance of this transformation kernel density estimator in the extreme tails. Numerical illustrations of its performance were given in Sections \ref{sec:numerical} and \ref{sec:GCM}.
This tail density estimator provides visually useful representations of extreme sample behaviour compared to, say, histogram estimators -- consider the contrast in visual clarity  between the histogram and kernel estimators illustrated in Figure \ref{fig_biv_real_data}. Furthermore, it can reliably be incorporated into existing diagnostic and performance measures, such as the tail index of \cite{perkins2013}.

There is, of course, scope for further development and analysis of these ideas. 
As our proposed tail density estimator is decoupled from 
the threshold estimation, a promising avenue for amelioration would be the inclusion
of more sophisticated threshold estimators than the simple
quantile thresholds we have utilised. 

Throughout we have constructed our kernel density estimates based on Gaussian kernels. In general, this means that the extreme tail behaviour of the kernel-based density estimators is necessarily the tail behaviour of the kernel $K$ mapped through  the inverse of the transformation $\bt$. This implies that if the tail behaviour of the sample does not correspond to that of the transformed kernel, then the kernel density estimate will poorly represent the true behaviour of the observed data distribution beyond the range of the observed data. A natural approach to resolving this problem could be to adapt the form of $K$ and $\bt$ to directly correspond to the (estimated) tail behaviour of the observed sample.

Additional improvements could be obtained by incorporating a local polynomial adjustment \cite[e.g.][]{geenens2014} to the boundary to improve over the transformation kernel approach, although here our primary interest is in the behaviour of the upper tail. Similarly, while the logarithm transformation is widely used due it conveniently mapping a semi-infinite interval to the real line, alternative transformations could be considered. Possibilities include the shifted power family of  \citet{wand1991} 
$$t(x) = \begin{cases}
(x+\lambda_1)^{\lambda_2} \operatorname{sign}(\lambda_2) &  \lambda_2 \neq 0 \\
\log(x+\lambda_2)  & \lambda_2 = 0 
\end{cases}
$$  
where 
the log-transformation fixes $\lambda_1=0, \lambda_2=-u_0$, and so other values may lead to better estimation, and the richer family of transformations proposed by \citet{wand1991}.
and those posited by \citet{geenens2014}. 

Finally, the performance of this transformation-based kernel density estimator is limited by the performance of standard kernel density estimator methods. In particular, its performance will decline as the dimension of the random vector $\bX$ increases.   While this is unavoidable, if one wishes to perform kernel density estimation in this setting, it is important that it is implemented as efficiently as possible. The results presented in this article provide one step towards achieving this.

\section*{Acknowledgements}

SEPK and SAS are supported by the Australian Centre of Excellence for Climate System Science (CoECSS), SEPK by the Australian Research Council (ARC) grant number DE140100952, SAS by the Australian Centre of Excellence for Mathematical and Statistical Frontiers in Big Data, Big Models, New Insights (ACEMS, CE140100049), and BB and SAS by the ARC Discovery Project scheme (DP160102544).

\section*{Appendix -- Proofs}

The below assumptions will be used to establish the optimality 
properties of our transformation kernel density estimators.  These assumptions
are usually expressed for random variables with unbounded support,
which in our case is the transformed variable $\bY=\bt(\bX)$. 
This set of conditions do not form a minimal set, but they serve as a convenient starting point
to state our results. 

\begin{enumerate}[label= (A\arabic*)]
\item The $d$-variate density $f_\bY$ is continuous, square integrable and ultimately 
monotone for  
all element-wise partial second derivatives.
\item  The $d$-variate kernel $K$ is a positive, symmetric, square integrable p.d.f. such that 
$\int_{\mathbb{R}^d} \by\by^\top K(\by) d\by = m_2 (K)\bI_d$ where $m_2 (K)$ is finite and $\bI_d$ is the $d \times d$ identity matrix.
\item The bandwidth matrix $\bH = \bH (n)$ forms a sequence of symmetric and positive definite matrices such that $n^{-1} |\bH|^{-1/2}$ and every element of $\bH$ approaches zero as $n\rightarrow \infty$.
\end{enumerate}

\noindent The proof of Theorem~\ref{thm:bias-fX} requires  Lemma~\ref{lem:fhat} (below) which 
establishes the minimal rate of MISE convergence of 
$\hat{f}_\bY$. This result has already been established \cite[e.g.][]{wand1992}, however we include 
details of a proof using an alternative notation for fourth order derivatives
of a multivariate function via four-fold Kronecker product, which is simpler to code than tensors. 

\begin{lemma}
\label{lem:fhat} 
Suppose that the conditions (A1--A3) hold. The MISE of the 
the kernel density estimator with unbounded data support $\hat{f}_\bY$ is 
\begin{align*}
\MISE \lbrace\hat{f}_\bY (\cdot; \bH) \rbrace
&= \big [\tfrac{1}{4}  m_2^2(K) ( \vec^\top \bH \otimes \vec^\top \bH) \bpsi_{\bY,4}
 + n^{-1} |\bH|^{-1/2} R(K) \big] \lbrace 1 + o(1) \rbrace.
\end{align*}
where $\bpsi_{\bY,4} = \int_{\mathbb{R}^d} \D^{\otimes 4} f_\bY (\by) f_\bY (\by) \mathrm{d}\by$.
\end{lemma}

\begin{proof}[Proof of Lemma~\ref{lem:fhat}]
The expected value of $\hat{f}_\bY$ is
\begin{align*}
\E \hat{f}_\bY (\by; \bH)
=\E K_{\bH} (\by - \bY) 
= \int_{\mathbb{R}^d} K_{\bH} (\by - \bw) f \left(\by\right) d\bw 
=  K_{\bH} * f_\bY\left(x\right)
\end{align*} 
where $*$ denotes the convolution operator between two functions. 
Asymptotically, using a Taylor series expansion, we have 
\begin{align*}
\E \hat{f}_\bY (\by; \bH) &=\int_{\mathbb{R}^d} |\bH|^{-1/2} K(\bH^{-1/2}(\by - \bw))f_\bY (\by)d\bw\\
&=\int_{\mathbb{R}^d} K(\bw)f_\bY(\by-\bH^{1/2}\bw)d\bw\\
&=\int_{\mathbb{R}^d} K(\bw)[f_\bY(\by)- \bw^\top \bH^{1/2} \D f_\bY(\bx) + \tfrac{1}{2} 
\bw^\top \bH^{1/2} \D^2 f_\bY(\by) \bH^{1/2} \bw] \lbrace 1 + o(1) \rbrace \, d\bw \\
&=[ f_\bY(\by)+ \tfrac{1}{2} \int_{\mathbb{R}^d} K(\bw) \tr (\bw\bw^\top\bH \D^2 f_\bY(\by))  \, d\bw ] \lbrace 1 + o(1) \rbrace\\
&=[ f_\bY(\by) + \tfrac{1}{2}  m_2(K) \tr (\bH \D^2 f_\bY(\by)) ] \lbrace 1 + o(1) \rbrace.
\end{align*} 
This allows us to write the bias of $\hat{f}_\bY (\by; \bH)$ as
\begin{align*}
\E \hat{f}_\bY (\by; \bH) - f_\bY(\by) 
= \tfrac{1}{2}  m_2(K) \tr (\bH \D^2 f_\bY(\by)) \lbrace 1 + o(1) \rbrace. 
\end{align*}
For the variance, we have $\Var \hat{f}_\bY (\by; \bH) = n^{-1} \E [K_\bH(\by - \bY)^2] - 
n^{-1} [ \E K_\bH(\by - \bY)]^2$. 
The second term is given by the above, so we are required to evaluate
\begin{align*}
\E \hat{f}_\bY (\by; \bH)^2 
&= \E[K_\bH (\by - \bY)^2] 
=\int_{\mathbb{R}^d}  K_{\bH} (\by - \bw)^2 f_\bY\left(\bw\right)d\bw\\
&=\int_{\mathbb{R}^d} |\bH|^{-1} K(\bH^{-1/2}(\by - \bw))^2 f_\bY\left(\bw\right)d\bw\\
&=\int_{\mathbb{R}^d} |\bH|^{-1/2} K(\bw)^2 f_\bY(\by-\bH^{-1/2}\bw) d\bw \\
&=|\bH|^{-1/2} f_\bY(\by)\int_{\mathbb{R}^d} K(\bw)^2 d\bw \lbrace 1 + o(1) \rbrace \\
&=|\bH|^{-1/2} f_\bY(\by) R(K) \lbrace 1 + o(1) \rbrace.
\end{align*} 
Thus the variance term is
\begin{align*}
\Var \lbrace \hat{f}_\bY (\by; \bH)  \rbrace 
= n^{-1} \lbrace |\bH|^{-1/2} f_\bY(\by) R(K)- [f_\bY(\by) + \tfrac{1}{2}  m_2(K) \tr (\bH \D^2 f_\bY(\by))]^2\rbrace \lbrace 1 + o(1) \rbrace.
\end{align*}
Since $\bH \rightarrow 0$ then $|\bH|^{-1/2}$ dominates both the constant term  $f_\bY(\by)$ and the $\tr (\bH)$
term so we can write 
\begin{align*}
\Var \lbrace \hat{f}_\bY (\by; \bH)  \rbrace
= n^{-1} |\bH|^{-1/2} f_\bY(\by) R(K) \lbrace 1 + o(1) \rbrace.
\end{align*}
The integrated square bias (ISB) is then
\begin{align*}
\ISB \lbrace \hat{f}_\bY (\cdot; \bH)\rbrace  
&= \int_{\mathbb{R}^d} \Bias^2 \hat{f}_\bY (\by; \bH) \mathrm{d}\by 
= \int _{\mathbb{R}^d}\tfrac{1}{4}  m_2^2(K) \tr^2 (\bH \D^2 f_\bY(\by)) \mathrm{d}\by \lbrace 1 + o(1) \rbrace \\
&= \tfrac{1}{4}  m_2^2(K) \int_{\mathbb{R}^d} \tr^2 (\bH \D^2 f_\bY(\by)) \mathrm{d}\by \lbrace 1 + o(1) \rbrace \\
&= \tfrac{1}{4}  m_2^2(K) ( \vec^\top \bH \otimes \vec^\top \bH) \bpsi_{\bY,4} \lbrace 1 + o(1) \rbrace,
\end{align*}
and similarly the integrated variance (IV) is
\begin{align*}
\IV \lbrace \hat{f}_\bY (\cdot; \bH) \rbrace 
&= \int _{\mathbb{R}^d}n^{-1} |\bH|^{-1/2} f_\bY(\by) R(K) \mathrm{d}\by \lbrace 1 + o(1) \rbrace \\
&= n^{-1} |\bH|^{-1/2} R(K) \int_{\mathbb{R}^d} f_\bY(\by) \mathrm{d}\by \lbrace 1 + o(1) \rbrace \\
&= n^{-1} |\bH|^{-1/2} R(K) \lbrace 1 + o(1) \rbrace,
\end{align*}
using the integrability assumptions in conditions (A1) and (A2).
Hence we obtain the result as 
$\MISE \lbrace\hat{f}_\bY (\cdot; \bH) \rbrace
= \ISB \lbrace\hat{f}_\bY (\cdot; \bH) \rbrace + \IV \lbrace\hat{f}_\bY (\cdot; \bH) \rbrace$.
\end{proof}

\begin{proof}[Proof of Theorem~\ref{thm:bias-fX}]
Let  $\by =\bt(\bx) = (\log(x_{1d}), \dots, \log(x_d))^\top$, and inversely  
$\bx = \bexp (\by) = (\exp(y_1), \dots, \exp(y_d))^\top$.  
The Jacobian is $|\mathbf{J}_\bt(\bx)| = 1/(x_1 \cdots x_d) = \exp (-|\by|)$
where $|\by| = y_1 + \cdots + y_d$. Thus $f_\bY(\by) = 1/|\mathbf{J}_\bt(\bx)| f_\bX(\bx) = \exp(|\by|) f_\bX(\bexp (\by))$. 
This representation will allow us to determine the Hessian matrix of $\D^2 f_\bY$ since it 
the previous lemma shows that it is a crucial element in $\MISE \lbrace \hat{f}_\bY(\cdot; \bH)\rbrace$. 

To evaluate derivatives of $f_\bY(\by)$ with respect to $\by$, we require
the following preliminary differentials:
\begin{align*}
d \exp(|\by|) 
&= \D [\exp(y_1 + \cdots + y_d )] ^\top d \by 
= [\exp(y_1), \dots, \exp(y_d)]^\top d \by 
= \bexp(\by)^\top d\by, \\
d \bexp(\by) 
&= [d \exp(y_1), \dots, d \exp (y_d)] 
= [\exp(y_1) d y_1, \dots, \exp(y_d) d y_d] = \Diag (\bexp(\by)) d \by, 
\end{align*}
and \begin{align*}
d \Diag (\bexp(\by))
&= \Diag (d \exp(y_1), \dots, d \exp(y_d) = \Diag (\exp(y_1) d y_1, \dots, \exp(y_d) d y_d) \\
&= \Diag (\bexp(\by)) \Diag (d\by). 
\end{align*} 
where $\Diag (\ba)$ is the diagonal matrix whose elements are $\ba$. 
It can be decomposed as $\Diag (\ba) = \sum_{j=1}^d \be_j^\top \ba \be_j \be_j^\top$ in terms of $\be_j$, 
the $j$-th elementary $d$-vector which is all zero except for 1
at the $j$-th element. 
So then
\begin{align*}
d \vec \Diag (\bexp(\by)) 
&= \sum_{j=1}^d \vec (\Diag (\bexp(\by)) \be_j \be_j^\top) \be_j^\top d\by.
\end{align*}

The differential of $f_\bY$ is 
\begin{align*}
d f_\bY(\by) &= (d \exp (|\by|)) f_\bX(\bexp(\by)) + \exp(|\by|) d f_\bX (\bexp(\by)) \\
&= f_\bX(\bexp(\by)) \bexp(\by)^\top d\by + \exp(|\by|) \D f_\bX (\bexp(\by))^\top d \bexp(\by) \\
&= f_\bX(\bexp(\by)) \bexp(\by)^\top d\by + \exp(|\by|) \D f_\bX (\bexp(\by))^\top \Diag (\bexp(\by)) d \by
\end{align*}
which implies that the first derivative is 
\begin{align*}
\D f_\bY(\by) 
&= f_\bX(\bexp(\by)) \bexp(\by) + \exp(|\by|) \Diag (\bexp(\by)) \D_\bX (\bexp(\by)) \\
&= f_\bX(\bexp(\by)) \bexp(\by) + \exp(|\by|) [\D f_\bX (\bexp(\by))^\top \otimes \bI_d] \vec \Diag (\bexp(\by)),
\end{align*}
using the first identification table
in \citet[p.~176]{magnus1999} to  convert these differentials to derivatives. 
The second form of $\D f_\bY(\by)$ derives from the identity $\vec ({\bf A} {\bf B} {\bf C}) = ({\bf C}^\top \otimes \bI_d) \vec {\bf B}$ for conformable matrices ${\bf A}, {\bf B},
{\bf C}$.

The differential of $\D f_\bY(\by)$ is 
\begin{align*}
d \D f_\bY(\by) &= f_\bX(\bexp(\by)) d \bexp(\by) + (d f_\bX(\bexp(\by))) \bexp(\by) \\
&\quad + (d \exp(|\by|)) \Diag (\bexp(\by)) \D f_\bX (\bexp(\by))\\
&\quad + \exp(|\by|) [\D f_\bX (\bexp(\by))^\top \otimes \bI_d] d \vec \Diag (\bexp(\by)) \\
&\quad + \exp(|\by|) [(d\bexp(\by))^\top \D^2 f_\bX (\bexp(\by)) \otimes \bI_d] \vec \Diag (\bexp(\by)) \\
&= f_\bX(\bexp(\by)) \Diag (\bexp(\by)) d \by + \bexp(\by) \D f_\bX(\exp(\by))^\top \Diag (\bexp(\by)) d \by \\
&\quad + \Diag (\bexp(\by)) \D f_\bX (\bexp(\by))\bexp(\by)^\top d\by \\
&\quad + \exp(|\by|) [\D f_\bX (\bexp(\by))^\top \otimes \bI_d] 
 {\textstyle \big\lbrace \sum_{j=1}^d \vec [\Diag (\bexp(\by))  \be_j \be_j^\top] \be_j^\top d\by \big \rbrace} \\
&\quad + \exp(|\by|) [d \by^\top \Diag (\bexp(\by)) \D^2 f_\bX (\bexp(\by)) \otimes \bI_d] \vec \Diag (\bexp(\by)) .
\end{align*}
This can be simplified by noting that for $d$-vectors $\ba, \bb$,  
$$
\sum_{j=1}^d (\ba^\top \otimes \bI_d) \vec [\Diag(\bb) \be_j \be_j^\top] \be_j^\top
= \sum_{j=1}^d \Diag(\bb) \be_j \be_j^\top \ba \be_j^\top
= \Diag(\bb) \Diag(\ba)
$$%
that is, 
\begin{align*}
d \D f_\bY(\by) &= \big\lbrace f_\bX(\bexp(\by)) \Diag (\bexp(\by))  + \bexp(\by) \D f_\bX(\exp(\by))^\top \Diag (\bexp(\by))  \\
&\quad + \Diag (\bexp(\by)) \D f_\bX (\bexp(\by))\bexp(\by)^\top d\by \\
&\quad + \exp(|\by|) \Diag (\bexp(\by)) \Diag(\D f_\bX (\bexp(\by)))  \\
&\quad + \exp(|\by|) \Diag (\bexp(\by)) \D^2 f_\bX (\bexp(\by)) \Diag (\bexp(\by)) \big\rbrace d \by.
\end{align*}
This implies that the Hessian matrix of $\D f_\bY(\by)$ is 
\begin{align}
\label{eq:Hessian-fY}
\D^2 f_\bY(\by) 
&= f_\bX(\bx)  \Diag (\bx) + \bx \D f_\bX(\bx)^\top \Diag (\bx)
+ \Diag (\bx) \D f_\bX(\bx)\bx^\top  \nonumber \\
&\quad + \pi(\bx) \Diag (\bx) \Diag (\D f_\bX (\bx)) 
+ \pi(\bx)  \Diag (\bx) \D^2 f_\bX (\bx) \Diag (\bx) 
\end{align}
as $\exp(|\by|) = |\mathbf{J}_\bt(\bx)|^{-1} = x_1 x_2 \cdots x_d = \pi(\bx) $. 

Firstly using the definition of $\hat{f}_{\bX}$, its expected value is
$\E \hat{f}_{\bX} (\bx;\bH)
= |\mathbf{J}_\bt(\bx)| \E \hat{f}_\bY(\by;\bH)$ 
and its associated bias, from combining Lemma~\ref{lem:fhat} and Equation \eqref{eq:Hessian-fY}, is 
\begin{align*}
\Bias &\lbrace \hat{f}_{\bX} (\bx;\bH) \rbrace \\
&= \E \hat{f}_{\bX} (\bx;\bH) - f_{\bX} (\bx)
= |\mathbf{J}_\bt(\bx)| \Bias \lbrace \hat{f}_{\bY} (\bt(\bx) ;\bH) \rbrace \\
&= \tfrac{1}{2}  m_2(K) \pi(\bx)^{-1} \tr (\bH \D^2 f_\bY(\by)) \lbrace 1 + o(1) \rbrace \\
&= \tfrac{1}{2}  m_2(K) \pi(\bx)^{-1}  \tr \big\lbrace \bH \big[ f_\bX(\bx)  \Diag (\bx) + \bx \D f_\bX(\bx)^\top \Diag (\bx) + \Diag (\bx) \D f_\bX(\bx)\bx^\top   \\
&\quad 
+ \pi(\bx) \Diag (\bx) \Diag (\D f_\bX (\bx)) 
+ \pi(\bx) \Diag (\bx) \D^2 f_\bX (\bx) \Diag (\bx) \big] \big\rbrace 
\lbrace 1 + o(1) \rbrace \\
&=\tfrac{1}{2}  m_2(K)  \big[ \pi(\bx)^{-1}  f_\bX(\bx) \tr (\bH \Diag(\bx)) 
+ 2 \pi(\bx)^{-1} \tr (\bH \bx \D f_\bX(\bx)^\top \Diag(\bx))  \\
&\quad + \tr (\bH \Diag (\bx) \Diag (\D f_\bX (\bx))) + \tr (\bH \Diag (\bx) \D^2 f_\bX (\bx) \Diag (\bx)) \big ] \lbrace 1 + o(1) \rbrace.
\end{align*}

Similarly we have
$\E [\hat{f}_{\bX} (\bx;\bH)^2]
= |\mathbf{J}_\bt(\bx)|^2 \E [\hat{f}_\bY(\bt(\bx);\bH)^2]$, 
leading to 
\begin{align*}
\Var \lbrace\hat{f}_{\bX} (\bx;\bH) \rbrace
&= \E [\hat{f}_{\bX} (\bx;\bH)^2] - \lbrace \E \hat{f}_{\bX} (\bx;\bH) \rbrace^2
= |\mathbf{J}_\bt(\bx)|^2 \Var \lbrace \hat{f}_{\bY} (\bt(\bx) ;\bH) \rbrace  \\
&= n^{-1} |\bH|^{-1/2} R(K) |\mathbf{J}_\bt(\bx)|^2  f_\bY(\by)  \lbrace 1 + o(1) \rbrace\\
&= n^{-1} |\bH|^{-1/2} R(K) \pi(\bx)^{-1} f_\bX(\bx)\lbrace 1 + o(1) \rbrace.
\qedhere
\end{align*}
\end{proof}

\bibliographystyle{asa}
\bibliography{biblio}

\begin{thebibliography}{68}
\newcommand{\enquote}[1]{``#1''}
\expandafter\ifx\csname natexlab\endcsname\relax\def\natexlab#1{#1}\fi

\bibitem[{Abramson(1982)}]{abramson1982}
Abramson, I.~S. (1982), \enquote{On bandwidth variation in kernel estimates--a
  square root law,} \textit{Ann. Statist.}, 10, 1217--1223.

\bibitem[{Balkema and de~Haan(1974)}]{balkema1974}
Balkema, A.~A. and de~Haan, L. (1974), \enquote{Residual life time at great
  age,} \textit{Ann. Probability}, 2, 792--804.

\bibitem[{Beirlant et~al.(2004)Beirlant, Goegebeur, Teugels, and
  Segers}]{beirlant2004}
Beirlant, J., Goegebeur, Y., Teugels, J., and Segers, J. (2004),
  \textit{Statistics of Extremes: Theory and Applications}, Chichester: John
  Wiley \& Sons.

\bibitem[{Beranger and Padoan(2015)}]{beranger2015b}
Beranger, B. and Padoan, S. (2015), \enquote{Extreme Dependence Models,} in
  \textit{Extreme Value Modeling and Risk Analysis}, Chapman and Hall/CRC, pp.
  325--352--.

\bibitem[{Bowman et~al.(1984)Bowman, Hall, and Titterington}]{bowman1984}
Bowman, A.~W., Hall, P., and Titterington, D.~M. (1984),
  \enquote{Cross-validation in nonparametric estimation of probabilities and
  probability densities,} \textit{Biometrika}, 71, 341--351.

\bibitem[{Chac\'on and Duong(2010)}]{chacon2010}
Chac\'on, J.~E. and Duong, T. (2010), \enquote{Multivariate plug-in bandwidth
  selection with unconstrained bandwidth matrices,} \textit{Test}, 19,
  375--398.

\bibitem[{Charpentier and Flachaire(2015)}]{charpentier2015}
Charpentier, A. and Flachaire, E. (2015), \enquote{Log-Transform Kernel Density
  Estimation of Income Distribution,} \textit{L'Actualité Économique}, 91,
  141--159.

\bibitem[{Chen(1999)}]{chen1999}
Chen, S.~X. (1999), \enquote{Beta kernel estimators for density functions,}
  \textit{Comput. Statist. Data Anal.}, 31, 131--145.

\bibitem[{Cheng and Amin(1983)}]{cheng1993}
Cheng, R. C.~H. and Amin, N. A.~K. (1983), \enquote{Estimating parameters in
  continuous univariate distributions with a shifted origin,} \textit{J. Roy.
  Statist. Soc. Ser. B}, 45, 394--403.

\bibitem[{Christopeit(1994)}]{christopeit1994}
Christopeit, N. (1994), \enquote{Estimating parameters of an extreme value
  distribution by the method of moments,} \textit{J. Statist. Plann.
  Inference}, 41, 173--186.

\bibitem[{Coles(2001)}]{coles2001}
Coles, S. (2001), \textit{An Introduction to Statistical Modeling of Extreme
  Values}, London: Springer-Verlag.

\bibitem[{Coles and Tawn(1994)}]{coles+t94}
Coles, S.~G. and Tawn, J.~A. (1994), \enquote{Statistical methods for
  multivariate extremes: {An} application to structural design,} \textit{J.
  Roy. Statist. Soc. Ser. C}, 43, 1--48.

\bibitem[{Cowan et~al.(2014)Cowan, Purich, Perkins, Pezza, Boschat, and
  Sadler}]{cowan2014}
Cowan, T., Purich, A., Perkins, S., Pezza, A., Boschat, G., and Sadler, K.
  (2014), \enquote{More Frequent, Longer, and Hotter Heat Waves for Australia
  in the Twenty-First Century,} \textit{Journal of Climate}, 27, 5851--5871.

\bibitem[{de~Carvalho et~al.(2013)de~Carvalho, Oumow, Segers, and
  Warcho{\l}}]{decarvalho2013}
de~Carvalho, M., Oumow, B., Segers, J., and Warcho{\l}, M. (2013), \enquote{A
  {E}uclidean likelihood estimator for bivariate tail dependence,}
  \textit{Comm. Statist. Theory Methods}, 42, 1176--1192.

\bibitem[{de~Haan and Ferreira(2006)}]{dehaan2006}
de~Haan, L. and Ferreira, A. (2006), \textit{Extreme Value Theory: An
  Introduction}, New York: Springer.

\bibitem[{Drees and Huang(1998)}]{drees1998}
Drees, H. and Huang, X. (1998), \enquote{Best attainable rates of convergence
  for estimators of the stable tail dependence function,} \textit{J.
  Multivariate Anal.}, 64, 25--47.

\bibitem[{Duong and Hazelton(2003)}]{duong2003}
Duong, T. and Hazelton, M.~L. (2003), \enquote{Plug-in bandwidth matrices for
  bivariate kernel density estimation,} \textit{J. Nonparametr. Statist.}, 15,
  17--30.

\bibitem[{Duong and Hazelton(2005)}]{duong2005}
--- (2005), \enquote{Cross-validation bandwidth matrices for multivariate
  kernel density estimation,} \textit{Scand. J. Statist.}, 32, 485--506.

\bibitem[{Einmahl et~al.(2001)Einmahl, de~Haan, and Piterbarg}]{einmahl2001}
Einmahl, J. H.~J., de~Haan, L., and Piterbarg, V.~I. (2001),
  \enquote{Nonparametric estimation of the spectral measure of an extreme value
  distribution,} \textit{Ann. Statist.}, 29, 1401--1423.

\bibitem[{Einmahl et~al.(2008)Einmahl, Krajina, and Segers}]{einmahl2008}
Einmahl, J. H.~J., Krajina, A., and Segers, J. (2008), \enquote{A method of
  moments estimator of tail dependence,} \textit{Bernoulli}, 14, 1003--1026.

\bibitem[{Einmahl et~al.(2012)Einmahl, Krajina, and Segers}]{einmahl2012}
--- (2012), \enquote{An {$M$}-estimator for tail dependence in arbitrary
  dimensions,} \textit{Ann. Statist.}, 40, 1764--1793.

\bibitem[{Einmahl and Segers(2009)}]{einmahl2009}
Einmahl, J. H.~J. and Segers, J. (2009), \enquote{Maximum empirical likelihood
  estimation of the spectral measure of an extreme-value distribution,}
  \textit{Ann. Statist.}, 37, 2953--2989.

\bibitem[{Falk et~al.(2011)Falk, H{\"u}sler, and Reiss}]{falk2011}
Falk, M., H{\"u}sler, J., and Reiss, R.-D. (2011), \textit{Laws of Small
  Numbers: Extremes and Rare Events}, Basel: Birkh\"auser/Springer Basel AG,
  extended ed.

\bibitem[{Fischer et~al.(2013)Fischer, Beyerle, and Knutti}]{fischer2013}
Fischer, E.~M., Beyerle, U., and Knutti, R. (2013), \enquote{Robust spatially
  aggregated projections of climate extremes,} \textit{Nature Clim. Change}, 3,
  1033--1038.

\bibitem[{Flato et~al.(2013)Flato, Marotzke, Abiodun, Braconnot, Chou, Collins,
  Cox, Driouech, Emori, Eyring, Forest, Gleckler, Guilyardi, Jakob, Kattsov,
  Reason, and Rummukainen}]{flato2013}
Flato, G., Marotzke, J., Abiodun, B., Braconnot, P., Chou, S., Collins, W.,
  Cox, P., Driouech, F., Emori, S., Eyring, V., Forest, C., Gleckler, P.,
  Guilyardi, E., Jakob, C., Kattsov, V., Reason, C., and Rummukainen, M.
  (2013), \enquote{Evaluation of Climate Models,} in \textit{Climate Change
  2013: The Physical Science Basis. Contribution of Working Group I to the
  Fifth Assessment Report of the Intergovernmental Panel on Climate Change},
  eds. Stocker, T., Qin, D., Plattner, G.-K., Tignor, M., Allen, S., Boschung,
  J., Nauels, A., Xia, Y., Bex, V., and Midgley, P., Cambridge, United Kingdom
  and New York, NY, USA: Cambridge University Press, pp. 741--866.

\bibitem[{Gasser and M{\"u}ller(1979)}]{gasser1979}
Gasser, T. and M{\"u}ller, H.-G. (1979), \enquote{Kernel estimation of
  regression functions,} in \textit{Smoothing Techniques for Curve Estimation},
  eds. Gasser, T. and Rosenblatt, M., Berlin: Springer, pp. 23--68.

\bibitem[{Geenens(2014)}]{geenens2014}
Geenens, G. (2014), \enquote{Probit Transformation for Kernel Density
  Estimation on the Unit Interval,} \textit{J. Amer. Statist. Assoc.}, 109,
  346--358.

\bibitem[{Hall et~al.(1992)Hall, Marron, and Park}]{hall1992}
Hall, P., Marron, J.~S., and Park, B.~U. (1992), \enquote{Smoothed
  cross-validation,} \textit{Probab. Theory Related Fields}, 92, 1--20.

\bibitem[{Hall and Tajvidi(2000)}]{hall2000}
Hall, P. and Tajvidi, N. (2000), \enquote{Distribution and dependence-function
  estimation for bivariate extreme-value distributions,} \textit{Bernoulli}, 6,
  835--844.

\bibitem[{Hill(1975)}]{hill1975}
Hill, B.~M. (1975), \enquote{A simple general approach to inference about the
  tail of a distribution,} \textit{Ann. Statist.}, 3, 1163--1174.

\bibitem[{Hosking(1985)}]{hosking1985b}
Hosking, J. R.~M. (1985), \enquote{Algorithm {AS} 215: Maximum-Likelihood
  Estimation of the Parameters of the Generalized Extreme-Value Distribution,}
  \textit{J. Roy. Stat. Soc. Ser. C}, 34, pp. 301--310.

\bibitem[{Hosking et~al.(1985)Hosking, Wallis, and Wood}]{hosking1985}
Hosking, J. R.~M., Wallis, J.~R., and Wood, E.~F. (1985), \enquote{Estimation
  of the generalized extreme-value distribution by the method of
  probability-weighted moments,} \textit{Technometrics}, 27, 251--261.

\bibitem[{Huang(1992)}]{huang1992}
Huang, X. (1992), \enquote{Statistics of Bivariate Extreme Values,} {Ph.D.}
  thesis, Erasmus University.

\bibitem[{H\"{u}sler and Reiss(1989)}]{husler1989}
H\"{u}sler, J. and Reiss, R.-D. (1989), \enquote{Maxima of normal random
  vectors: between independence and complete dependence,} \textit{Statist.
  Probab. Lett.}, 7, 283--286.

\bibitem[{Jin and Shao(1999)}]{jin1999}
Jin, Z. and Shao, Y. (1999), \enquote{On kernel estimation of a multivariate
  distribution function,} \textit{Statist. Probab. Lett.}, 41, 163--168.

\bibitem[{Joe(1990)}]{joe1990}
Joe, H. (1990), \enquote{Families of min-stable multivariate exponential and
  multivariate extreme value distributions,} \textit{Statist. Probab. Lett.},
  9, 75--81.

\bibitem[{Jones et~al.(2009)Jones, Wang, and Fawcett}]{jones2009}
Jones, D.~A., Wang, W., and Fawcett, R. (2009), \enquote{High-quality spatial
  climate data-sets for Australia,} \textit{Aust. Meteorol. Oceanogr. J.}, 58,
  233--1026.

\bibitem[{Kotz and Nadarajah(2000)}]{kotz2000}
Kotz, S. and Nadarajah, S. (2000), \textit{Extreme Value Distributions: Theory
  and Applications}, London: Imperial College Press.

\bibitem[{Loftsgaarden and Quesenberry(1965)}]{loftsgaarden1965}
Loftsgaarden, D.~O. and Quesenberry, C.~P. (1965), \enquote{A nonparametric
  estimate of a multivariate density function,} \textit{Ann. Math. Statist.},
  36, 1049--1051.

\bibitem[{Lye et~al.(1993)Lye, Hapuarachchi, and Ryan}]{lye1993}
Lye, L., Hapuarachchi, K., and Ryan, S. (1993), \enquote{Bayes estimation of
  the extreme-value reliability function,} \textit{IEEE Trans. Reliab.}, 42,
  641--644.

\bibitem[{Macleod(1989)}]{macleod1989}
Macleod, A.~J. (1989), \enquote{Remark {AS R76}: A Remark on Algorithm {AS
  215}: Maximum-Likelihood Estimation of the Parameters of the Generalized
  Extreme-Value Distribution,} \textit{J. Roy. Statist. Soc. Ser. C}, 38,
  198--199.

\bibitem[{Magnus and Neudecker(1999)}]{magnus1999}
Magnus, J.~R. and Neudecker, H. (1999), \textit{Matrix Differential Calculus
  with Applications in Statistics and Econometrics: Revised edition},
  Chichester: John Wiley and Sons.

\bibitem[{{Marcon} et~al.(2014){Marcon}, {Padoan}, {Naveau}, and
  {Muliere}}]{marcon2014}
{Marcon}, G., {Padoan}, S.~A., {Naveau}, P., and {Muliere}, P. (2014),
  \enquote{{Multivariate nonparametric estimation of the pickands dependence
  function using {B}ernstein polynomials},} \textit{arXiv:1405.5228}.

\bibitem[{Maritz and Munro(1967)}]{maritz1967}
Maritz, J.~S. and Munro, A.~H. (1967), \enquote{On the use of the generalised
  extreme-value distribution in estimating extreme percentiles,}
  \textit{Biometrics}, 23, 79--103.

\bibitem[{Markovich(2007)}]{markovich2007}
Markovich, N. (2007), \textit{Nonparametric Analysis of Univariate Heavy-tailed
  Data: Research and Practice}, Chichester: John Wiley \& Sons.

\bibitem[{Perkins et~al.(2014)Perkins, Moise, Whetton, and
  Katzfey}]{perkins2014}
Perkins, S.~E., Moise, A., Whetton, P., and Katzfey, J. (2014),
  \enquote{Regional changes of climate extremes over {A}ustralia - a comparison
  of regional dynamical downscaling and global climate model simulations,}
  \textit{Int. J. Climatol.}, 34, 3456--3478.

\bibitem[{Perkins et~al.(2007)Perkins, Pitman, Holbrook, and
  McAneney}]{perkins2007}
Perkins, S.~E., Pitman, A.~J., Holbrook, N.~J., and McAneney, J. (2007),
  \enquote{Evaluation of the {AR4} Climate Models' Simulated Daily Maximum
  Temperature, Minimum Temperature, and Precipitation over {A}ustralia Using
  Probability Density Functions,} \textit{J. Climate}, 20, 4356--4376.

\bibitem[{Perkins et~al.(2013)Perkins, Pitman, and Sisson}]{perkins2013}
Perkins, S.~E., Pitman, A.~J., and Sisson, S.~A. (2013), \enquote{Systematic
  differences in future 20 year temperature extremes in {AR4} model projections
  over {A}ustralia as a function of model skill,} \textit{Int. J. Climatol.},
  33, 1153--1167.

\bibitem[{{Pickands III}(1975)}]{pickands1975}
{Pickands III}, J. (1975), \enquote{Statistical inference using extreme order
  statistics,} \textit{Ann. Statist.}, 3, 119--131.

\bibitem[{Prescott and Walden(1980)}]{prescott1980}
Prescott, P. and Walden, A.~T. (1980), \enquote{Maximum likelihood estimation
  of the parameters of the generalized extreme-value distribution,}
  \textit{Biometrika}, 67, 723--724.

\bibitem[{Rootz{\'e}n et~al.(2017)Rootz{\'e}n, Segers, and
  L.~Wadsworth}]{rootzen2017}
Rootz{\'e}n, H., Segers, J., and L.~Wadsworth, J. (2017), \enquote{Multivariate
  peaks over thresholds models,} \textit{Extremes}.

\bibitem[{Rootz{\'e}n and Tajvidi(2006)}]{rootzen2006}
Rootz{\'e}n, H. and Tajvidi, N. (2006), \enquote{Multivariate generalized
  {P}areto distributions,} \textit{Bernoulli}, 12, 917--930.

\bibitem[{Rudemo(1982)}]{rudemo1982}
Rudemo, M. (1982), \enquote{Empirical choice of histograms and kernel density
  estimators,} \textit{Scand. J. Statist.}, 9, 65--78.

\bibitem[{Sain et~al.(1994)Sain, Baggerly, and Scott}]{sain1994}
Sain, S.~R., Baggerly, K.~A., and Scott, D.~W. (1994),
  \enquote{Cross-validation of multivariate densities,} \textit{J. Amer.
  Statist. Assoc.}, 89, 807--817.

\bibitem[{Scott(2015)}]{scott2015}
Scott, D.~W. (2015), \textit{Multivariate Density Estimation: Theory, Practice,
  and Visualization}, Hoboken, NJ: John Wiley \& Sons, 2nd ed.

\bibitem[{Seneviratne et~al.(2012)Seneviratne, Nicholls, Easterling, Goodess,
  Kanae, Kossin, Luo, Marengo, McInnes, Rahimi, Reichstein, Sorteberg, Vera,
  and Zhang}]{seneviratne2012}
Seneviratne, S., Nicholls, N., Easterling, D., Goodess, C., Kanae, S., Kossin,
  J., Luo, Y., Marengo, J., McInnes, K., Rahimi, M., Reichstein, M., Sorteberg,
  A., Vera, C., and Zhang, X. (2012), \enquote{Changes in Climate Extremes and
  their Impacts on the Natural Physical Environment,} in \textit{IPCC WGI/WGII
  Special Report on Managing the Risks of Extreme Events and Disasters to
  Advance Climate Change Adaptation (SREX)}, eds. Field, C., Barros, V.,
  Stocker, T., Qin, D., Dokken, D., Ebi, K., Mastrandrea, M., Mach, K.,
  Plattner, G.-K., Allen, S., Tignor, M., and Midgley, P., Cambridge, United
  Kingdom and New York, NY, USA: Cambridge University Press, pp. 190--230.

\bibitem[{Sheather and Jones(1991)}]{sheather1991}
Sheather, S.~J. and Jones, M.~C. (1991), \enquote{A reliable data-based
  bandwidth selection method for kernel density estimation,} \textit{J. Roy.
  Statist. Soc. Ser. B}, 53, 683--690.

\bibitem[{Sillmann et~al.(2013{\natexlab{a}})Sillmann, Kharin, Zhang, Zwiers,
  and Bronaugh}]{sillmann2013a}
Sillmann, J., Kharin, V.~V., Zhang, X., Zwiers, F.~W., and Bronaugh, D.
  (2013{\natexlab{a}}), \enquote{Climate extremes indices in the {CMIP5}
  multimodel ensemble: {P}art 1. {M}odel evaluation in the present climate,}
  \textit{J. Geophys. Res.-Atmos.}, 118, 1716--1733.

\bibitem[{Sillmann et~al.(2013{\natexlab{b}})Sillmann, Kharin, Zwiers, Zhang,
  and Bronaugh}]{sillmann2013b}
Sillmann, J., Kharin, V.~V., Zwiers, F.~W., Zhang, X., and Bronaugh, D.
  (2013{\natexlab{b}}), \enquote{Climate extremes indices in the {CMIP5}
  multimodel ensemble: {P}art 2. {F}uture climate projections,} \textit{J.
  Geophys. Res.-Atmos.}, 118, 2473--2493.

\bibitem[{Silverman(1986)}]{silverman1986}
Silverman, B.~W. (1986), \textit{Density Estimation for Statistics and Data
  Analysis}, London: Chapman and Hall.

\bibitem[{Smith(1985)}]{smith1985}
Smith, R.~L. (1985), \enquote{Maximum likelihood estimation in a class of
  nonregular cases,} \textit{Biometrika}, 72, 67--90.

\bibitem[{Smith et~al.(1990)Smith, Tawn, and Yuen}]{smith1990b}
Smith, R.~L., Tawn, J.~A., and Yuen, H.~K. (1990), \enquote{Statistics of
  Multivariate Extremes,} \textit{Int. Statist. Rev.}, 58, 47--58.

\bibitem[{Taylor et~al.(2012)Taylor, Stouffer, and Meehl}]{taylor2012}
Taylor, K.~E., Stouffer, R.~J., and Meehl, G.~A. (2012), \enquote{An Overview
  of CMIP5 and the Experiment Design,} \textit{Bulletin of the American
  Meteorological Society}, 93, 485--498.

\bibitem[{Vautard et~al.(2013)Vautard, Gobiet, Jacob, Belda, Colette,
  D\'{e}qu\'{e}, Fern\'{a}ndez, Garc\'{i}a-D\'{i}ez, Goergen, G\"{u}ttler,
  Halenka, Karacostas, Katragkou, Keuler, Kotlarski, Mayer, van Meijgaard,
  Nikulin, Patarci\'{c}, Scinocca, Sobolowski, Suklitsch, Teichmann,
  Warrach-Sagi, Wulfmeyer, and Yiou}]{vautart2013}
Vautard, R., Gobiet, A., Jacob, D., Belda, M., Colette, A., D\'{e}qu\'{e}, M.,
  Fern\'{a}ndez, J., Garc\'{i}a-D\'{i}ez, M., Goergen, K., G\"{u}ttler, I.,
  Halenka, T., Karacostas, T., Katragkou, E., Keuler, K., Kotlarski, S., Mayer,
  S., van Meijgaard, E., Nikulin, G., Patarci\'{c}, M., Scinocca, J.,
  Sobolowski, S., Suklitsch, M., Teichmann, C., Warrach-Sagi, K., Wulfmeyer,
  V., and Yiou, P. (2013), \enquote{The simulation of European heat waves from
  an ensemble of regional climate models within the EURO-CORDEX project,}
  \textit{Climate Dynamics}, 41, 2555--2575.

\bibitem[{Wand et~al.(1991)Wand, Marron, and Ruppert}]{wand1991}
Wand, M., Marron, J., and Ruppert, D. (1991), \enquote{Transformations in
  Density Estimation,} \textit{J. Amer. Statist. Assoc.}, 86, 343--353.

\bibitem[{Wand(1992)}]{wand1992}
Wand, M.~P. (1992), \enquote{Error analysis for general multivariate kernel
  estimators,} \textit{J. Nonparametr. Statist.}, 2, 1--15.

\bibitem[{Wand and Jones(1994)}]{wand1994}
Wand, M.~P. and Jones, M.~C. (1994), \enquote{Multivariate plug-in bandwidth
  selection,} \textit{Comput. Statist.}, 9, 97--116.

\bibitem[{Wand and Jones(1995)}]{wand1995}
--- (1995), \textit{Kernel Smoothing}, London: Chapman and Hall.

\end{thebibliography}

\newpage

\section*{Supplementary Material for `Exploratory data analysis 
for moderate extreme values using non-parametric kernel methods'}

\center{B. B\'{e}ranger, 
T. Duong,  
S. E.  Perkins-Kirkpatrick 
and S. A. Sisson
}

\begin{abstract}
\noindent  

This document contains similar results to the univariate analyses in Section 3.1 (i.e. Figures 1 and 2, and Table 1) except that the sample size is $n=500$ and $n=1000$. The main manuscript uses $n=2000$.

\end{abstract}

\appendix
\section{Simulated data - Univariate: tail sample sizes $n=500$ and $1000$ }

%
%

\begin{figure}[h!]
\centering 
\setlength{\tabcolsep}{3pt}
\begin{tabular}{@{}ccc@{}}
Target Fr\'{e}chet & Target  Gumbel & Target GPD \\
\includegraphics[width=0.32\textwidth]{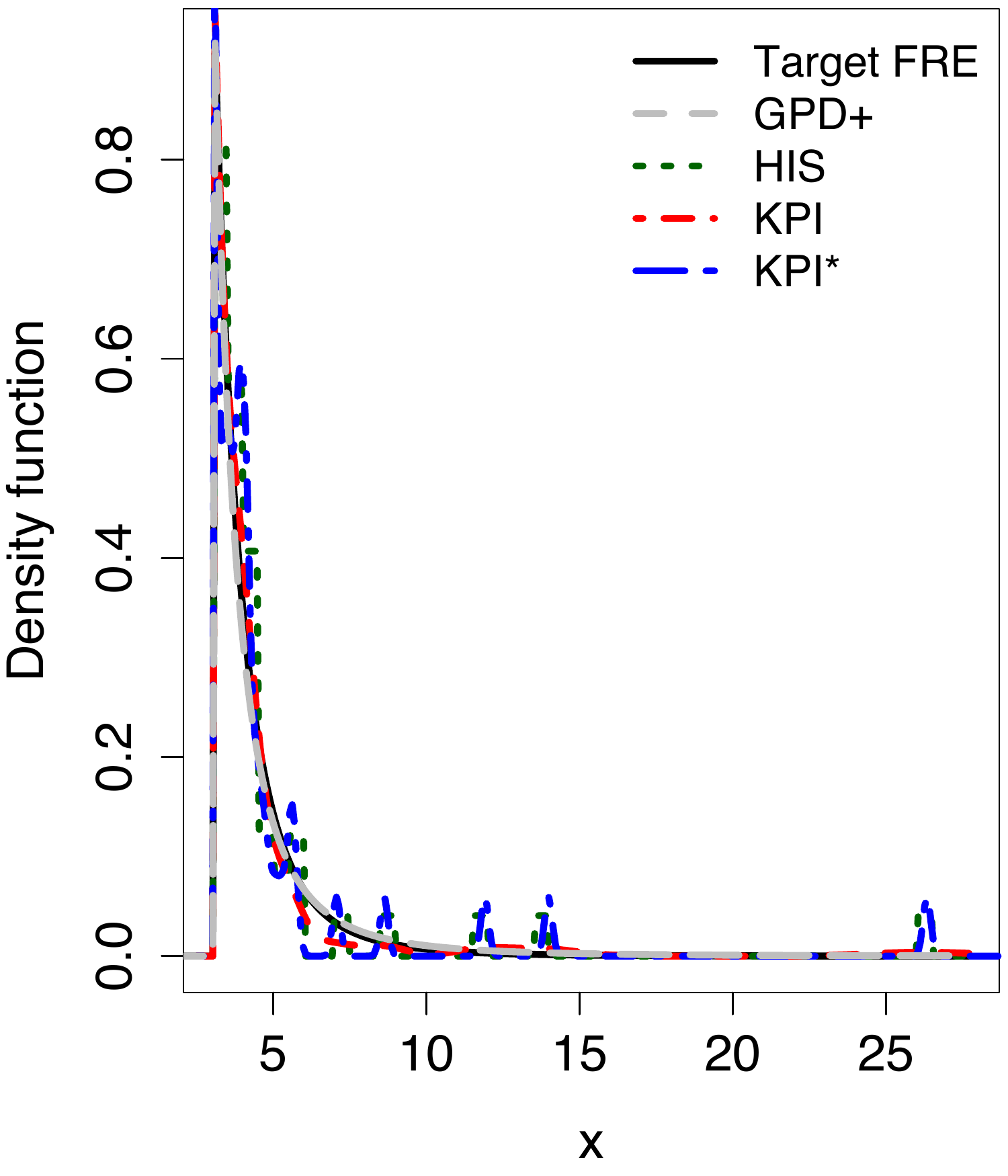} &
\includegraphics[width=0.32\textwidth]{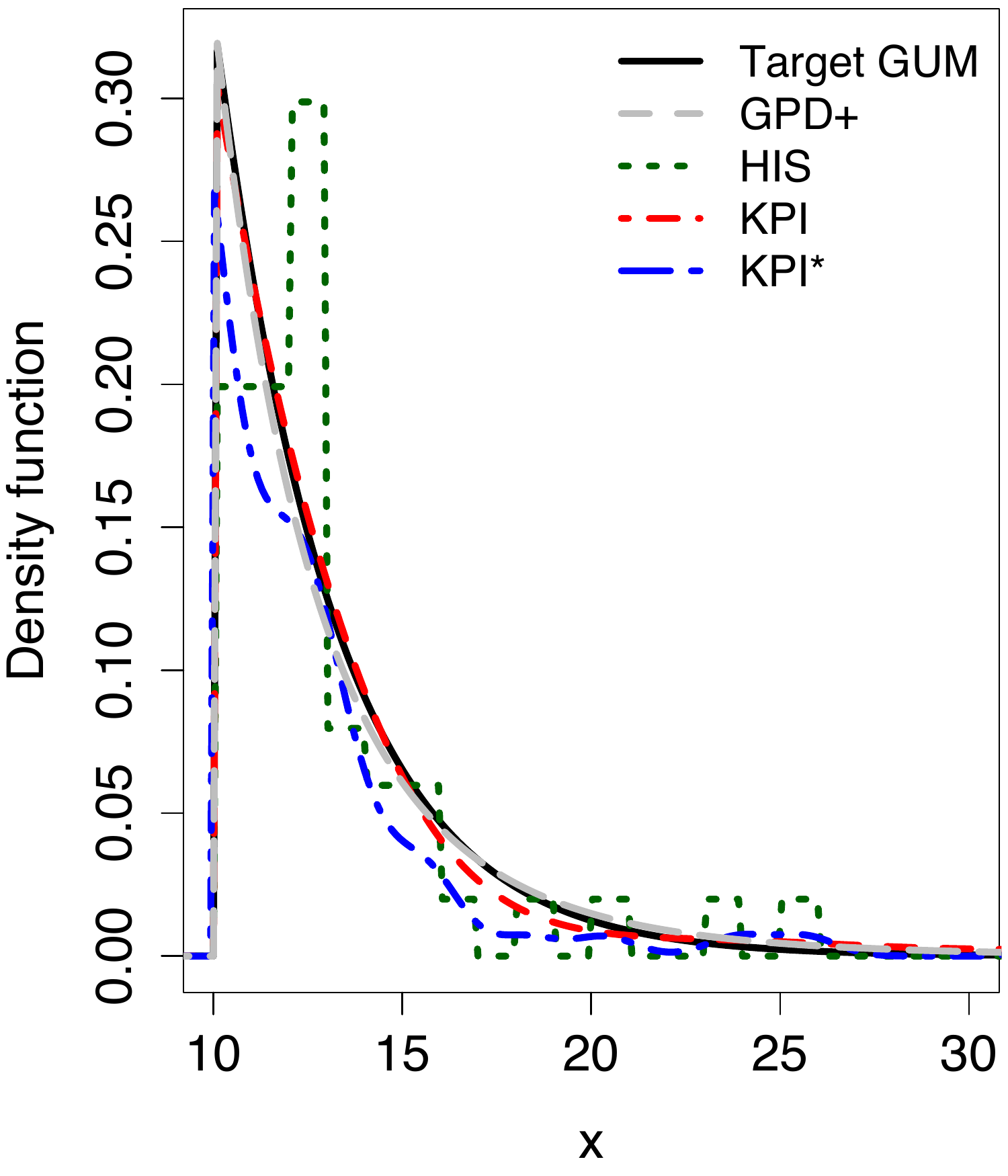} &
\includegraphics[width=0.32\textwidth]{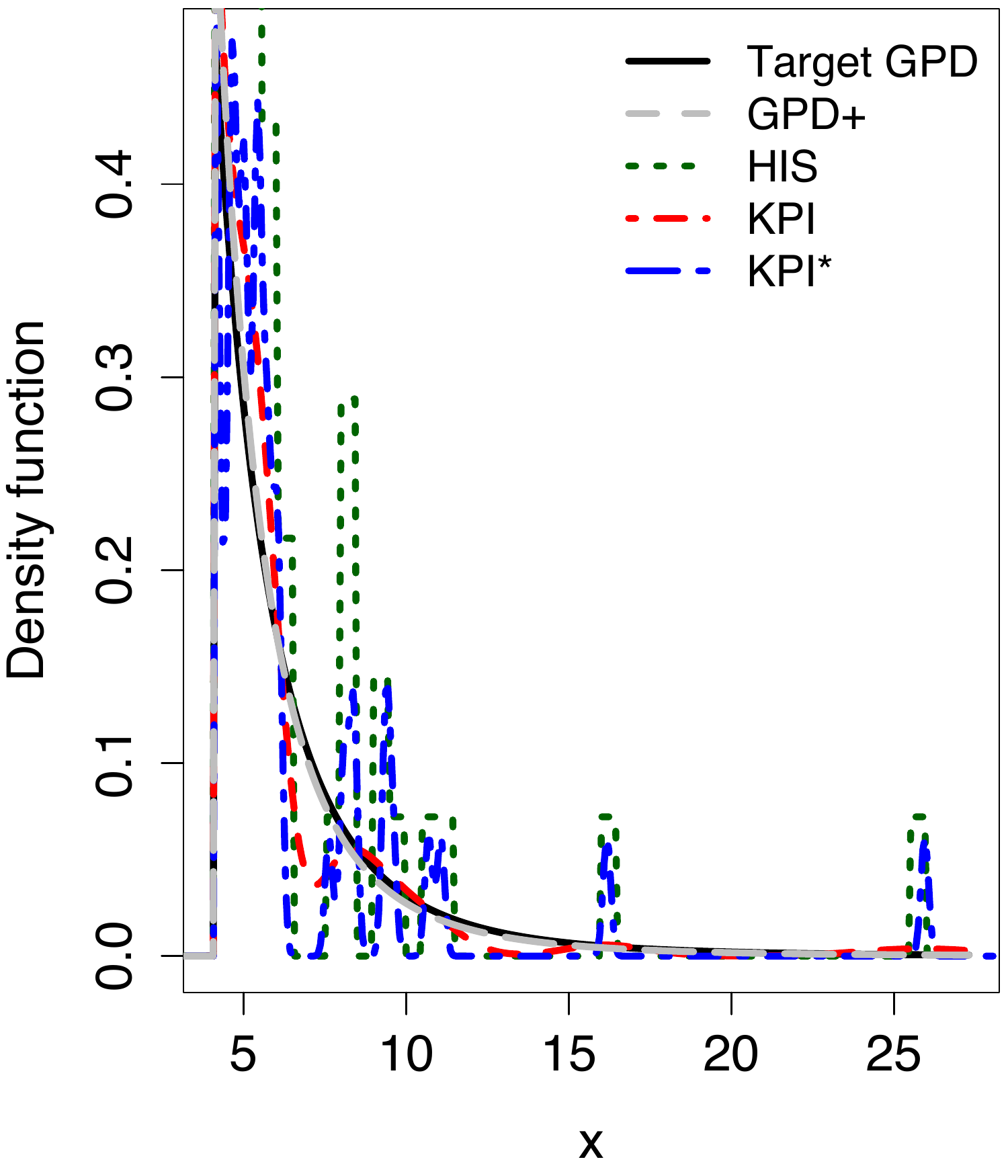} \\
\includegraphics[width=0.32\textwidth]{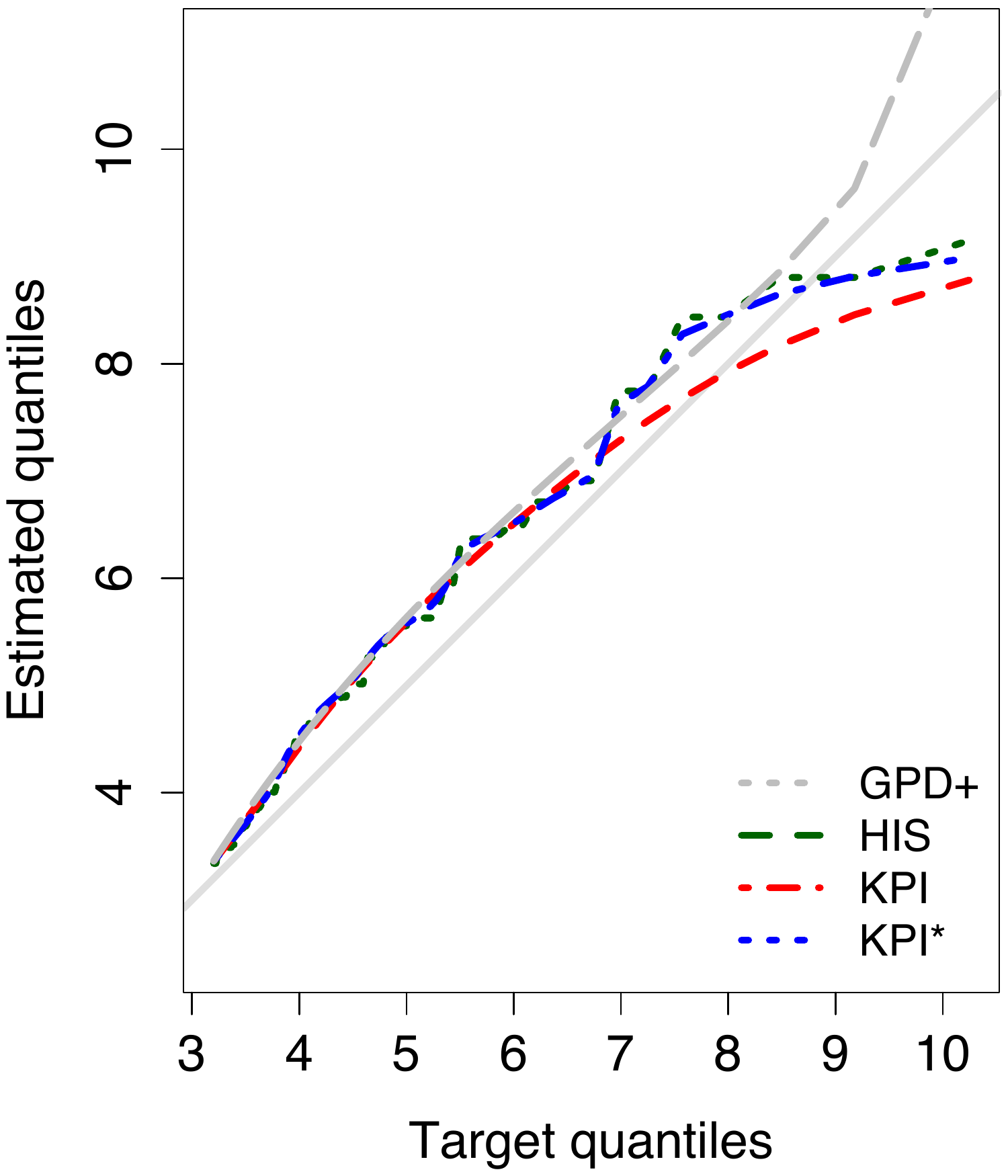} &
\includegraphics[width=0.32\textwidth]{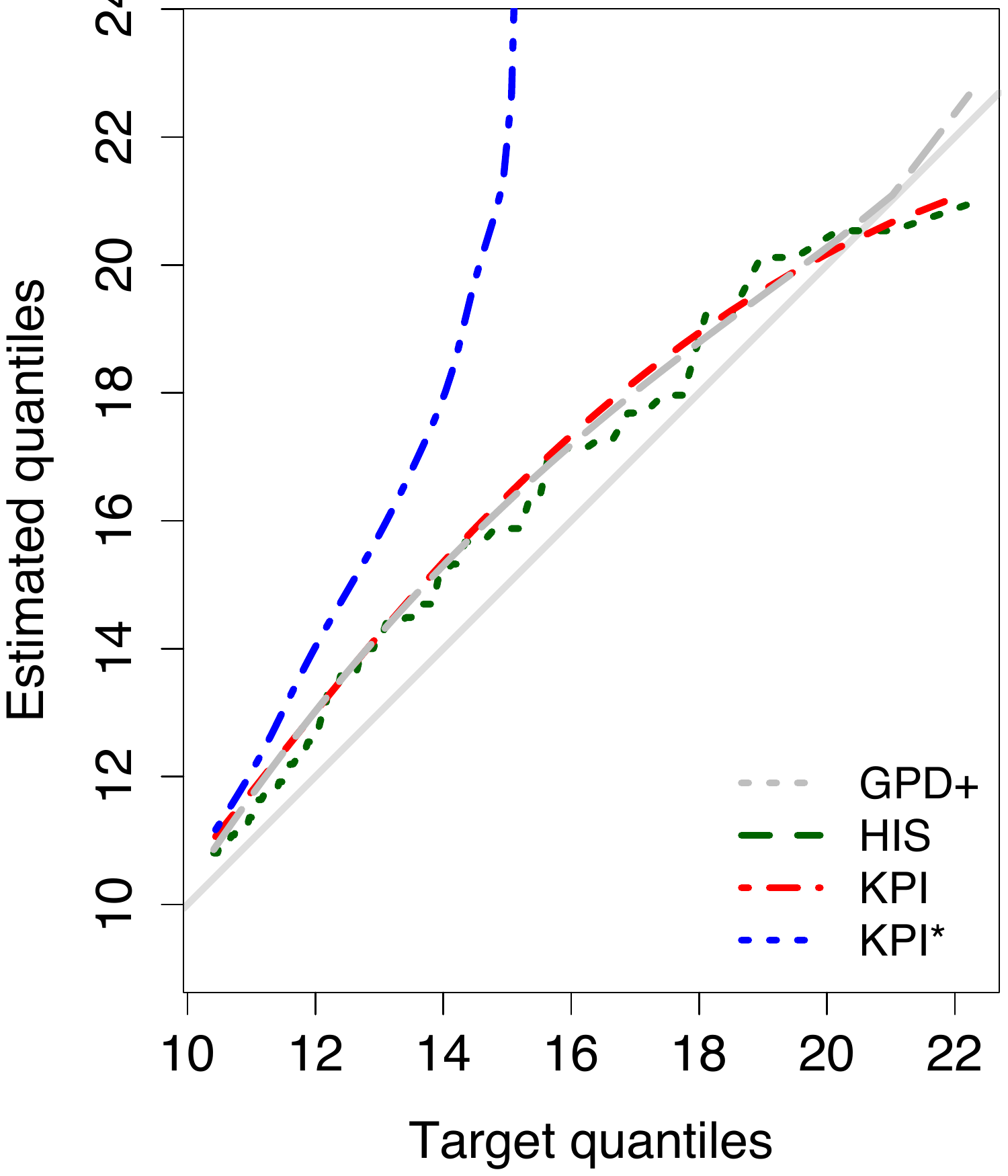} &
\includegraphics[width=0.32\textwidth]{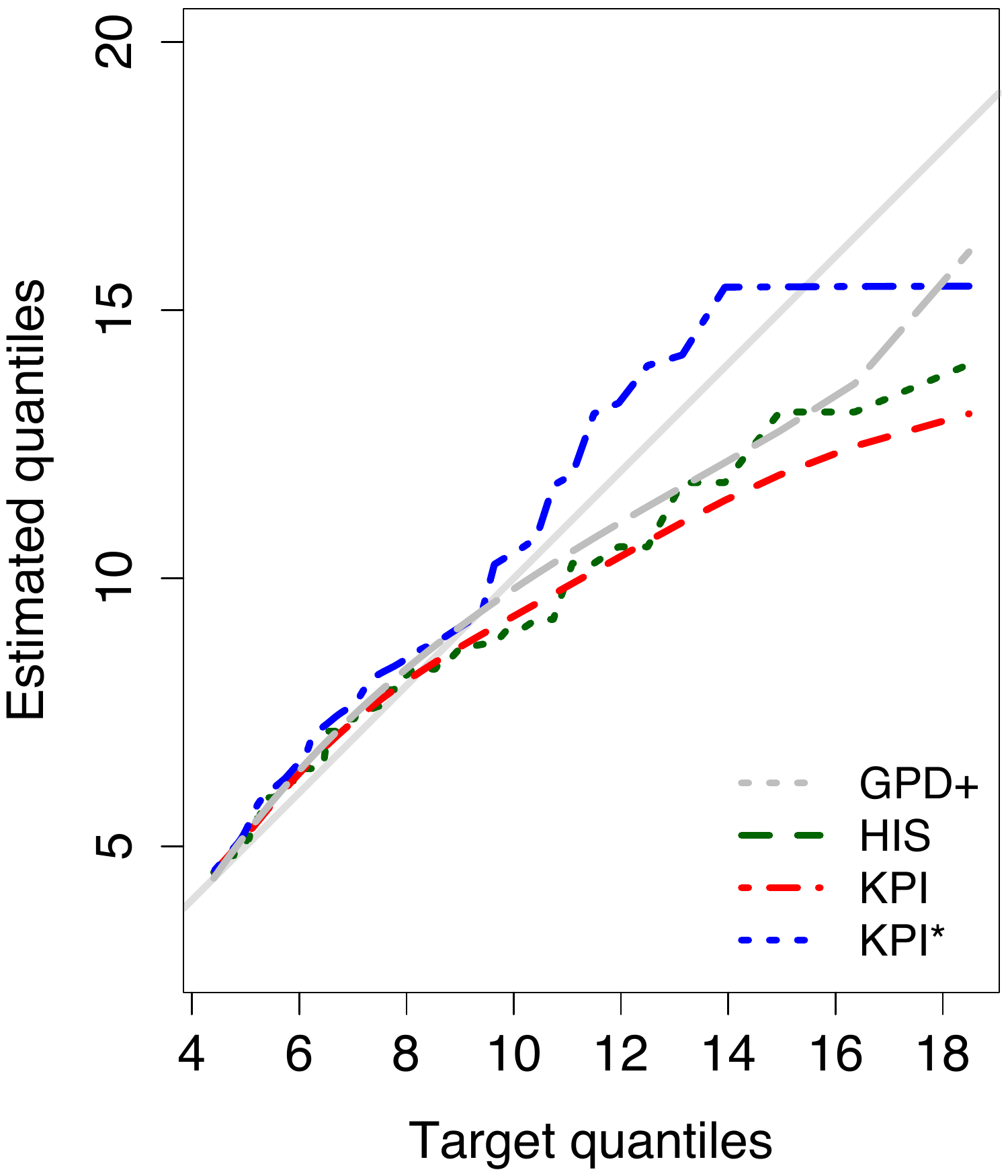}
\end{tabular}
\caption{ \small Generalised Pareto estimator $\check{f}_{X^{[u]}}$ (GPD+; grey long-dashed) and non-parametric estimators of the univariate tail density (top) and of the tail quantiles (bottom) when the target density is Fr\'{e}chet (left), Gumbel (centre) and generalised Pareto (right). Sample size is $n=1000$.
Fr\'{e}chet ($\mu=1$, $\sigma=0.5$, $\xi=0.25$), Gumbel ($\mu=1.5$, $\sigma=3$) and Pareto ($\mu=0$, $\sigma=1$, $\xi=0.25$) target densities are represented by a solid black line. The histogram estimator $\tilde{f}_{X^{[u]}}$ with normal scale binwidth (HIS) is represented by a dotted green line, the
transformed kernel plug-in estimator $\hat{f}_{X^{[u]}}$ (KPI) by a short dashed red line and the standard kernel estimator $\hat{f}^*_{X^{[u]}}$ (KPI*) by a dot-dash blue line. }
\end{figure}

\begin{figure}[h!]
\centering 
\setlength{\tabcolsep}{3pt}
\begin{tabular}{@{}ccc@{}}
Target Fr\'{e}chet & Target Gumbel & Target GPD \\
\includegraphics[width=0.31\textwidth]{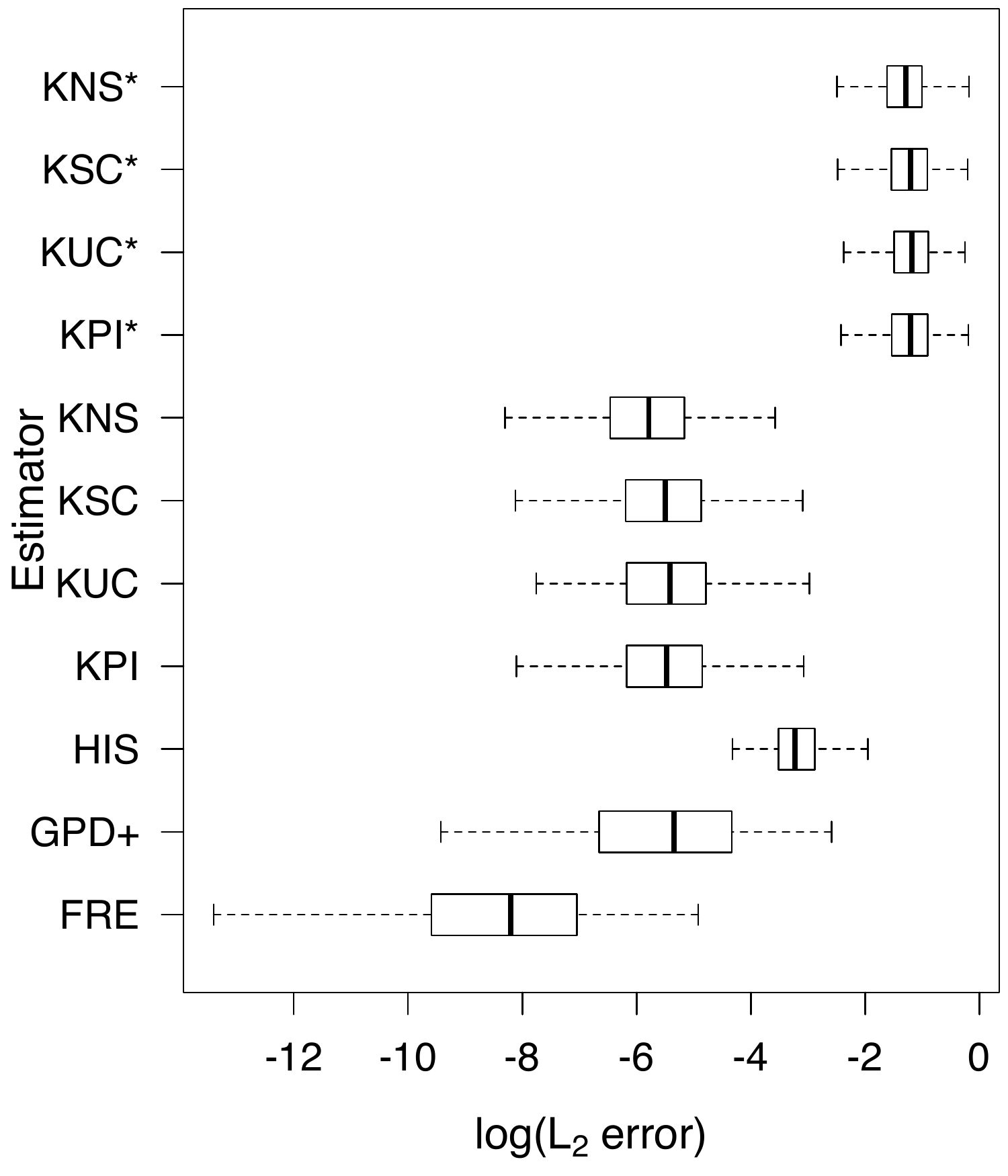} &
\includegraphics[width=0.31\textwidth]{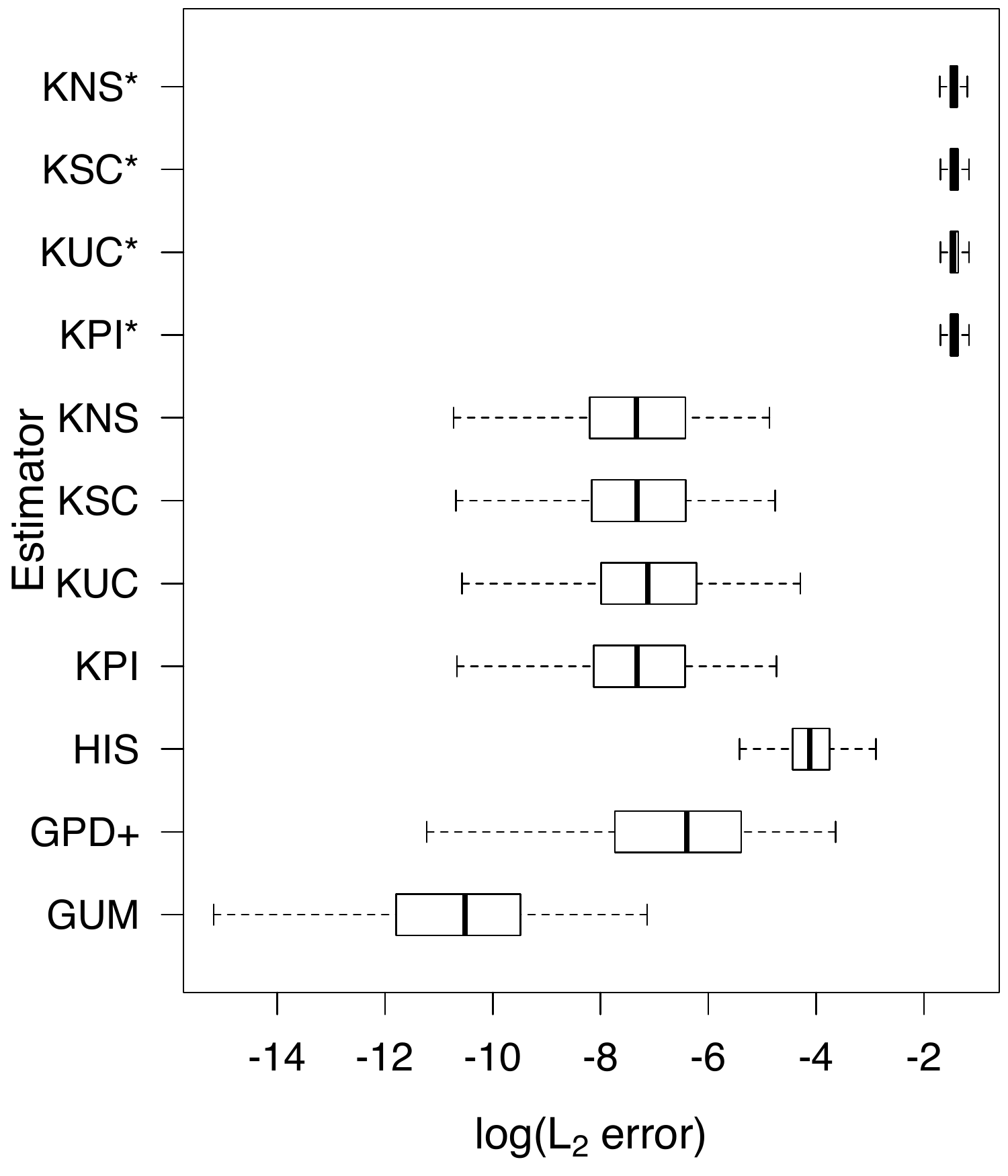} &
\includegraphics[width=0.31\textwidth]{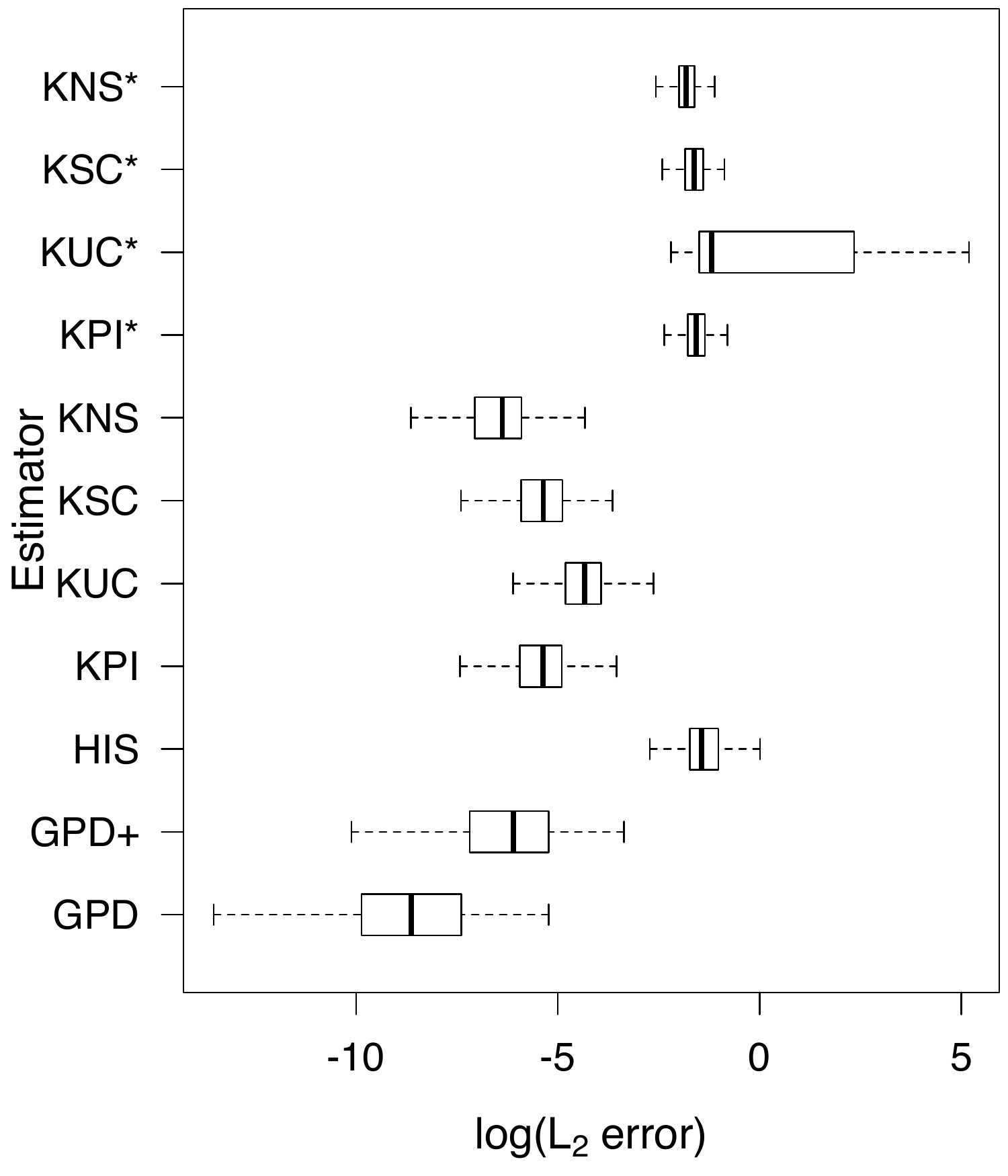} 
\end{tabular}
\caption{ \small Box-plots of the $\log L_2$ errors for the parametric Fr\'{e}chet (FRE), Gumbel (GUM), generalised Pareto (GPD), and histogram (HIS) tail density estimators as well as the generalised Pareto $\check{f}_{X^{[u]}}$ (GPD+). Transformed kernel density estimators  $\hat{f}_{X^{[u]}}$ use the plug-in (KPI), 
unbiased cross validation (KUC), smoothed cross validation (KSC) 
and normal scale kernel (KNS) optimal bandwidth selectors. Standard kernel density estimators $\hat{f}^*_{X^{[u]}}$ are indicated by an asterisk (*). 
True target densities are (left panel) Fr\'{e}chet, (centre) Gumbel and (right) GPD. Box plots are based on 400 replicates of $n=1,000$ observations with a tail sample size of $m = 50$.
}
\end{figure}

\begin{table}[h!]
\begin{center}
\begin{tabular}{@{\extracolsep{4pt}}lcccccccccc@{}}
Target & $\tilde{T}_2$ & $\hat{T}_2$ & $\hat{T}^*_2$ & $\check{T}_2$ \\
\hline
FRE & 0.70 &  \bf{0.85} & 0.84 & \bf{0.85}\\
GUM & {\bf 1.00} & {\bf 1.00} & 0.01 & 0.31 \\
GPD & 0.32 & {\bf 0.95} & 0.01 & 0.83 \\
 \hline
\end{tabular}

\caption{\small Proportion of 400 simulated datasets with sample size $n=1,000$, from each known target distribution (Fr\'echet, Gumbel and GPD) that are correctly identified as coming from each of these distributions by having the smallest tail index value. Bold text indicates the highest proportion for each target model. Nonparametric density estimators are the
 histogram ($\tilde{T}_2$), the transformed kernel ($\hat{T}_2$) 
and the standard kernel ($\hat{T}^*_2$). 
The parametric GPD estimator on tail data is $\check{T}_2$. Tail indices are calculated according to the $L_2$ loss. 
}
\end{center}
\label{tab06}
\end{table}

%
%

\begin{figure}[h!]
\centering 
\setlength{\tabcolsep}{3pt}
\begin{tabular}{@{}ccc@{}}
Target Fr\'{e}chet & Target  Gumbel & Target GPD \\
\includegraphics[width=0.32\textwidth]{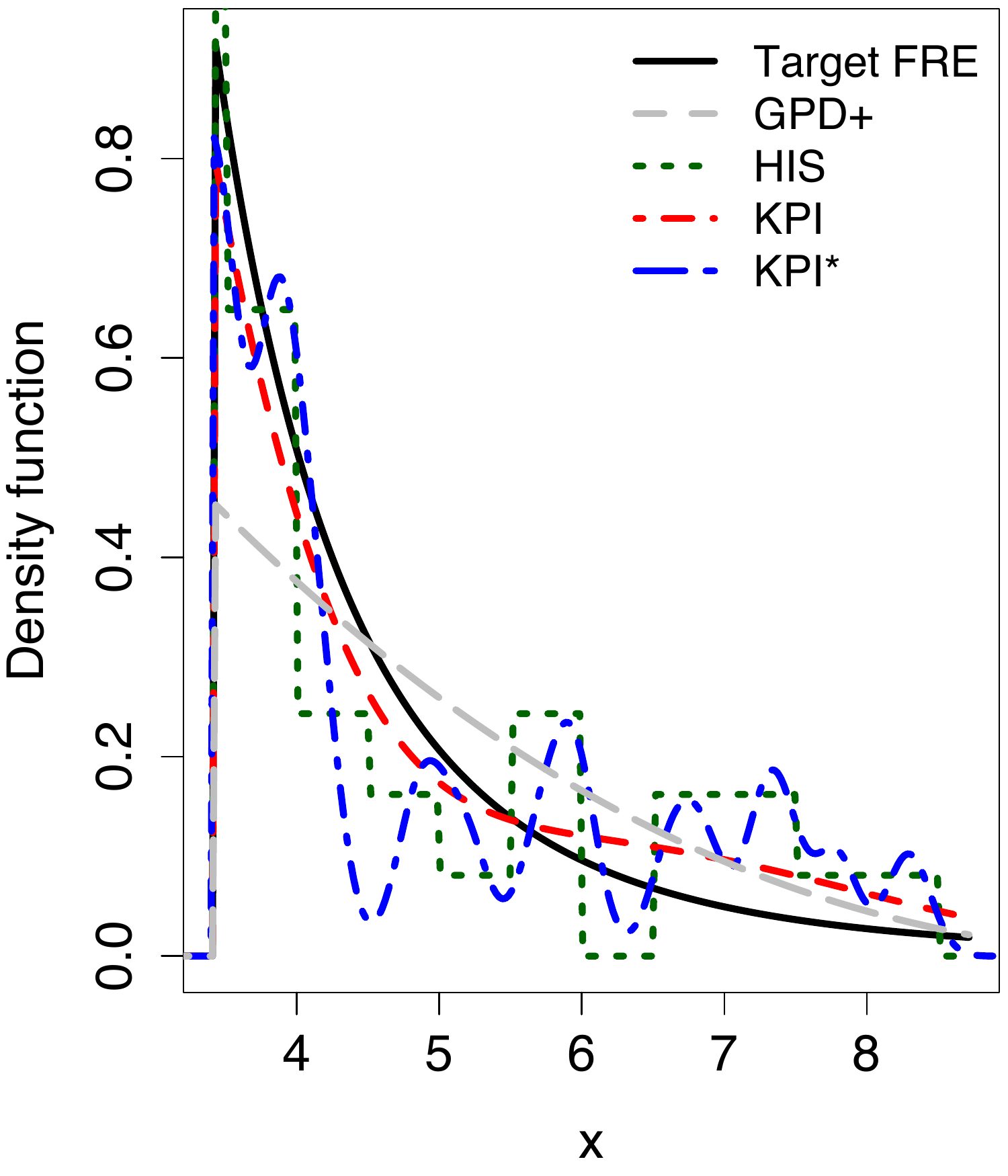} &
\includegraphics[width=0.32\textwidth]{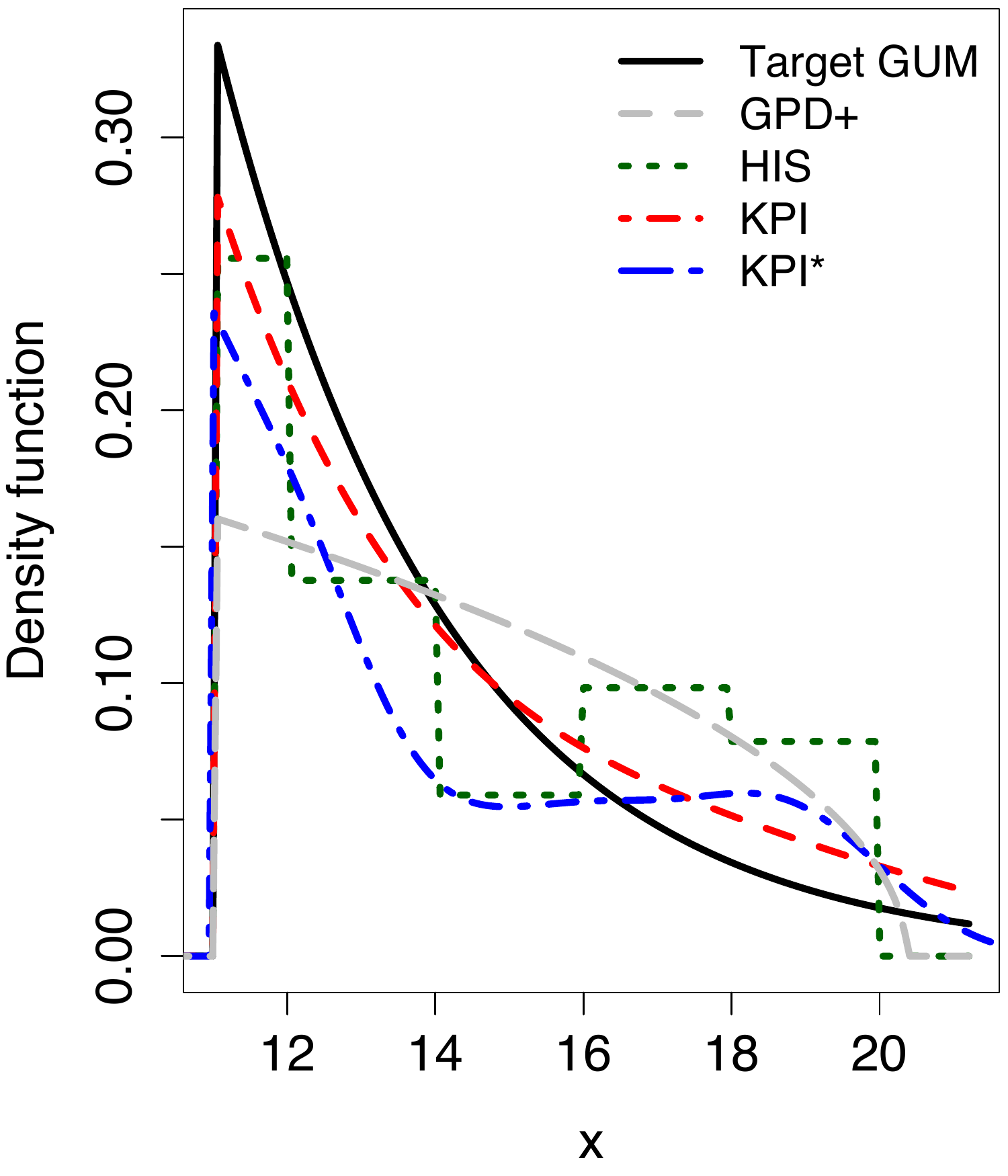} &
\includegraphics[width=0.32\textwidth]{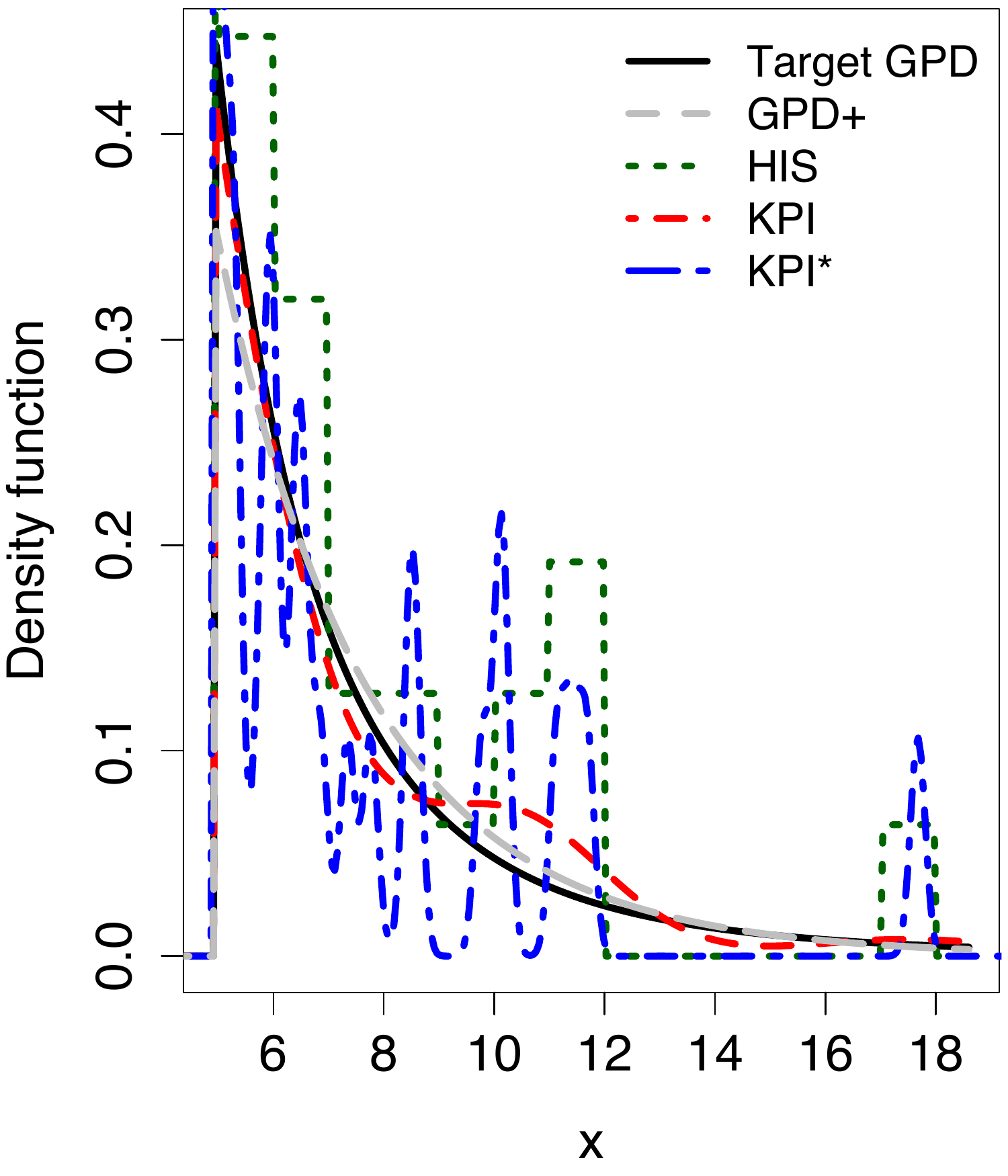} \\
\includegraphics[width=0.32\textwidth]{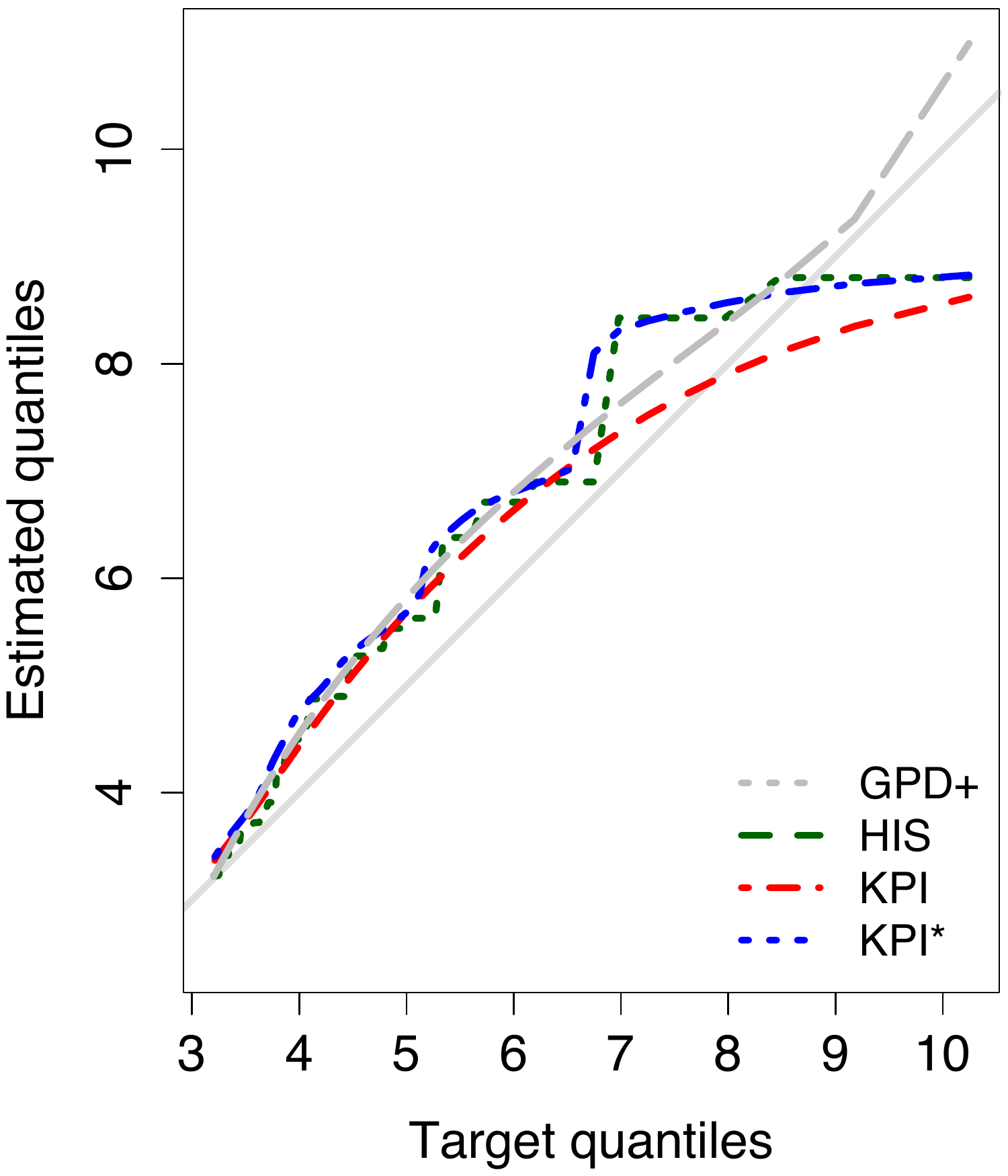} &
\includegraphics[width=0.32\textwidth]{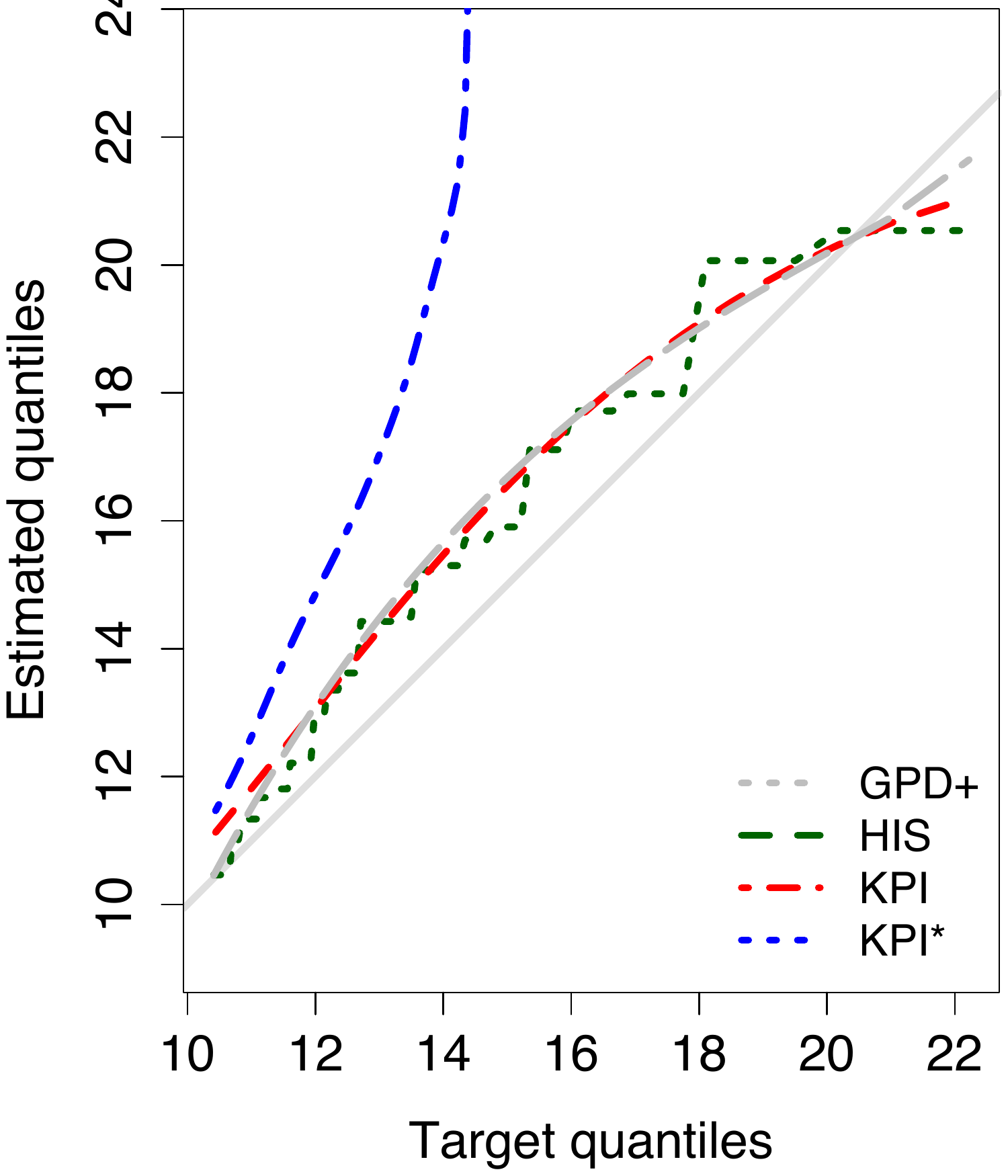} &
\includegraphics[width=0.32\textwidth]{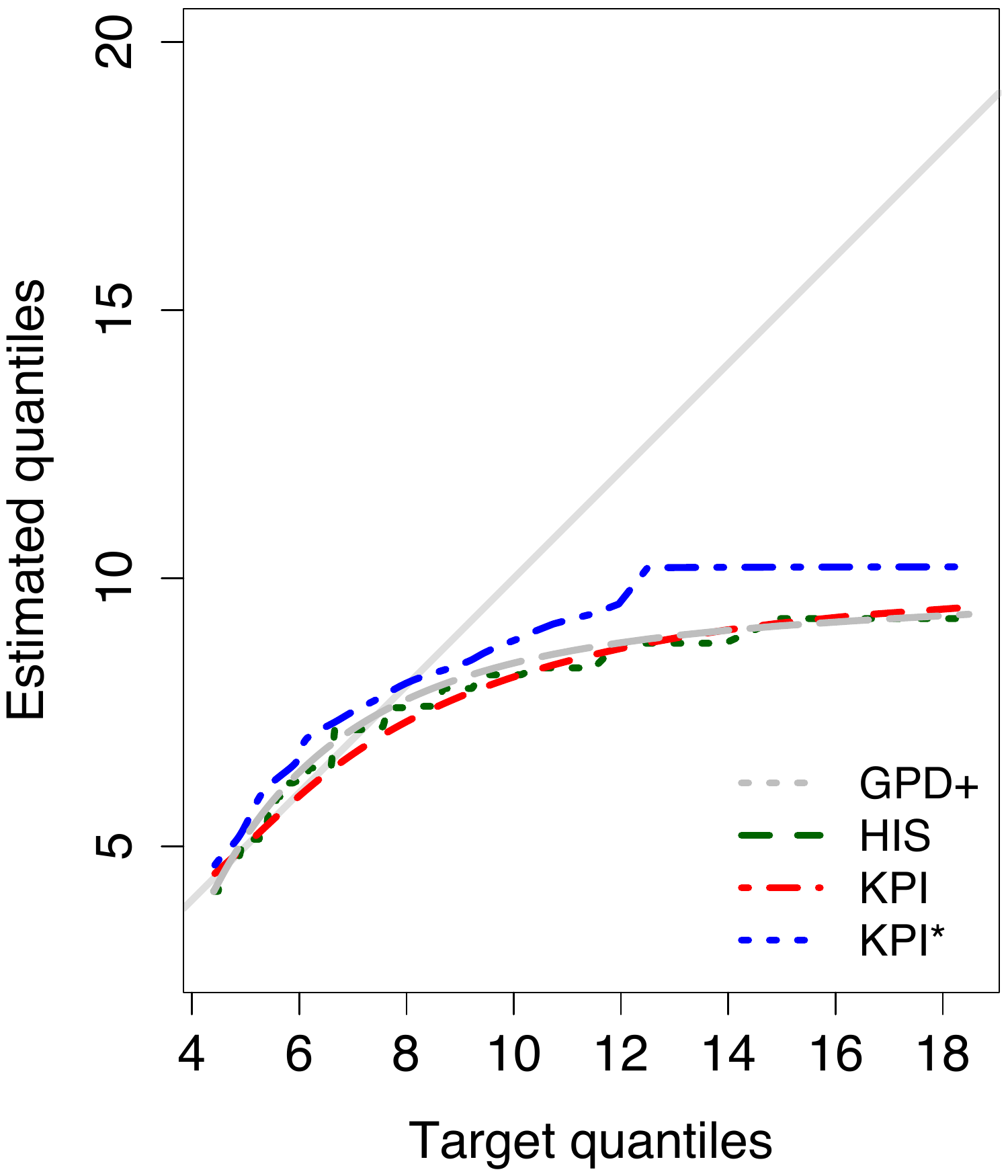}
\end{tabular}
\caption{ \small Generalised Pareto estimator $\check{f}_{X^{[u]}}$ (GPD+; grey long-dashed) and non-parametric estimators of the univariate tail density (top) and of the tail quantiles (bottom) when the target density is Fr\'{e}chet (left), Gumbel (centre) and generalised Pareto (right). Sample size is $n=500$.
Fr\'{e}chet ($\mu=1$, $\sigma=0.5$, $\xi=0.25$), Gumbel ($\mu=1.5$, $\sigma=3$) and Pareto ($\mu=0$, $\sigma=1$, $\xi=0.25$) target densities are represented by a solid black line. The histogram estimator $\tilde{f}_{X^{[u]}}$ with normal scale binwidth (HIS) is represented by a dotted green line, the
transformed kernel plug-in estimator $\hat{f}_{X^{[u]}}$ (KPI) by a short dashed red line and the standard kernel estimator $\hat{f}^*_{X^{[u]}}$ (KPI*) by a dot-dash blue line. }
\end{figure}

\begin{figure}[h!]
\centering 
\setlength{\tabcolsep}{3pt}
\begin{tabular}{@{}ccc@{}}
Target Fr\'{e}chet & Target Gumbel & Target GPD \\
\includegraphics[width=0.31\textwidth]{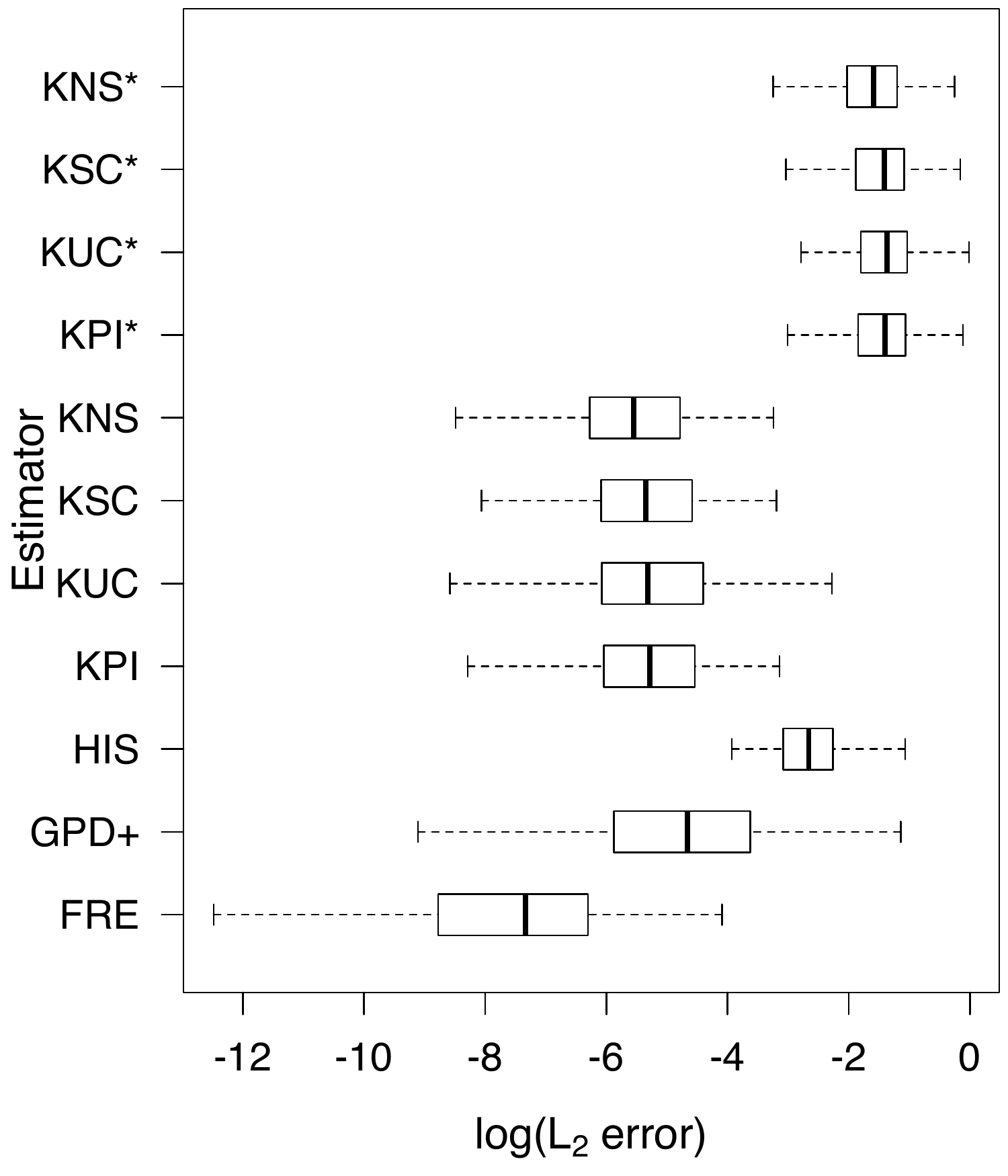} &
\includegraphics[width=0.31\textwidth]{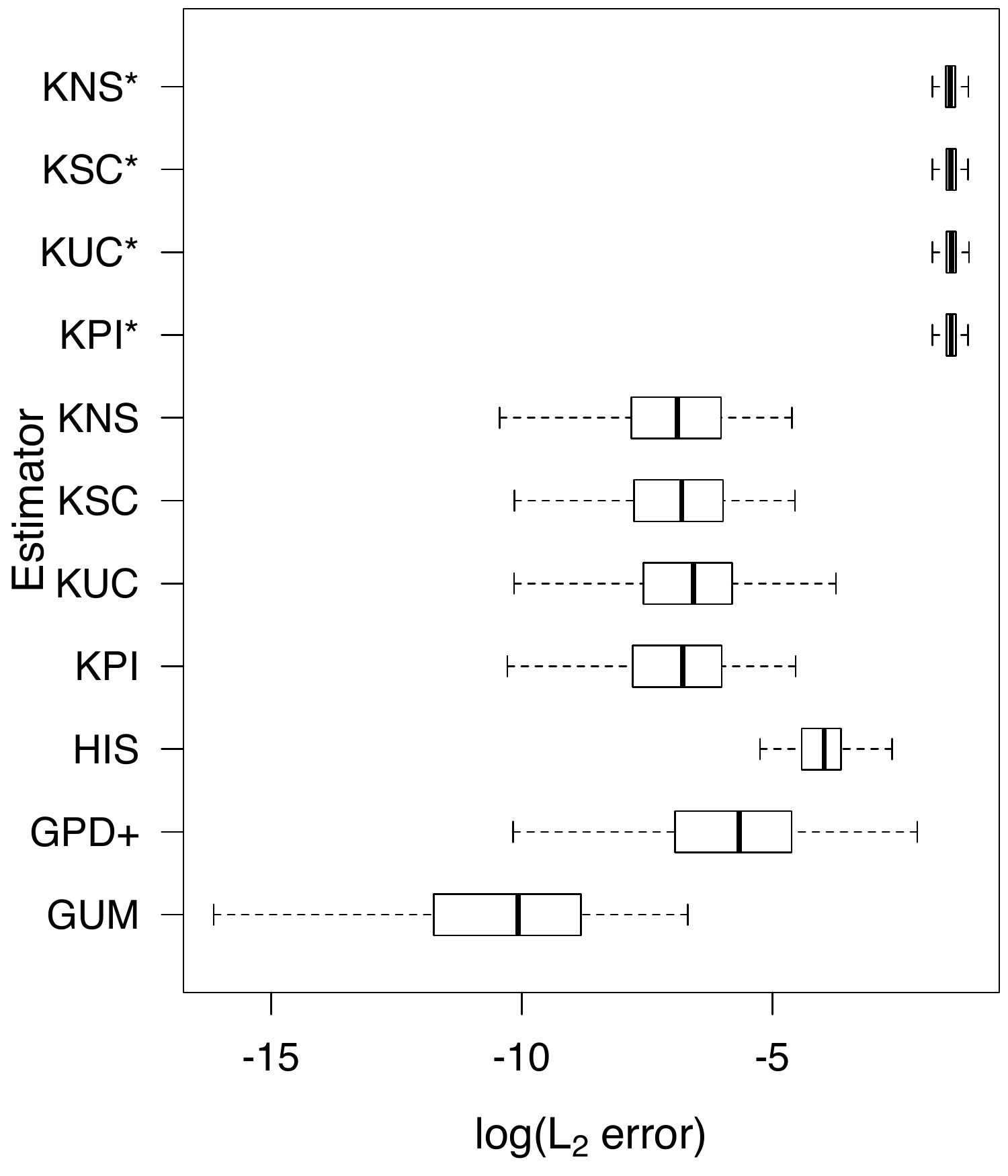} &
\includegraphics[width=0.31\textwidth]{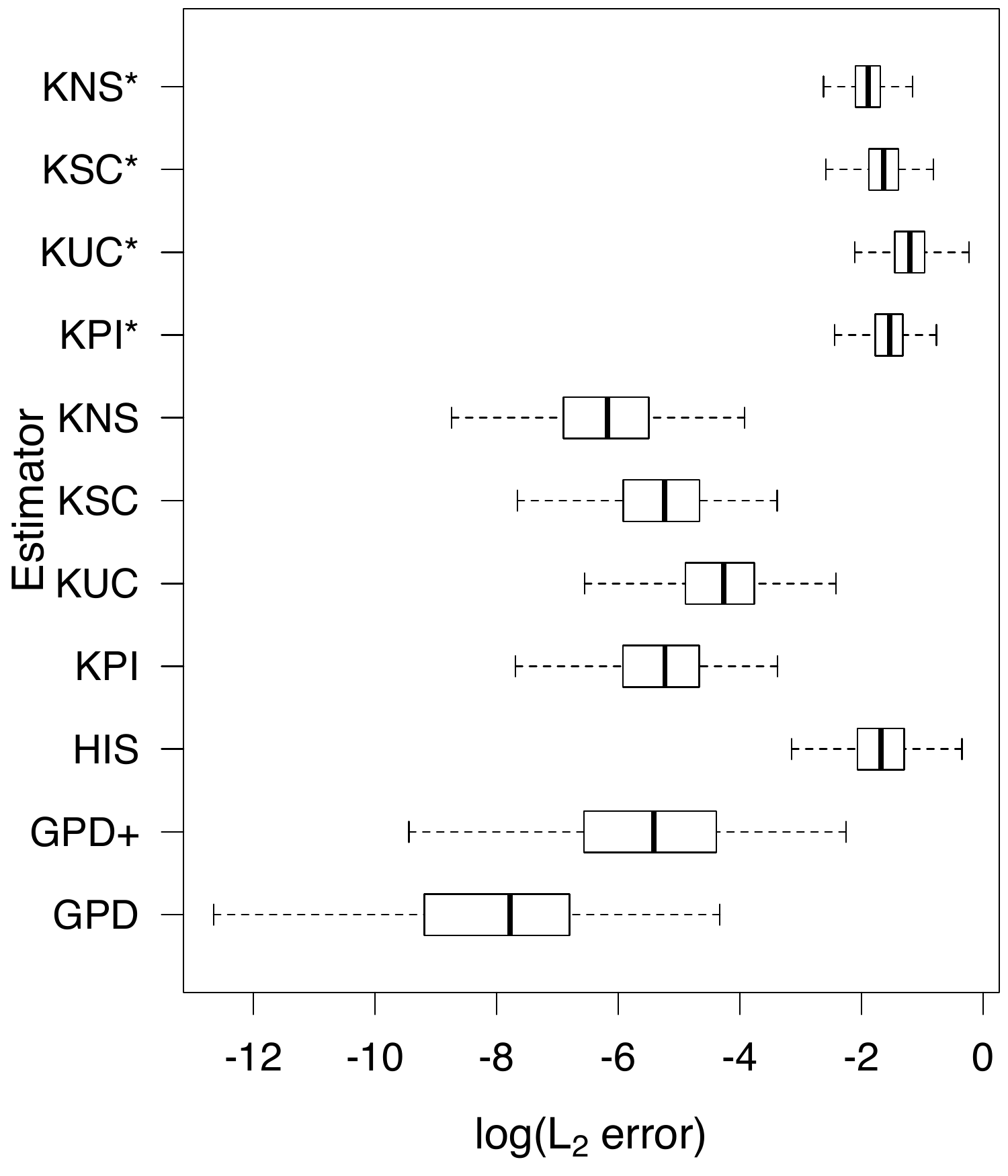}
\end{tabular}
\caption{\small Box-plots of the $\log L_2$ errors for the parametric Fr\'{e}chet (FRE), Gumbel (GUM), generalised Pareto (GPD), and histogram (HIS) tail density estimators as well as the generalised Pareto $\check{f}_{X^{[u]}}$ (GPD+). Transformed kernel density estimators  $\hat{f}_{X^{[u]}}$ use the plug-in (KPI), 
unbiased cross validation (KUC), smoothed cross validation (KSC) 
and normal scale kernel (KNS) optimal bandwidth selectors. Standard kernel density estimators $\hat{f}^*_{X^{[u]}}$ are indicated by an asterisk (*). 
True target densities are (left panel) Fr\'{e}chet, (centre) Gumbel and (right) GPD. Box plots are based on 400 replicates of $n=500$ observations with a tail sample size of $m = 25$.
}
\end{figure}

\begin{table}[h!]
\begin{center}
\begin{tabular}{@{\extracolsep{4pt}}lcccccccccc@{}}
Target & $\tilde{T}_2$ & $\hat{T}_2$ & $\hat{T}^*_2$ & $\check{T}_2$ \\
\hline
FRE & 0.64 &  \bf{0.82} & 0.75 & 0.77 \\
GUM & 0.33 & {\bf 1.00} & 0.08 & 0.28 \\
GPD & 0.48 & {\bf 0.85} & 0.06 & 0.65\\
 \hline
\end{tabular}

\caption{\small Proportion of 400 simulated datasets with sample size $n=500$, from each known target distribution (Fr\'echet, Gumbel and GPD) that are correctly identified as coming from each of these distributions by having the smallest tail index value. Bold text indicates the highest proportion for each target model. Nonparametric density estimators are the
 histogram ($\tilde{T}_2$), the transformed kernel ($\hat{T}_2$) 
and the standard kernel ($\hat{T}^*_2$). The parametric GPD estimator on tail data is $\check{T}_2$. Tail indices are calculated according to the $L_2$ loss.
}
\end{center}
\label{tab07}
\end{table}

\end{document}